\documentclass[11pt,a4paper]{article}

\usepackage[margin=1in]{geometry}

\usepackage{amsmath,amssymb,amsthm}

\usepackage{hyperref}

\hypersetup{colorlinks=true,linkcolor=blue, citecolor=blue}

\usepackage{filecontents}
 
\newtheorem{theorem}{Theorem}
\newtheorem{lemma}[theorem]{Lemma}

\newtheorem{corollary}[theorem]{Corollary}
\newtheorem{remark}[theorem]{Remark}
\newtheorem{definition}[theorem]{Definition}
\newtheorem{condition}[theorem]{Condition}
\newtheorem{conjecture}[theorem]{Conjecture}
\newtheorem{proposition}[theorem]{Proposition}

\newtheorem*{thmmain}{Theorem \ref{thm:main}}
\newtheorem*{lemonestepuniq}{Lemma~\ref{lem:onestepuniq}}
\newtheorem*{lemsteptwouniq}{Proposition~\ref{lem:steptwouniq}}
\newtheorem*{lemexistence}{Lemma~\ref{lem:existence}}
\newtheorem*{lemtwostep}{Proposition~\ref{lem:twostep}}
\newtheorem*{lemtwosteptoprove}{Lemma~\ref{lem:twosteptoprove}}
\newtheorem*{lemmonotone}{Lemma~\ref{lem:mo12no12tone}}
\newtheorem*{lemonesteprecursion}{Proposition~\ref{lem:onesteprecursion}}
    
\newcommand{\bn}{{\mathbb N}}
\newcommand{\deriv}[2]{\frac{\mathrm{d}{#1}}{\mathrm{d}{#2}}}
\newcommand{\pderiv}[2]{\frac{\partial{#1}}{\partial{#2}}}
\newcommand{\pb}{\mathbf{p}}
\newcommand{\rb}{\mathbf{r}} 

\newcommand{\abs}[1]{\left\vert#1\right\vert}

\renewcommand{\geqslant}{\geq}
\renewcommand{\leqslant}{\leq}
\def\Pr{\mathop{\rm Pr}\nolimits}
\newcommand{\thetree}{\mathbb{T}_{d,n}}
\newcommand{\thetreetwo}{\mathbb{T}_{2,n}}
\newcommand{\theleaves}{\Lambda_{\thetree}}
\newcommand{\theleavestwo}{\Lambda_{\thetreetwo}}
\newcommand{\theroot}{v_{d,n}}
\newcommand{\theroottwo}{v_{2,n}}
\newcommand{\child}[2]{#1^{(#2)}}
\newcommand{\theratio}[4]{\gamma(#1,#2,#3,#4)}

\newcommand{\maxset}[1]{\mathcal{M}_{#1}}

\begin{document}
\title{Uniqueness for the 3-State Antiferromagnetic Potts Model on the Tree\thanks{
The research leading to these results has received funding from the European Research Council under the European Union's Seventh Framework Programme (FP7/2007-2013) ERC grant agreement no.\ 334828. The paper reflects only the authors' views and not the views of the ERC or the European Commission. The European Union is not liable for any use that may be made of the information contained therein. Department of Computer Science, University of Oxford, Wolfson Building, Parks Road, Oxford, OX1~3QD, UK.}
}
\author{Andreas Galanis, Leslie Ann Goldberg, and Kuan Yang}

\date{July 26, 2018}
 
\maketitle

\begin{abstract}   
The antiferromagnetic $q$-state Potts model is perhaps the most  canonical model for which the uniqueness threshold on the tree is not yet understood, largely because of the absence of monotonicities.   Jonasson  established the uniqueness threshold in the zero-temperature case, which corresponds to the $q$-colourings model. In the permissive case (where the temperature is positive), the Potts model has an extra parameter $\beta\in (0,1)$, which  makes the task of analysing the uniqueness threshold even harder and much less is known.

In this paper, we focus on the case $q=3$ and give a detailed analysis of the Potts model on the tree by refining Jonasson's approach. In particular, we establish the uniqueness threshold on the $d$-ary tree for all values of $d\geq 2$. When $d\geq3$, we show that the 3-state antiferromagnetic Potts model has uniqueness for all $\beta\geq 1-3/(d+1)$. The case $d=2$ is critical since it relates to  the 3-colourings model on the binary tree ($\beta=0$), which has non-uniqueness. Nevertheless, we show that the Potts model has uniqueness for all $\beta\in (0,1)$ on the binary tree. Both of these results are tight since it is known that uniqueness does not hold in the complementary regime.

Our proof technique gives for general $q>3$ an analytical condition  for proving uniqueness based on the two-step recursion on the tree, which we conjecture to be sufficient to establish the uniqueness threshold for all non-critical cases ($q\neq d+1$). 
\end{abstract} 

\section{Introduction}

The $q$-state Potts model is a fundamental spin system from statistical physics that has  been thoroughly studied in probability and computer science. The model has two parameters $q$ and $\beta$, where $q\geq 3$ is the number of the states, and $\beta>0$ is a parameter which corresponds to the temperature of the system\footnote{Often, in the literature, $\beta$ is taken to be the 
\emph{inverse tempertature}. Since we don't need the physical details here, we simplify the notation
by taking $\beta$ to be $e$ to the inverse temperture.}. 
The set of states is given by   $[q]=\{1,\hdots,q\}$ and we will usually refer to them as \emph{colours}. The case $q=2$ is known as  the Ising model, and the Potts model is the generalisation of the Ising model to multiple states. When $\beta=0$, the Potts model is known as the $q$-colourings model.

A \emph{configuration} of the Potts model on a finite graph $G=(V,E)$ is an assignment $\sigma: V\rightarrow [q]$. The \emph{weight} of the configuration $\sigma$ is given by $w_G(\sigma)=\beta^{m(\sigma)}$, where $m(\sigma)$ denotes the number of monochromatic edges in $G$ under the assignment $\sigma$. The Gibbs distribution of the model, denoted by $\Pr_G[\cdot]$, is the probability distribution on the set of all configurations, where 
the probability mass of each configuration~$\sigma$ is
proportional to its weight~$w_G(\sigma)$.  Thus, for any $\sigma:V\rightarrow[q]$ it holds that
\[\Pr_G[\sigma] = w_G(\sigma)/Z_G,\]
where $Z_G = \sum_{\sigma: V\to[q]} w_G(\sigma)$ is the so-called \emph{partition function}. Note that in the case $\beta=0$ the Gibbs distribution becomes the uniform distribution on the set of proper $q$-colourings of $G$. The Potts model is said to be \emph{ferromagnetic} if $\beta>1$, which means that more likely configurations have
many monochromatic edges. It is said to be \emph{antiferromagnetic} if $\beta<1$, which means that more likely
configurations have fewer monochromatic edges.
This paper is about the antiferromagnetic case.

For spin systems like the Ising model and the Potts model, one of the most well-studied subjects in statistical physics is the so-called uniqueness phase transition on lattice graphs, such as the grid or the regular tree. Roughly, the uniqueness phase transition on an infinite graph captures whether boundary configurations can exert non-vanishing influence on far-away vertices.  In slightly more detail, for a vertex $v$ and an integer $n$, fix an arbitrary configuration on the vertices that are at distance 
at least~$n$ from $v$. Does the influence on the state of $v$ coming from the boundary configuration vanish when $n\rightarrow \infty$? If yes, the model has \emph{uniqueness},  and it has  \emph{non-uniqueness} otherwise.\footnote{\label{eq:rferf3g555}The terminology comes from the theory of Gibbs measures, where the interest is in examining whether there is a unique infinite-volume measure  whose marginals on finite regions is given by the Gibbs distribution (it can be shown that an infinite-volume measure always exists). See \cite{Geo88,friedlivelenik2017} for a thorough  exposition of the theory. The two formulations of uniqueness/non-uniqueness that we have described, i.e., examining infinite-volume measures and examining the limit of marginals in growing finite regions, turn out to be   equivalent.} (See Definition~\ref{def:uniqueness} for a precise formulation in the case of the tree.) Note that uniqueness is a strong property,  which guarantees that the effect of fixing an arbitrary boundary configuration eventually dies out. As an example, for the antiferromagnetic Ising model on the   $d$-ary tree it is well-known that uniqueness holds iff $\beta\geq \tfrac{d-1}{d+1}$; the value $\tfrac{d-1}{d+1}$ is a point of a phase transition and is also known as the \emph{uniqueness threshold} because it is the point at which the uniqueness phase transition occurs.

The uniqueness phase transition plays a prominent role in connecting the efficiency of algorithms for sampling from the Gibbs distribution to   the properties of the Gibbs distribution itself. One of the first examples of such a connection is in the analysis of the Gibbs sampler Markov chain   for the Ising model on the 2-dimensional lattice, where the uniqueness phase transition marks the critical value of  $\beta$ where the mixing time switches from polynomial to exponential (see \cite{ martinelli2, martinelli1,Thomas}).

From a computational complexity perspective, it is the uniqueness phase transition on  the regular tree which is particularly important. For many 2-state spin models, including the antiferromagnetic Ising model and the hard-core model, it has been proved \cite{Sinclair,SlySun,GSVi,LiLuYin} that the uniqueness phase transition  on the tree coincides with a 
more general computational transition in the complexity of approximating the partition function   or sampling from the Gibbs distribution. In the case of the antiferromagnetic Ising model for example,  the problem of approximating  the partition function on $(d+1)$-regular graphs undergoes a computational transition at the tree uniqueness threshold: it admits a polynomial-time algorithm when $\beta\in (\tfrac{d-1}{d+1},1)$ and it is NP-hard  for $\beta\in(0,\tfrac{d-1}{d+1})$. This connection has been established in full generality for antiferromagnetic 2-state systems.

For antiferromagnetic multi-state systems, the situation is much less clear and, in fact, even understanding the uniqueness phase transition on the tree poses major challenges. One of the key reasons behind these difficulties is that certain monotonicities that hold for two-state systems simply do not hold in the multi-state setting, which therefore necessitates far more elaborate techniques. For analysing the uniqueness threshold on the tree, this difficulty has already been illustrated in the case of the $q$-colourings model,  where Jonasson \cite{Jonasson02}, building upon work of Brightwell and Winkler \cite{BW02}, established via a painstaking method that the model is in uniqueness on the $d$-ary tree iff $q>d+1$. 
The goal of this paper is to extend this analysis to the Potts model (beyond the zero-temperature case).

There are several reasons for focusing on establishing uniqueness on the tree. For the colourings model and the antiferromagnetic Potts model,  it is widely conjectured that the uniqueness phase transition on the $d$-ary tree captures the complexity of  approximating the partition function on graphs with maximum degree $d+1$, 
as is the case  for antiferromagnetic 2-state  models. 
It has been known since the 80s that non-uniqueness holds for the colourings model when $q\leq d+1$ and for the Potts model when $\beta<1-q/(d+1)$, see \cite{Peruggi0}. More recently, it was shown in \cite{GSV} that the problem of approximating the partition function is NP-hard when $q<d+1$ for the colourings model 
and when $\beta<1-q/(d+1)$ for the Potts model (for $q$ even). It is not known however whether efficient algorithms can be designed in the complementary regime; for correlation decay algorithms in particular (see \cite{GK, PY, LYZZ}), it has been difficult to capture the uniqueness threshold in the analysis --- this becomes even harder in the case of the Potts model where uniqueness is not known.  For a more direct algorithmic consequence of uniqueness, it has been demonstrated that, on sparse random graphs,  sampling algorithms for the Gibbs distribution can be designed  by exploiting the underlying tree-like structure  and the decay properties on the tree guaranteed by uniqueness. In particular, in the $G(n,d/n)$ random graph,  Efthymiou \cite{Efthymiou} 
developed a sampling algorithm for $q$-colourings when $q>(1+\epsilon)d$, based on Jonasson's uniqueness result.
Related results on $G(n,d/n)$
appear in  \cite{YZ,EHSV,Sinclair2017,MS22}.
Also, after presenting our main result, we will describe an application on random regular graphs, appearing in \cite{samplingpaper}.

\subsection{Our result}\label{sec:results}
In this paper, we study the uniqueness threshold for  the antiferromagnetic Potts model on the tree.
We establish the uniqueness threshold for $q=3$  for every $d\geq 2$. Our proof technique, which is a refinement of Jonasson's approach, also gives, for general $q>3$, an analytical condition  for proving uniqueness, which we conjecture  to be sufficient  for establishing the uniqueness threshold whenever $q\neq d+1$. As we shall discuss shortly, the case $q=d+1$ is  special, since it incorporates the critical case for the colourings model.  To formally state our result, we will need a few definitions.

Given a graph $G=(V,E)$, a configuration $\sigma:V\rightarrow [q]$, and a subset $U$ of $V$, we use $\sigma(U)$ to denote the restriction of the configuration $\sigma$ to the vertices in $U$. For a vertex $v\in V$ and a colour $c\in[q]$, we denote by $\Pr_G[\sigma(v) = c]$ the probability that $v$ takes the colour $c$ in the Gibbs distribution.  Let $\thetree$ be the $d$-ary tree 
with height $n$ (i.e., every path from the root to a leaf has $n$~edges, and every non-leaf vertex has $d$ children).\footnote{Note that the $d$-ary tree is essentially the same as a regular tree with degree $d+1$;  the only difference is that the root of the $d$-ary tree has degree $d$ while the root of a  $(d+1)$-regular tree has degree $d+1$. Accordingly, the uniqueness phase transition occurs at exactly the same location in both trees.}
Let $\theleaves$  be the set of leaves of $\thetree$ 
 and  let  $\theroot$  be its root. The following definition formalises uniqueness on the $d$-ary tree. (See also~\cite{BW02} for details 
 about how to translate
   Definition~\ref{def:uniqueness} to the Gibbs theory formalisation.) 
 
\begin{definition} \label{def:uniqueness}
The $q$-state Potts model with parameter $\beta$ has  \emph{uniqueness}  on the infinite $d$-ary tree if, for all colours $c\in [q]$,  it holds that
\begin{equation}\label{eq:f4g4553}
\limsup_{n\to\infty} \max_{\tau: \theleaves \to [q]}\abs{\Pr_{ \thetree}[\sigma(\theroot) = c\mid\sigma(\theleaves)=\tau]-\frac{1}{q}}=0.
\end{equation}
It has \emph{non-uniqueness} otherwise.
\end{definition}
Equation~\ref{eq:f4g4553} formalises the fact that the correlation between the root of a $d$-ary tree  and vertices at distance $n$ from the root vanishes as $n\rightarrow \infty$. We are now ready to state our main result.

\begin{theorem}\label{thm:uniqueness}
Let $q=3$. When $d\geq 3$, the $3$-state Potts model on  the $d$-ary tree has uniqueness for all $\beta\in[\tfrac{d-2}{d+1},1)$. When $d=2$,  the $3$-state Potts model on the binary tree has uniqueness for all $\beta\in (0,1)$.
\end{theorem}

Theorem~\ref{thm:uniqueness}  precisely pinpoints the uniqueness threshold for the 3-state Potts model since it is known that the  model is in non-uniqueness in the complementary regime. When $d\geq 3$, non-uniqueness for $\beta<\tfrac{d-2}{d+1}$ follows from the existence of multiple semi-translation-invariant Gibbs measures\footnote{\label{foot:referf}Roughly, in a semi-translation-invariant Gibbs measure,  even-layered vertices have the same marginals and  odd-layered vertices  have the same marginals. By studying the number of fixpoints of a particular recursion, one can establish whether there exist multiple such measures. See \cite[Theorem 2.3 \& Theorem 3.2]{BW02} for details on this connection in the context of the colourings model and \cite[Corollary 7.5]{GSV} in the context of the Potts model. We also remark that such measures on the tree have been studied in the statistical mechanics literature as well, for example Peruggi, di Liberto, and Monroy \cite{Peruggi1,Peruggi2} give a description of the phase diagrams of the models in non-uniqueness. We refer the reader to the book \cite{Rozikov} for a detailed treatment of Gibbs measures on the infinite tree.}. When $d=2$,  the 3-state Potts model for $\beta=0$ corresponds to 
the 3-colouring model, and non-uniqueness holds in this case because of the existence of so-called \emph{frozen} 3-colourings; in these colourings, the configuration on the leaves determines uniquely the colour of the root, see \cite{BW02}.

Interestingly, our result and proof technique for the 3-state Potts model suggests that the only obstruction to uniqueness 
in the 3-colouring model on the binary tree are  the frozen colouring configurations. It is reasonable to  believe that this critical behaviour 
in the colourings model
happens more generally whenever $q=d+1$. 
For comparison, note that
the colourings model has non-uniqueness when
$q<d+1$  (\cite{BW02}, see also   footnote \ref{foot:referf})
and it has uniqueness 
when   $q>d+1$  \cite{Jonasson02}.

This critical behaviour for the colourings model when $q=d+1$  arises in the context of the Potts model as well, and, as we shall see in the next section, it causes complications in the proof of Theorem~\ref{thm:uniqueness}. Nevertheless, we formulate a general condition for all non-critical cases ($q\neq d+1$) which will be sufficient to establish the uniqueness threshold.
We conjecture that the condition holds whenever $q\neq d+1$ (see Conjecture~\ref{conj}). The condition is tailored to the Potts model on a tree, unlike other known sufficient criteria for uniqueness (see for example \cite{Dobrushin,Weitz}). Our condition reduces to single-variable inequalities and can be verified fairly easily for small values of $q,d$. 
Since Theorem~\ref{thm:uniqueness}
includes the critical case $(q,d)=(3,2)$, our proof
of the theorem  necessarily goes a slightly different way (as we explain below),  so
in  Section~\ref{sec:approach}, we give a more detailed outline of our proof approach.

\subsection{Application}

We have already discussed some results in the literature where the uniqueness of spin-models on trees
enables fast algorithms for sampling from these models on bounded-degree graphs and sparse random graphs.
It turns out that Theorem~\ref{thm:uniqueness} can also be used in this way.
In particular, Blanca et al.~have obtained the following theorem.

\begin{theorem} 
[Theorem 8 of  \cite{samplingpaper}]
Let $q\geq 3$, $d\geq2$,   and $\beta\in(0,1)$ be in the uniqueness regime of the $d$-ary tree with $\beta\neq (d+1-q)/(d+1)$. 
Then, there exists a constant $\delta>0$ such that, for all sufficiently large $n$, the following holds   with probability $1-o(1)$ over the choice of a random $(d+1)$-regular 
graph $G=(V,E)$ with $n$ vertices.

There is a polynomial-time algorithm which, 
given the graph~$G$ as input, outputs a random  assignment $\sigma\colon V\rightarrow [q]$ 
from a distribution which is within total variation distance $O(1/n^{\delta})$ from the Gibbs distribution of the Potts model
on~$G$ with parameter~$\beta$.  
\end{theorem}
  
Thus, Theorem~\ref{thm:uniqueness} has the following corollary.

\begin{corollary}  
Let $q=3$.
Suppose either $d=2$ and $\beta \in (0,1)$
or $d\geq 3$ and 
$\beta\in(\tfrac{d-2}{d+1},1)$.
In either case,  there exists a constant $\delta>0$ such that, for all sufficiently large $n$, the following holds   with probability $1-o(1)$ over the choice of a random $(d+1)$-regular graph $G=(V,E)$ with $n$ vertices.

There is a polynomial-time algorithm which, 
given the graph~$G$ as input, outputs a random  assignment $\sigma\colon V\rightarrow [q]$ 
from a distribution which is within total variation distance $O(1/n^{\delta})$ from the Gibbs distribution of the Potts model
on~$G$ with parameter~$\beta$.  
\end{corollary}

We next discuss our approach for proving Theorem~\ref{thm:uniqueness}.

\section{Proof Approach}\label{sec:approach}
In this section, we outline the key steps of our proof approach for proving uniqueness for the antiferromagnetic Potts model on the tree. As mentioned in the Introduction, the  model does not enjoy the monotonicity properties which are present in two-state systems (or the ferromagnetic case)\footnote{All two-state systems are either monotone or antimonotone on the tree, and therefore the root is most sensitive to boundary configurations where all the leaves have the same state. Uniqueness/non-uniqueness is therefore determined by examining whether the marginal at the root under these two extremal configurations coincide. Similarly, for the ferromagnetic Potts model, one can show that the extremal configurations on the leaves are those where all the leaves have the same colour.}, so we have to establish more elaborate criteria to resolve the uniqueness threshold. 

We first review Jonasson's approach for colourings \cite{Jonasson02}. One of the key insights there is to consider the ratio of the probabilities that the root takes two distinct colours and show that this converges to 1 as the height of the tree grows large. Jonasson analysed first a one-step recursion to establish bounds on the marginals of the root and used those to obtain upper bounds on the ratio. Then,  he bootstrapped these bounds by analysing a more complicated two-step recursion and showed that the ratio converges to 1.

Our approach refines Jonasson's approach in the following way; we jump into the two-step recursion and analyse the associated optimisation problem by giving an explicit description of the maximisers for general $q$ and $d$ (see Lemma~\ref{lem:existence}).  It turns out that the maximisers change as the value of the ratio gets closer to 1, so to prove the desired convergence to 1, we need to account for the roughly $q^d$ possibilities for the maximiser. This yields an analytic  condition that can be checked easily for small values of $q,d$ and thus establish uniqueness. In the context of Theorem~\ref{thm:uniqueness} where $q=3$, most of the technical work is to deal analytically with the potentially large values of the arity $d$ of the tree. 

A further complication arises in the case $q=3$ and $d=2$ (and more generally $q=d+1$), since this incorporates the critical behaviour for colourings described in Section~\ref{sec:results}. This manifests itself in our proof by breaking the (global) validity of our uniqueness condition. We therefore have to use an analogue of Jonasson's approach to account for this case by first using the one-step recursion to argue that the ratio gets sufficiently close to 1 and then finishing the argument with the two-step recursion. 

Our proofs are computer-assisted but rigorous --- namely we use the (rigorous) {\tt Resolve} function of Mathematica
to check certain inequalities. We also provide Mathematica code to assist the
reader with tedious-but-straightforward calculations (such as differentiating complicated functions).
The Mathematica code is in 
Section~\ref{sec:code}.

\subsection{Ratio for proving Theorem~\ref{thm:uniqueness}}

For $\beta\in (0,1)$ and $n>0$, define the following ratio.
\begin{equation}
\label{eq:ratio}
\theratio{q}{\beta}{d}{n}
= 
\max_{\stackrel{ \tau: \theleaves \to [q]}{c_1,c_2\in[q]}} \frac
{\Pr_{\thetree}[\sigma(\theroot) = c_1\mid\sigma(\theleaves)=\tau]}
{\Pr_{ \thetree}[\sigma(\theroot) = c_2\mid\sigma(\theleaves)=\tau]}.
\end{equation}
Note that if $\beta > 0$ and $n>0$, then for
every $ \tau\colon \theleaves \to [q]$ and every $c\in[q]$,
$ \Pr_{ \thetree}[\sigma(\theroot) = c\mid\sigma(\theleaves)=\tau]>0$.
So  $\theratio{q}{\beta}{d}{n}$ is well-defined.

Suppose, for fixed $q$, $\beta$ and $d$,
that
 $\lim_{n\to\infty} \gamma(q,\beta,d,n) = 1$.
This implies that the limsup in the uniqueness definition (Definition~\ref{def:uniqueness}) is zero.
Thus,  Theorem~\ref{thm:uniqueness} 
is an immediate consequence of the following theorem.

\newcommand{\statethmmain}{
    If $\beta\in(0,1)$ 
    then  $\lim_{n\to\infty}  \theratio{3}{\beta}{2}{n} = 1$.
    If $d\geq 3$
    and  $1-3/(d+1)\leq \beta < 1$ then
    $\lim_{n\to\infty} \theratio{3}{\beta}{d}{n}
    = 1$.
}
\begin{theorem}\label{thm:main}
    \statethmmain
\end{theorem}

In  Section~\ref{sec:conclusion}  we 
obtain Theorem~\ref{thm:uniqueness}
by proving Theorem~\ref{thm:main}.

\subsection{The two-step recursion}\label{sec:introtwostep}
In this section, we formulate an appropriate recursion on the infinite $d$-ary tree, which will be one of our main tools for tracking the ratio $\gamma(q,\beta,d,n)$.

We denote the set of $q$-dimensional probability vectors by $\triangle$, i.e., 
$$\triangle = \{(p_1, p_2, \ldots, p_q)\colon0 \leq  p_1, p_2, \ldots, p_q \leq 1\, \land\, p_1+p_2+\cdots+p_q = 1\}.$$
Suppose that $c_1$ and $c_2$ are two colours in $[q]$.
We define two functions $g_{c_1,c_2,\beta}$ and $h_{c_1,c_2,\beta}$, indexed by these colours.
The argument of each of these functions
is a tuple $(\child{\pb}{1}, \ldots, \child{\pb}{d})$
where, for each $j\in[d]$,   $\child{\pb}{j} \in \triangle$. 
The functions are defined as follows.
\begin{equation}\label{eq:gh12def}
\begin{aligned}
g_{c_1,c_2,\beta}(\child{\pb}{1},\ldots,\child{\pb}{d})
&:=\prod^d_{k=1}\bigg(1-\frac{(1-\beta) \big(\child{p}{k}_{c_1}-\child{p}{k}_{c_2}\big)}{\beta \child{p}{k}_{c_2}+\sum_{c\neq c_2}\child{p}{k}_{c}}\bigg).\\ 
h_{c_1,c_2,\beta}(\child{\pb}{1}, \ldots, \child{\pb}{d})
&:=1+\frac{(1-\beta)\big(1-g_{c_1,c_2,\beta}(\child{\pb}{1}, \ldots, \child{\pb}{d})\big)}{\beta +\sum_{c\neq c_2}g_{c,c_2,\beta}(\child{\pb}{1}, \ldots, \child{\pb}{d})}.
\end{aligned}
\end{equation}
Note that the functions $g_{c_1,c_2,\beta}$ and $h_{c_1,c_2,\beta}$ are well-defined when $\beta\in (0,1)$ and all of $\child{\pb}{1}, \ldots, \child{\pb}{d}$ have non-negative entries; they are also well-defined when $\beta=0$ and all of $\child{\pb}{1}, \ldots, \child{\pb}{d}$ have  positive entries.

One feature of the functions $g_{c_1,c_2,\beta}$ and $h_{c_1,c_2,\beta}$ which will be important shortly is that they are scale-free. This means that we can multiply each of their arguments by some constant  without changing their value, i.e., for scalars $t_1,\hdots,t_d>0$ it holds that 
\begin{equation}\label{eq:scalefree23}
\begin{aligned}
g_{c_1,c_2,\beta}(t_1\child{\pb}{1},\ldots,t_d\child{\pb}{d})&=g_{c_1,c_2,\beta}(\child{\pb}{1},\ldots,\child{\pb}{d}),\\ 
h_{c_1,c_2,\beta}(t_1\child{\pb}{1},\ldots,t_d\child{\pb}{d})&=h_{c_1,c_2,\beta}(\child{\pb}{1},\ldots,\child{\pb}{d}).
\end{aligned}
\end{equation}
The following proposition, proved in Section~\ref{sec:proverec}, 
 shows the relevance of these functions for analysing the tree.
\newcommand{\statelemtwostep}{
Suppose $q\geq 3$, $ d\geq 2$ and $\beta\in (0,1)$.
For an integer $n\geq2$, let $T$ be the tree $\thetree$  with root $z=v_{d,n}$ and leaves $\Lambda=\theleaves$. Let $\tau \colon \Lambda \to [q]$ be an arbitrary configuration.

Let $z_1,\ldots,z_d$ be the children of $z$ in~$T$ and, for $i\in [d]$, let $\{z_{i,j}\}_{j\in[d]}$  be the children of $z_i$. Denote by $T_{i,j}$ the subtree of $T$ rooted at $z_{i,j}$ and by $\Lambda_{i,j}$  the set of leaves of $T_{i,j}$.  For  $i\in[d]$, $j\in[d]$ and $c\in[q]$, let $r^{(i,j)}_c:=
\Pr_{T_{i,j}}[\sigma(z_{i,j})=c \mid \sigma(\Lambda_{i,j})= \tau(\Lambda_{i,j})]$, and denote by $\rb^{(i,j)}$ the vector $\rb^{(i,j)}=\big(r^{(i,j)}_1, \ldots, r^{(i,j)}_q\big)$.
Then for any colours $c_1\in[q]$ and $c_2\in[q]$  we have 
\[
\frac 
{\Pr_{\thetree}[\sigma(z) = c_1\mid\sigma( \Lambda)=\tau]}
{\Pr_{ \thetree}[\sigma(z) = c_2\mid\sigma( \Lambda)=\tau]}=\prod^{d}_{k=1}h_{c_1,c_2,\beta}\big(\rb^{(k,1)}, \ldots, \rb^{(k,d)}\big).
\]
}
\begin{proposition}\label{lem:twostep}
\statelemtwostep
\end{proposition}

We refer to the recursion introduced in Proposition~\ref{lem:twostep} as the
\emph{two-step recursion}. The two-step recursion will allow us to iteratively bootstrap our bounds on the ratio $\theratio{q}{\beta}{d}{n}$. 
To formalise this, we will use the following definition.
\begin{definition}\label{def:max}
Suppose $q\geq 3$, $ d\geq 2$ and $\beta\in [0,1)$. 
For any $\alpha > 1$, 
let 
$$\triangle_\alpha = \{(p_1,\ldots,p_q)\in \triangle \colon \max_{i\in[q]} p_i \leq \alpha \min_{j\in[q]} p_j \}.$$ (Note that every vector in $\triangle_\alpha$ has strictly positive entries.) For colours $c_1\in[q]$ and $c_2\in[q]$, let
\begin{equation}\label{eq:erff5g54r323234}
M_{\alpha,c_1,c_2,\beta}={\max}_{(\child{\pb}{1},\hdots,\child{\pb}{d})\in\triangle_{\alpha}^d}\, h_{c_1,c_2,\beta}\big(\child{\pb}{1},\hdots,\child{\pb}{d}\big).
\end{equation}
Since  $\triangle_{\alpha}$ is compact and $h_{c_1,c_2,\beta} $ is continuous, the maximisation in \eqref{eq:erff5g54r323234} is well-defined.
\end{definition} 
Definition~\ref{def:max} ensures that
 $\triangle_\alpha$ is the subset of $\triangle$ induced by probability vectors whose entries are within a factor of 
 $\alpha>1$ of each other. $M_{\alpha,c_1,c_2,\beta}$ is the maximum of the two-step recursion function $h_{c_1,c_2,\beta}$ when each of its arguments are from $\triangle_\alpha$. The following proposition gives a preliminary condition 
for establishing uniqueness when $\beta\in(0,1)$ --- it is proved in Section~\ref{sec:conclusion}.
\newcommand{\statelemsteptwouniq}{
Let $q\geq 3$, $d\geq 2$ and $\beta\in (0,1)$. Suppose that for all $\alpha>1$ and any colours $c_1,c_2\in [q]$, it holds that 
\begin{equation*}
M_{\alpha, c_1,c_2,\beta}<\alpha^{1/d}.
\end{equation*}
Then, it holds that $\theratio{q}{\beta}{d}{n}\rightarrow 1$ as $n\rightarrow \infty$, i.e., the $q$-state Potts model with parameter $\beta$ has uniqueness on the $d$-ary tree. 
}
\begin{proposition}\label{lem:steptwouniq}
\statelemsteptwouniq 
\end{proposition}
In the next section, we will show how to simplify the condition in Proposition~\ref{lem:steptwouniq}.

\subsection{A simpler condition for uniqueness}\label{sec:simplecondition}
Proposition~\ref{lem:steptwouniq} gives a sufficient condition on the two-step recursion that is sufficient for establishing uniqueness based on the maximisation of $h_{c_1,c_2,\beta}$. Due to the many variables involved in the maximisation, this is rather complicated for any direct verification. We will simplify this maximisation signifantly by showing that it suffices to consider very special vectors whose entries are either equal to $\alpha$ or 1. We start with the following definition of ``extremal tuples''.
\begin{definition}
Let $\alpha>1$, and consider a colour $c\in [q]$. A tuple $(\child{\pb}{1},\hdots,\child{\pb}{d})$ is called \emph{$(\alpha,c)$-extremal} iff  for all $k \in [d]$,
    \begin{itemize}
        \item for all $c'\in[q]$, either $\child{p}{k}_{c'} = \child{p}{k}_{c}$, or $\child{p}{k}_{c'} = \alpha\cdot\child{p}{k}_{c}$;
        \item there exists $c'\in [q]$ such that $\child{p}{k}_{c'} = \alpha\cdot\child{p}{k}_{c}$.
    \end{itemize}
\end{definition}

Our interest in extremal tuples is justified by the following lemma, whose proof is given in Section~\ref{sec:existence}.
\newcommand{\statelemexistence}{Let $q\geq 3$, $d\geq 2$ and $\beta\in [0,1)$. For any colours $c_1,c_2\in [q]$, there is an $(\alpha,c_2)$-extremal tuple which achieves the maximum in ${\max}_{(\child{\pb}{1},\hdots,\child{\pb}{d})\in\triangle_{\alpha}^d}\, h_{c_1,c_2,\beta}\big(\child{\pb}{1},\hdots,\child{\pb}{d}\big)$ (cf. \eqref{eq:erff5g54r323234}).}
\begin{lemma}\label{lem:existence}
\statelemexistence
\end{lemma}

One of the consequences of Lemma~\ref{lem:existence} is that  the validity of the inequality in Proposition~\ref{lem:steptwouniq} is monotone with respect to $\beta$. In particular, we have the following lemma (also proved in Section~\ref{sec:existence}).
\newcommand{\statelemmonotone}{
Let $q\geq 3$, $d\geq 2$ and $\beta', \beta''\in [0,1)$ with $\beta'\leq \beta''$. Then,  for all $\alpha>1$ and any colours $c_1,c_2\in [q]$, it holds that 
\begin{equation*}
M_{\alpha, c_1,c_2,\beta''}\leq M_{\alpha, c_1,c_2,\beta'}.
\end{equation*}
}
\begin{lemma}\label{lem:mo12no12tone}
\statelemmonotone
\end{lemma}

Another consequence of Lemma~\ref{lem:existence} is that, combined with the scale-free property, it reduces the verification of the condition in Proposition~\ref{lem:steptwouniq} to  the verification of single-variable inequalities in $\alpha$.
These inequalities are obtained by   trying all $d$-tuples  of $q$-dimensional vectors whose entries are 
as follows.
\begin{equation}\label{ex:ca}
\mathrm{Ex}_{c}(\alpha)=\big\{(p_1,\hdots,p_q)
\in \{1,\alpha\}^q
\mid p_{c}=1 \land \exists c'\in[q]  \text{ such that }  p_{c'}=\alpha \big\}.
\end{equation}
The following simplified condition will be our main focus henceforth. 
\begin{condition}\label{cond:we}
Let $q\geq 3$, $d\geq 2$. Set $\beta_*:=\max\big\{1-\tfrac{q}{d+1},0\big\}$. For $\alpha>1$,  let $\mathcal{C}(\alpha)$ be the condition
\begin{equation*}
\mathcal{C}(\alpha): \quad \forall c_1,c_2\in [q],\ h_{c_1,c_2,\beta_*}\big(\child{\pb}{1},\hdots,\child{\pb}{d}\big)<\alpha^{1/d} \mbox{ for all } \child{\pb}{1},\hdots,\child{\pb}{d}\in \mathrm{Ex}_{c_2}(\alpha).
\end{equation*}
If $\mathcal{C}(\alpha)$ holds, we say that the pair $(q,d)$ satisfies Condition~\ref{cond:we} for $\alpha$. 
\end{condition}
Now, to verify the inequality in Proposition~\ref{lem:steptwouniq}, we will show shortly that it suffices only to establish Condition~\ref{cond:we} for all $\alpha>1$, which turns out to be a much more feasible task because of the very explicit form of the set $\mathrm{Ex}_{c_2}(\alpha)$. In the next section, we discuss how to do this in detail, but for now let us state a proposition which asserts that this is indeed sufficient.

\begin{proposition}\label{lem:condimpliesuniq}
Suppose that the pair $(q,d)$ satisfies Condition~\ref{cond:we} for all $\alpha>1$. Let $\beta_*=\max\big\{1-\tfrac{q}{d+1},0\big\}$. Then, the $q$-state Potts model on the $d$-ary tree has uniqueness for all $\beta\in(0,1)$ satisfying $\beta\geq \beta_*$.
\end{proposition}

\begin{proof}
We consider first the case where $\beta_*>0$.  We will show that for all  colours $c_1,c_2\in[q]$, it holds that 
\begin{equation}\label{eq:z4r34rf4rf4}
M_{\alpha,c_1,c_2,\beta_*}<\alpha^{1/d} \mbox{ for all $\alpha>1$}.
\end{equation}
Then by Lemma~\ref{lem:mo12no12tone}, we obtain that, for all $\beta\in [\beta^*,1)$ it holds that $M_{\alpha,c_1,c_2,\beta}<\alpha^{1/d}$ as well for all $\alpha>1$ and therefore, by Proposition~\ref{lem:steptwouniq}, the Potts model has uniqueness for all such $\beta$.

To prove \eqref{eq:z4r34rf4rf4}, consider an arbitrary $\alpha>1$ and colours $c_1,c_2\in [q]$. By Lemma~\ref{lem:existence}, there exists an $(\alpha,c_2)$-extremal tuple $(\pb^{(1)},\hdots,\pb^{(d)}\big)$ such that 
\begin{equation}\label{eq:z4r34rf4rf4a}
M_{\alpha,c_1,c_2,\beta_*}=h_{c_1,c_2,\beta_*}(\pb^{(1)},\hdots,\pb^{(d)}\big).
\end{equation}
For $c\in [q]$ and $k\in [d]$, denote by $p^{(k)}_c$ the entry of $\pb^{(k)}$ corresponding to colour $c$ and let $\hat{\pb}^{(k)}$ be the vector $t_k {\pb}^{(k)}$ where $t_k=1/p^{(k)}_{c_2}$. By the definition of an $(\alpha,c_2)$-extremal tuple, we have that 
\[\hat{\pb}^{(1)},\hdots,\hat{\pb}^{(d)}\in \mathrm{Ex}_{c_2}(\alpha).\] 
Moreover, by the scale-free property \eqref{eq:scalefree23} we have that 
\begin{equation}\label{eq:z4r34rf4rf4b}
h_{c_1,c_2,\beta_*}(\hat{\pb}^{(1)},\hdots, \hat{\pb}^{(d)})=h_{c_1,c_2,\beta_*}(\child{\pb}{1},\ldots,\child{\pb}{d}).
\end{equation}
Finally, since the pair $(q,d)$ satisfies Condition~\ref{cond:we} for all $\alpha>1$, we have that 
\begin{equation}\label{eq:z4r34rf4rf4c}
h_{c_1,c_2,\beta_*}(\hat{\pb}^{(1)},\hdots, \hat{\pb}^{(d)})<\alpha^{1/d}.
\end{equation}
Combining \eqref{eq:z4r34rf4rf4a}, \eqref{eq:z4r34rf4rf4b}, and \eqref{eq:z4r34rf4rf4c} yields \eqref{eq:z4r34rf4rf4}, as needed.

The case $\beta_*=0$ is analogous. Now, we need to show that we have uniqueness for all $\beta\in (0,1)$ assuming that Condition~\ref{cond:we} holds for all $\alpha>1$. Just as before, we obtain that  
$M_{\alpha,c_1,c_2,\beta_*}<\alpha^{1/d}$ for all $\alpha>1$ and hence   by Lemma~\ref{lem:mo12no12tone}, we have that $M_{\alpha,c_1,c_2,\beta}<\alpha^{1/d}$ for all $\alpha>1$ and $\beta\in (0,1)$. Uniqueness for $\beta\in(0,1)$ therefore follows from applying Proposition~\ref{lem:steptwouniq}.
\end{proof}

\begin{remark}
Note that, when $\beta_*>0$, the conclusion of Proposition~\ref{lem:condimpliesuniq} asserts uniqueness in the half-open interval $[\beta_*,1)$;  when $\beta_*=0$,  it instead asserts uniqueness in the open interval $(0,1)$. 
\end{remark}

\subsection{Verifying the Condition}
In this section, we give more details on how to verify Condition~\ref{cond:we}.

To apply Proposition~\ref{lem:condimpliesuniq}, we will need to verify Condition~\ref{cond:we}. The latter is fairly simple to verify for small values of $q,d$ since it reduces to single-variable inequalities in $\alpha$. We illustrate the details when $(q,d)=(3,3)$ and $(q,d)=(4,4)$. 

\begin{lemma}\label{lem:examplecond}
The pairs $(q,d)=(3,3)$ and $(q,d)=(4,4)$ satisfy Condition~\ref{cond:we} for all $\alpha>1$.
\end{lemma}
\begin{proof}
By symmetry among the colours, it suffices to verify the condition for colours $c_1=1$ and $c_2=q$. In Section~\ref{app:examplecond}, we just try all possible $d$-tuples $(\pb^{(1)},\hdots, \pb^{(d)})$ with $\pb^{(1)},\hdots, \pb^{(d)}\in \mathrm{Ex}_q(\alpha)$. For each such $d$-tuple, the  inequality 
\[h_{c_1,c_2,\beta_*}(\pb^{(1)},\hdots, \pb^{(d)})<\alpha^{1/d}\]
is a single-variable inequality in $\alpha$ which can be verified using Mathematica's {\tt Resolve} function for all $\alpha>1$. For $(q,d)=(3,3)$ and $(q,d)=(4,4)$, all the resulting inequalities are satisfied.
\end{proof}

Combining Lemma~\ref{lem:examplecond} with 
Proposition~\ref{lem:condimpliesuniq} we get the following immediate corollary.
\begin{corollary} \label{cor:example}
The $3$-state Potts model on the $3$-ary tree has uniqueness for all $\beta\in[1/4,1)$.
The $4$-state Potts model on the $4$-ary tree has uniqueness for all $\beta\in[1/5,1)$.
\end{corollary}
Corollary~\ref{cor:example} establishes the uniqueness threshold
for $(q,d)=(3,3)$ and $(q,d)=(4,4)$.
More generally, we are interested in the following question:
When is Condition~\ref{cond:we} satisfied for all $\alpha>1$? We conjecture the following.
\begin{conjecture}\label{conj}
When $q\neq d+1$, the pair $(q,d)$ satisfies Condition~\ref{cond:we} for all $\alpha>1$.
\end{conjecture}

We have only been able to verify Conjecture~\ref{conj} for specific values of $q,d$ (with methods similar to those used in the proof of Lemma~\ref{lem:examplecond}). However, it is important to
note that the restriction $q\neq d+1$ in the conjecture cannot be  removed.
For example,   the pair $(q,d)=(3,2)$ does not
satisfy Condition~\ref{cond:we} for all $\alpha>1$ --- it only satisfies
the condition for $\alpha$ fairly close to~$1$.
Thus, to prove Theorem~\ref{thm:uniqueness} 
we need a different argument to account for the case  $(q,d)=(3,2)$. 

Thus, instead of trying to prove Conjecture~\ref{conj} for \emph{all} values of~$\alpha$
(which wouldn't be enough for our theorem),
 we follow Jonasson's approach and use the one-step recursion to argue that the ratio~$\gamma(q,\beta,d,n)$ 
 gets moderately close to 1; close enough that we can then use the two-step recursion to finish the proof of uniqueness. Note that, in contrast to the two-step recursion, the one-step recursion is not sufficient on its own to obtain tight uniqueness results for any values of $q,d$ (this was also observed by Jonasson \cite{Jonasson02} in the case of colourings). 

First, we state the one-step recursion that we are going to use on the tree. This recursion, as well as the two-step recursion of Proposition~\ref{lem:twostep}, are well-known,  but we prove them explicitly  in Section~\ref{sec:proverec} for completeness.  

\newcommand{\statelemonestep}{ Suppose $q\geq 3$, $ d\geq 2$ and 
    $\beta\in (0,1)$.    For an integer $n\geq 1$, let $T$ be the tree $\thetree$  with root $v=v_{d,n}$ and leaves $\Lambda=\theleaves$. Let $\tau:\Lambda\rightarrow[q]$ be an arbitrary configuration.

    Let $v_1,\ldots,v_d$ be the children of $v$ in~$T$. For $i\in [d]$, let $T_i$ be the subtree of $T$ rooted at $v_i$ and let $\Lambda_i$ denote the set of leaves of the subtree $T_i$.
    Then, for any colour $c\in[q]$, it holds that 
    \begin{equation*}
    \Pr_{T}[\sigma(v) = c\mid \sigma(\Lambda)=\tau] = 
    \frac{\prod_{i=1}^d \big(1-(1-\beta)
        \Pr_{T_i}[\sigma(v_i)= c \mid \sigma(\Lambda_{i}) = \tau(\Lambda_{i})]\big)}
    {\sum_{c'=1}^q\prod_{i=1}^d \big( 1-(1-\beta)
        \Pr_{T_i}[\sigma(v_i) = c'\mid \sigma(\Lambda_{i}) = \tau(\Lambda_{i})]\big)}.
    \end{equation*}
}
\begin{proposition}\label{lem:onesteprecursion}
   \statelemonestep
\end{proposition}
Tracking the one-step recursion relatively accurately requires a fair amount of work, and to aid the verification of Condition~\ref{cond:we} in the case $q=3$, we do this for general values of $d$. In particular, we prove the following lemma in Section~\ref{sec:bounds12onestep}.
\newcommand{\statelemonestepuniq}{
Let $q=3$ and $c\in [3]$ be an arbitrary colour. For $d\geq 2$, consider the $d$-ary tree $\thetree$ with height $n$ and let $\tau:\theleaves\rightarrow [3]$ be an arbitrary configuration on the leaves.

When $d=2$, for all $\beta\in (0,1)$, for all sufficiently large $n$ it holds that
\[\frac{459}{2000}\leqslant \Pr_{\thetreetwo}[\sigma(\theroottwo)=c\mid \sigma(\theleavestwo) = \tau] \leqslant \frac{1107}{2500}.\]
When $d\geq 3$, for all $\beta\in [1-\tfrac{3}{d+1},1)$, there exist sequences $\{L_n\}$ and $\{U_n\}$ (depending on $d$ and $\beta$) such that for all sufficiently large $n$ 
\[L_n\leq\Pr_{\thetree}[\sigma(\theroot)=c\mid \sigma(\theleaves)=\tau]\leq U_n \mbox{ and } U_n/L_n\leq 53/27.\]
}
\begin{lemma}\label{lem:onestepuniq}
\statelemonestepuniq
\end{lemma}
The following corollary is an immediate consequence of Lemma~\ref{lem:onestepuniq}.

\begin{corollary}\label{thm:ratiobound}
For $d=2$ and every $\beta\in (0,1)$  
there is a positive integer $n_0$ such that,
for every $n\geq n_0$,  we have
$\theratio{3}{\beta}{d}{n} \leq {53}/{27}$.

For every $d\geq 3$ and every 
$\beta$ satisfying  $1-3/(d+1)\leq \beta <1$
there is a positve integer $n_0$ such that,
for every $n\geq n_0$, we have
$\theratio{3}{\beta}{d}{n}  \leq {53}/{27}$.
\end{corollary}

We combine this with the following lemma which verifies Condition~\ref{cond:we} for all $\alpha\in(1,53/27]$. The proof  is given in Section~\ref{sec:twosteptoprove}. 
\newcommand{\statelemtwosteptoprove}{
Let $q=3$ and $d\geq 2$. Then, the pair $(q,d)$ satisfies Condition~\ref{cond:we} for all $\alpha\in(1,53/27]$.
}
\begin{lemma}\label{lem:twosteptoprove}
\statelemtwosteptoprove
\end{lemma}
Using Corollary~\ref{thm:ratiobound} and Lemma~\ref{lem:twosteptoprove}, we give the proof of Theorem~\ref{thm:main}  (which implies 
  Theorem~\ref{thm:uniqueness})
in Section~\ref{sec:conclusion}.

\section{Concluding uniqueness}\label{sec:conclusion}
In this section, we prove Proposition~\ref{lem:steptwouniq} and also conclude the proof of Theorem~\ref{thm:main} (assuming for now Lemmas~\ref{lem:onestepuniq} and~\ref{lem:twosteptoprove} and also Lemma~\ref{lem:existence}, which we have already used). Recall that 
\begin{equation*}\tag{\ref{eq:erff5g54r323234}}
M_{\alpha,c_1,c_2,\beta}={\max}_{(\child{\pb}{1},\hdots,\child{\pb}{d})\in\triangle_{\alpha}^d}\, h_{c_1,c_2,\beta}\big(\child{\pb}{1},\hdots,\child{\pb}{d}\big).
\end{equation*}

We will need the following proposition.
\begin{proposition}\label{lem:gammaisdecreasing}
Let $q\geq 3$, $d\geq 2$ and $\beta\in (0,1)$. Suppose that, for some integer $n\geq 3$ and some $\alpha>1$, it holds that $\theratio{q}{\beta}{d}{n-2}=\alpha$ and  $M_{\alpha,c_1,c_2,\beta}<\alpha^{1/d}$  for all colours $c_1,c_2\in[q]$. 
Then $\theratio{q}{\beta}{d}{n}\leq(M_{\alpha,c_1,c_2,\beta})^d<\theratio{q}{\beta}{d}{n-2}$.
\end{proposition}
\begin{proof}
Consider the tree $\thetree$ with root $z=v_{d,n}$ and leaves  $\Lambda= \theleaves$. Let $\tau: \theleaves\rightarrow [q]$ be an arbitrary configuration. As in Proposition~\ref{lem:twostep}, let $\{z_{i,j}\}_{i,j\in [d]}$ denote the grandchildren of the root, 
let $T_{i,j}$ be the subtree  of $T$ rooted at $z_{i,j}$, and let $\Lambda_{i,j}$ be the set of leaves of $T_{i,j}$. Further, let $\rb^{(i,j)}$ be the marginal distribution at $z_{i,j}$  in the subtree $T_{i,j}$, conditioned on the configuration $\tau(\Lambda_{i,j})$.  

By the assumption $\theratio{q}{\beta}{d}{n-2}=\alpha$ and the definition \eqref{eq:ratio} of the ratio $\theratio{q}{\beta}{d}{n-2}$, we have that   $\rb^{(i,j)}\in \triangle_\alpha$ for all $i,j\in [d]$.
Proposition~\ref{lem:twostep} also guarantees that for  colours $c_1\in [q]$ and $c_2\in [q]$ we  have
    \begin{equation*}
    \frac
    {\Pr_{\mathbb{T}_{d,n}}[\sigma(z) = c_1\mid\sigma( \Lambda)=\tau]}
    {\Pr_{\mathbb{T}_{d,n}}[\sigma(z) = c_2\mid\sigma( \Lambda)=\tau]}=\prod^{d}_{k=1}h_{c_1,c_2,\beta}\big(\rb^{(k,1)}, \ldots, \rb^{(k,d)}\big)\leq (M_{\alpha,c_1,c_2,\beta})^d <\alpha,
    \end{equation*}
where the strict inequality follows by the assumption that  $M_{\alpha,c_1,c_2,\beta}<\alpha^{1/d}$. Since $\tau$ was an arbritrary configuration on the leaves $\Lambda$, we obtain that $\theratio{q}{\beta}{d}{n}<\theratio{q}{\beta}{d}{n-2}$ as needed.
\end{proof}

We start with Proposition~\ref{lem:steptwouniq}, which we restate here for convenience.
\begin{lemsteptwouniq}
\statelemsteptwouniq
\end{lemsteptwouniq}
\begin{proof}
Fix $q$, $d$ and $\beta$ as in the statement.
For all $n\geq 1$, let $\alpha_n:= \theratio{q}{\beta}{d}{n}$. We may  assume that $\alpha_n>1$ for all $n\geq 1$ (otherwise, uniqueness follows trivially by 
choosing $n_0$ such that $\alpha_{n_0}=1$ and then applying
  Proposition~\ref{lem:onesteprecursion} repeatedly to show $\alpha_n = 1$ for all $n\geq n_0$.) 

Using Proposition~\ref{lem:gammaisdecreasing} and the assumption that $M_{\alpha,c_1,c_2,\beta}<\alpha^{1/d}$ for all $\alpha>1$ and colours $c_1,c_2\in [q]$, we obtain that
\begin{equation}\label{eq:decreasing}
1<\alpha_n<\alpha_{n-2}.
\end{equation}
This implies that both of the sequences $\{\alpha_{2n}\}$ and $\{\alpha_{2n+1}\}$ are decreasing. Since both 
of these sequences is bounded below by 1, we obtain that for $n\rightarrow \infty$ it holds that\footnote{The notation $\alpha_{2n}\downarrow \alpha_{\mathrm{ev}}$ means that the sequence $\alpha_{2n}$ converges to $\alpha_{\mathrm{ev}}$ by decreasing monotonically.}
\[ 
\alpha_{2n}\downarrow \alpha_{\mathrm{ev}}, \quad \alpha_{2n+1}\downarrow \alpha_{\mathrm{odd}}\]
for some $\alpha_{\mathrm{ev}}, \alpha_{\mathrm{odd}}\geq 1$. We claim that in fact both of $\alpha_{\mathrm{ev}}, \alpha_{\mathrm{odd}}$ are equal to $1$, which proves that 
$  \theratio{q}{\beta}{d}{n}\rightarrow 1$ as $n\rightarrow \infty$.

Suppose for the sake of contradiction that  $\alpha_{\mathrm{ev}}>1$ (a similar argument applies for $\alpha_{\mathrm{odd}}$). Let $\mathbf{m}_{2n}=(\child{\pb}{1}_{2n},\hdots,\child{\pb}{d}_{2n})$ achieve the maximum in \eqref{eq:erff5g54r323234} for $\alpha=\alpha_{2n}$, i.e.,
\[M_{\alpha_{2n},c_1,c_2,\beta}=h_{c_1,c_2,\beta}\big(\child{\pb}{1}_{2n},\hdots,\child{\pb}{d}_{2n}\big).\] 
Note that for all $n\geq 1$ we have that
\begin{equation}\label{eq:contrah1242r2}
h_{c_1,c_2,\beta}\big(\child{\pb}{1}_{2n},\hdots,\child{\pb}{d}_{2n}\big)\geq (\alpha_{\mathrm{ev}})^{1/d};
\end{equation}
otherwise, we would have that $M_{\alpha_{2n}, c_1,c_2,\beta}<(\alpha_{\mathrm{ev}})^{1/d}$ and hence, by Proposition~\ref{lem:gammaisdecreasing}, we would have that $\alpha_{2n+2}<\alpha_{\mathrm{ev}}$, contradicting that $\alpha_{2n}\downarrow \alpha_{\mathrm{ev}}$. 

Moreover, observe that $\mathbf{m}_{2n}$ belongs to the compact space $\triangle^d$ for all $n\geq 1$ and therefore there exists a subsequence $\{n_k\}_{k\geq 1}$ and $(\child{\pb}{1},\hdots,\child{\pb}{d})\in \triangle^d$ 
such that 
\[\mathbf{m}_{2n_k}\rightarrow (\child{\pb}{1},\hdots,\child{\pb}{d}),\]
In fact, since $\{\alpha_{2n_k}\}_{k\geq 1}$ is a subsequence of the convergent sequence $\{\alpha_{2n}\}_{n\geq 1}$, we have that  the sequence $\{\alpha_{2n_k}\}$ converges to $\alpha_{\mathrm{ev}}$ as well. From $\mathbf{m}_{2n_k}\in \triangle_{\alpha_{2n_k}}^d$, we therefore obtain that $(\child{\pb}{1},\hdots,\child{\pb}{d})\in \triangle^d_{\alpha_{\mathrm{ev}}}$. Applying the assumption $M_{\alpha,c_1,c_2,\beta}<\alpha^{1/d}$ for $\alpha=\alpha_{\mathrm{ev}}$, we therefore have that
\begin{equation}\label{eq:contrah1242r2b}
h_{c_1,c_2,\beta}(\child{\pb}{1},\hdots,\child{\pb}{d})<(\alpha_{\mathrm{ev}})^{1/d}.
\end{equation}
Since the function $h_{c_1,c_2,\beta}$ is continuous on $\triangle^{d}$ 
for all $\beta\in (0,1)$, we have that as $k\rightarrow\infty$
\[h_{c_1,c_2,\beta}(\child{\pb}{1}_{2n_k},\hdots,\child{\pb}{d}_{2n_k})\rightarrow h_{c_1,c_2,\beta}(\child{\pb}{1},\hdots,\child{\pb}{d}).\]
This contradicts \eqref{eq:contrah1242r2} and \eqref{eq:contrah1242r2b}.  Therefore, $\alpha_{\mathrm{ev}}=1$, and similarly $\alpha_{\mathrm{odd}}=1$, completing the proof.
\end{proof}

Assuming Lemmas~\ref{lem:onestepuniq} and~\ref{lem:twosteptoprove}, we can also conclude the proof of Theorem~\ref{thm:main} in a similar way.
\begin{thmmain}
\statethmmain
\end{thmmain}
\begin{proof}
We first consider the case $d\geq 3$. Let $\beta\in[1-\tfrac{3}{d+1},1)$ and for all $n\geq 1$, set $\alpha_n=\gamma(3,d,\beta,n)$. By Lemma~\ref{lem:onestepuniq}, we have that there exists $n_0$ such that for all $n\geq n_0$, it holds that
\[\alpha_n\in (1,53/27].\]
(The reason that the left end-point of the interval is open is that we finish if $\alpha_n=1$, as in the proof
of Proposition~\ref{lem:steptwouniq}.)

By Lemma~\ref{lem:twosteptoprove},
the pair $(q,d)$ satisfies Condition~\ref{cond:we}
for $\alpha_n$ (for $n\geq n_0$).
By the definition of Condition~\ref{cond:we},
 for all $c_1,c_2\in[q]$,
 and 
 $(\child{\pb}{1},\hdots,\child{\pb}{d})\in \mathrm{Ex}_{c_2}(\alpha_n)$,
 we have
 $ h_{c_1,c_2,\beta_*}\big(\child{\pb}{1},\hdots,\child{\pb}{d}\big)<\alpha_n^{1/d}  $.
 By the scale-free property of~$h_{c_2,c_2,\beta_*}$,
 $M_{\alpha_n,c_1,c_2,\beta_*} < \alpha_n^{1/d}$
 so, by Lemma~\ref{lem:mo12no12tone},
  $M_{\alpha_n,c_1,c_2,\beta} < \alpha_n^{1/d}$.
   Using Proposition~\ref{lem:gammaisdecreasing}  we obtain that for all $n\geq n_0+2$, it holds that
\begin{equation}\label{eq:decreasingb}
1<\alpha_n< \alpha_{n-2}.
\end{equation}
This implies that both of the sequences $\{\alpha_{2n}\}_{n\geq n_0}$ and $\{\alpha_{2n+1}\}_{n\geq n_0}$ are decreasing, and since both are bounded below by 1, they converge. 
We now use an argument that is almost identical to the one used in 
the proof of  Proposition~\ref{lem:steptwouniq}.
 The only difference is that now the sequences start from $2n_0$ and $2n_0+1$ instead of from $n=2$ and $n=3$, respectively.
 Using this argument,
 we obtain that the limits of $\{\alpha_{2n}\}_{n\geq n_0}$ and $\{\alpha_{2n+1}\}_{n\geq n_0}$ must be equal to 1, thus proving that  $\gamma(3,d,\beta,n)\rightarrow 1$ as $n\rightarrow \infty$.

The argument for the case $d=2$ and $\beta\in (0,1)$ is actually the same; the only difference to the case $d\geq 3$ is that $\beta$ lies in an open interval instead of a half-open interval.
\end{proof}

\section{Proving Tree Recursions}
\label{sec:proverec}

In this section, we give proofs of the (standard) tree recursions, which we have already used.
We first prove Proposition~\ref{lem:onesteprecursion} for the one-step recursion.

\begin{lemonesteprecursion}
\statelemonestep
\end{lemonesteprecursion}
\begin{proof}
For any graph $G$,
we use  $V(G)$ to denote the vertex set of   $G$.
Recall that, for any configuration $\sigma: V(G)\rightarrow [q]$, its weight in the Potts model with parameter $\beta$ is given by 
$w_G(\sigma)=\beta^{m(\sigma)}$
where $m(\sigma)$ denotes the number of monochromatic edges in $G$ under the assignment $\sigma$.
If $v\in V(G)$ then
we use the notation  $w_G(\sigma(v)= c)$ to
denote the quantity $w_G(\sigma(v)=c) = \sum_{\sigma'\colon V(G)\to[q], \sigma'(v)=c} w_G(\sigma')$. 
Similarly, if $S$ is a subset of $V(G)$
and $\tau$ is an assignment $\tau:S \rightarrow [q]$
then we use the notation
$w_G(\sigma(S)=\tau) = \sum_{\sigma'\colon V(G)\to[q], \sigma'(S)=\tau} w_G(\sigma')$. 
We will typically be interested in the case where $G$ is a sub-tree of~$T$.
We have
    \begin{equation}\label{eq:conditioned}
    \Pr_{T}[\sigma(v) = c\mid \sigma(\Lambda)=\tau] =\frac{\Pr_{T}[\sigma(v) = c\land \sigma(\Lambda)=\tau]}{\Pr_{T}[\sigma(\Lambda)=\tau]}=\frac{w_{T}(\sigma(v) = c\land \sigma(\Lambda)=\tau)}{w_{T}(\sigma(\Lambda)=\tau)}.
    \end{equation}

    Now we compute $w_T(\sigma(v) = c\land \sigma(\Lambda)=\tau)$:
    \begin{align}\label{eq:weight}
    w_T(\sigma(v) &= c\land \sigma(\Lambda)=\tau)\cr =&\prod_{i=1}^d
 \Big( w_{T_i}\big(\sigma(v_i) \neq c\land \sigma(\Lambda_i)=\tau(\Lambda_i)\big)+\beta\, w_{T_i}\big(\sigma(v_i) = c\land \sigma(\Lambda_i)=\tau(\Lambda_i)\big)\Big)\cr
    =&\prod_{i=1}^d w_{T_i}\big(\sigma(\Lambda_i)=\tau(\Lambda_i)\big) \left(1- (1-\beta)\frac{w_{T_i}\big(\sigma(v_i) = c\land \sigma(\Lambda_i)=\tau(\Lambda_i)\big)}{w_{T_i}\big(\sigma(\Lambda_i)=\tau(\Lambda_i)\big) }\right)\cr
    =&\prod_{i=1}^d w_{T_i}\big(\sigma(\Lambda_i)=\tau(\Lambda_i)\big)
    \Big(1-(1-\beta)
    \Pr_{T_i}\big[\sigma(v_i)= c \mid \sigma(\Lambda_{i}) = \tau(\Lambda_{i})\big]\Big).
    \end{align}
    Also, we have $w_T(\sigma(\Lambda)=\tau) = \sum_{c'=1}^q w_T(\sigma(v) = c'\land\sigma(\Lambda)=\tau)$. Combining this with~ \eqref{eq:conditioned}~and~\eqref{eq:weight}, we obtain the statement of the lemma.
\end{proof}

We next prove Proposition~\ref{lem:twostep} for the two-step recursion of Section~\ref{sec:introtwostep}.
\begin{lemtwostep}
\statelemtwostep
\end{lemtwostep}
\begin{proof}[Proof of Proposition~\ref{lem:twostep}]
For  $i\in[d]$ and $c\in[q]$, let $T_i$ be the subtree of $T$ rooted at $z_i$ 
with leaves $\Lambda_{T_i}$
and let
$r^{(i)}_c=\Pr_{T_i}[\sigma(z_{i})=c \mid \sigma(\Lambda_{T_i})= \tau(\Lambda_{T_i})]$; let $\rb^{(i)}$ be the vector  $(r^{(i)}_1, \ldots, r^{(i)}_q )$.

For every $i\in[d]$ and $c\in[q]$ 
we can apply  Proposition~\ref{lem:onesteprecursion}
to $T_i$ to obtain 
\begin{equation*}
r^{(i)}_c = \frac{\prod_{j=1}^d \big(1-(1-\beta)r^{(i,j)}_c\big)}{\sum_{c'=1}^q \prod_{j=1}^d \big(1-(1-\beta)r^{(i,j)}_{c'}\big)}.
\end{equation*}
We have $r^{(i)}_c>0$ for every $c\in[q]$, hence  we can apply
 Proposition~\ref{lem:onesteprecursion} to $T$ 
to obtain
\begin{equation*}
\Pr_{\thetree}[\sigma(z)=c\mid \sigma(\Lambda) = \tau] = \frac{\prod_{i=1}^d \big(1-(1-\beta)r^{(i)}_c\big)}{\sum_{c'=1}^q \prod_{i=1}^d \big(1-(1-\beta)r^{(i)}_{c'}\big)}.
\end{equation*}

 Thus, for every $i\in[d]$ and $c_1,c_2\in[q]$,
\begin{equation}\label{eq:f34g11123}
    \frac{r^{(i)}_{c_1}}{r^{(i)}_{c_2}}= \prod^d_{j=1}\frac{1-(1-\beta)r^{(i,j)}_{c_1}}{1-(1-\beta)r^{(i,j)}_{c_2}}=g_{c_1,c_2,\beta} \big(\rb^{(i,1)},\ldots,\rb^{(i,d)}\big).
\end{equation}
Analogously, for every $c_1, c_2 \in [q]$, we have
\begin{equation}\label{eq:f34g11123a}
\frac{\Pr_{\thetree}[\sigma(z)=c_1\mid \sigma(\Lambda) = \tau]}{\Pr_{\thetree}[\sigma(z)=c_2\mid \sigma(\Lambda) = \tau]} = \prod_{k=1}^d\frac{1-(1-\beta)r^{(k)}_{c_1}}{1-(1-\beta)r^{(k)}_{c_2}}=\prod_{k=1}^d\bigg(1+\frac{(1-\beta)\big(1-r^{(k)}_{c_1}/r^{(k)}_{c_2}\big)}{\beta+\sum_{c\neq c_2} r^{(k)}_{c}/r^{(k)}_{c_2}}\bigg)
\end{equation}
Plugging \eqref{eq:f34g11123} into \eqref{eq:f34g11123a}, and using the definition of $h_{c_1,c_2,\beta}$ from \eqref{eq:gh12def}, we obtain the statement of the lemma.
\end{proof}

\section{Bounds from the one-step recursion -- Proof of Lemma~\ref{lem:onestepuniq}}\label{sec:bounds12onestep}
In this section, we prove Lemma~\ref{lem:onestepuniq}.
\subsection{Bounding the marginal probability at the root by the one-step recursion}
\label{sec:fns}
We begin by giving an upper and a lower bound for the marginal probability that the root is assigned a colour $c$ via the one-step recursion (see the upcoming Lemma~\ref{lem:boundbyonesteprecursion}).

First we define two functions. Let
\[f_u(d, \beta, x,y) =
\frac{\big(1-(1-\beta)y\big)^d}{\big(1-(1-\beta) y\big)^d + 2\big(1-(1-\beta) x\big)^{d/2}\big(1-(1-\beta)(1-x-y)\big)^{d/2}} \]
and
\[f_\ell(d, \beta,x, y) = \frac{\big(1-(1-\beta) x\big)^d}{\big(1-(1-\beta) x\big)^d + \big(1-(1-\beta) y\big)^d + \big(1-(1-\beta)(1-x-y)\big)^d}.\]

\def\tempx{z}
\def\newx{x}
We will use the following lemma.  
\begin{lemma}\label{lem:ferfe}
    Let $f$ be a convex function on an interval $I=[a,b]$. 
    \begin{enumerate}
        \item \label{it:dec1} Given $\rho\in [2a,a+b]$, the function $g(\newx) = f(\newx) + f(\rho-\newx)$ is decreasing on $J=[a,\rho/2]$.
        \item \label{it:inc2} Given $\rho\in [a+b,2b]$, the function $g(\newx) = f(\newx) + f(\rho-\newx)$ is increasing on $J=[\rho/2,b]$. 
       \end{enumerate}
   \end{lemma}
   \begin{proof}
       We first prove Item~\ref{it:dec1}.  Suppose  $y,\tempx\in J$  satisfy  $y< \tempx$, we will show that $g(\tempx)\leq g(y)$.
       We have $y < \tempx \leq \rho-\tempx < \rho-y$ and $y\geq a$, $\rho-y\leq b$ (using $\rho\leq a+b$). It follows that all of 
       $y,\tempx,\rho-y,\rho-\tempx$ belong to $I$. Moreover, by the convexity of $f$ on $I$, we conclude that  the slope of $f$ in the interval 
       $[\rho-\tempx,\rho-y]$ is greater or equal to the slope of $f$ in the interval $[y,\tempx]$,  i.e., 
       $\frac{f(\rho-y)-f(\rho-\tempx)}{(\rho-y)-(\rho-\tempx)}\geq \frac{f(\tempx)-f(y)}{\tempx-y}$. 
       Re-arranging, we obtain $g(\tempx) \leq g(y)$.
       
       The proof of Item~\ref{it:inc2} is  analogous. For $y,\tempx\in J$  satisfying  $y< \tempx$, we have that $\rho-\tempx<\rho-y\leq y < \tempx$ and all of $y,\tempx,\rho-y,\rho-\tempx$ belong to $I$ 
       (using $\rho\geq a+b$). By the convexity of $f$ on $I$, we conclude that  the slope of $f$ in the interval 
       $[\rho-\tempx,\rho-y]$ is less or equal to the slope of $f$ in the interval $[y,\tempx]$, which gives that $g(\tempx)\geq g(y)$.
   \end{proof}

The following lemma gives recursively-generated  bounds on 
the probability that the root of $\thetree$ is a given colour.
 
\begin{lemma}\label{lem:boundbyonesteprecursion}
    Suppose $q= 3$, $d\geq 2$ and $\beta\in [0,1]$.
    For any $n\geq0$, let $T=\thetree$  with root $v=v_T$ and leaves $\Lambda=\Lambda_T$.
    Let $v_1,\ldots,v_d$ be the children of $v$ in~$T$, $T_i$ be the subtree of $T$ rooted at $v_i$ and $\Lambda_i$ denote the set of leaves of the subtree $T_i$.
    Consider any configuration $\tau \colon \Lambda \to [q]$
    and any real numbers $L,U \in [0,1]$ 
    such that, for all $i\in [d]$ and all $j\in[3]$
    we have 
    $$L\leq \Pr_{T_i}[\sigma(v_i)=j \mid \sigma(\Lambda_i)= \tau(\Lambda_i)]\leq U.$$
    Then, for all colours $c\in[q]$, we also have 
    \begin{equation*}
    f_\ell(d, \beta, U, L) \leq \Pr_T[\sigma(v)=c \mid \sigma(\Lambda)= \tau] \leq f_u(d, \beta, U, L).
    \end{equation*}
\end{lemma}
\begin{proof}
    By symmetry between the colours, we may assume that $c=1$. 
    Let $p= \Pr_T[\sigma(v)=1 \mid \sigma(\Lambda)= \tau]$.
   For any colour $c'\in [3]$ and 
   any child
   $i\in [d]$, let 
    $p_{i,c'}=\Pr_{T_i}[\sigma(v_i)=c' \mid \sigma(\Lambda_i)= \tau(\Lambda_i)]$.
    
    For convenience, let $\hat{\beta} := 1-\beta \in [0,1]$.  
    By Proposition~\ref{lem:onesteprecursion}, with $q=3$ and $c=1$, we have that
    \begin{equation}\label{eq:btb5y56}
    p = \frac{1}{1+R}, \mbox{ where } R := \frac{\sum_{c'=2}^3\prod_{i=1}^d\big(1-\hat\beta p_{i,c'}\big)}{\prod_{i=1}^d \big(1-\hat\beta p_{i,1}\big)}.
    \end{equation}
    We first show that $ p \geq f_\ell(d, \beta, U, L)$. Since $p_{i,1}\leq U$ for every $i\in [d]$, we obtain that
    \begin{equation}\label{eq:rvrveetr2}
    R \leq \frac{\sum_{c'=2}^3\prod_{i=1}^d\big(1-\hat\beta p_{i,c'}\big)}{\big(1-\hat\beta U\big)^d}.
    \end{equation}
    For $c'\in[3]$, let  $\bar{p}_{c'}$   denote the mean $\frac{1}{d}\sum_{i\in[d]} p_{i,c'}$. 
    The function $f(x) = \ln(1-\hat{\beta} x)$ is concave on the interval $[0,1]$
    so by Jensen's inequality
    \begin{equation*}
    \frac{1}{d} \sum_{i=1}^d \ln(1-\hat{\beta} p_{i,c'}) \leq \ln(1-\hat{\beta} \bar{p}_{c'}).
    \end{equation*}
    Thus, for $c'\in\{2,3\}$ we have
    $   \prod_{i=1}^d\big(1-\hat\beta p_{i,c'}\big)\leq 
    \big(1-\hat\beta \bar{p}_{c'}\big)^d$ which implies
    \begin{equation}\label{eq:4g4g6gg}
    \sum_{c'=2}^3\prod_{i=1}^d\big(1-\hat\beta p_{i,c'}\big)\leq    
    \big(1-\hat\beta \bar{p}_2\big)^d+\big(1-\hat\beta \bar{p}_3\big)^d.
    \end{equation}

    Let $f(x) =    (1-\hat\beta x)^d$. Let $a=0, b=1, \rho=\bar{p}_2 + \bar{p}_3$ and consider the interval $I=[a,b]$.  Since $f$ is convex on the interval $I$ and $\rho\in [2a,a+b]$, Item~\ref{it:dec1} of Lemma~\ref{lem:ferfe}
    implies that the function $g(x) = f(x) + f(\rho-x)$ is decreasing on
    $J=[a,\rho/2]=[0,\rho/2]$. 
    Since $L\leq \bar{p}_2$ and $L\leq \bar{p}_3$,
    the values $L$ and  $\min\{\bar{p}_2,\bar{p}_3\}$  are in~$J$
    and  $\min\{\bar{p}_2,\bar{p}_3\}\geq L$, so $g(\min\{\bar{p}_2,\bar{p}_3\})\leq g(L)$, i.e.,
    $$
    \big(1-\hat\beta \bar{p}_2\big)^d+\big(1-\hat\beta \bar{p}_3\big)^d
    \leq \big(1-\hat\beta L\big)^d+\big(1-\hat\beta (\bar{p}_2+\bar{p}_3-L)\big)^d.
    $$
    Since, for every $i\in[d]$,  
    $p_{i,2}+p_{i,3} =    1-p_{i,1}\geq 1-U$, we have
    $\bar{p}_2 + \bar{p}_3 \geq 1-U$, so 
    $$   
    \big(1-\hat\beta \bar{p}_2\big)^d+\big(1-\hat\beta \bar{p}_3\big)^d 
    \leq \big(1-\hat\beta L\big)^d+\big(1-\hat\beta (1-U-L)\big)^d.
    $$
    Plugging this into \eqref{eq:4g4g6gg} and then into \eqref{eq:rvrveetr2}, we obtain that 
    \begin{equation*}
    R \leq \frac{\big(1-\hat\beta L\big)^d+\big(1-\hat\beta (1-U-L)\big)^d}{\big(1-\hat\beta U\big)^d}.
    \end{equation*}
    Therefore, using \eqref{eq:btb5y56}, we obtain the lower bound $ p\geq f_\ell(d,\beta, U, L)$.

    Next, we show that $ p \leq f_u(d,\beta, U, L)$.    To give an upper bound on $ p$, it suffices to lower bound $R$. Since $p_{i,1}\geq L$ for every $i\in [d]$, we obtain the lower bound
    \begin{equation}\label{eq:rvrv342}
    R \geq \frac{\sum_{c'=2}^3\prod_{i=1}^d\big(1-\hat\beta p_{i,c'}\big)}{\big(1-\hat\beta L\big)^d}.
    \end{equation}
    Using the arithmetic-mean geometric-mean 
    inequality we have  
    \begin{equation}\label{eq:wds134}
    \sum_{c'=2}^3\prod_{i=1}^d\big(1-\hat\beta p_{i,c'}\big)\geq 2\prod^d_{i=1}\Big(\big(1-\hat\beta p_{i,2}\big)\big(1-\hat\beta p_{i,3}\big)\Big)^{1/2}.
    \end{equation}
    
    Now let  $f(x) =   -\ln(1-\hat\beta x\big)$ and consider an arbitrary $i\in [d]$. Let $a=-1, b=1, \rho=p_{i,2}+p_{i,3}$ and consider the interval $I=[a,b]$. 
    Since $f$ is convex on $I$ and $\rho\in [a+b,2b]$, Item~\ref{it:inc2} of Lemma~\ref{lem:ferfe}
    implies that the function $g(x) = f(x) + f(\rho-x)$ is increasing on
    the interval $J=[\rho/2,b]=[\rho/2,1]$.
    Let $x= U$ and $y = \max\{p_{i,2},p_{i,3}\}$.
    Since $p_{i,2}$ and $p_{i,3}$ are at most $U$,
    $x$ and $y$ are in $J$ and satisfy $x\geq y$.
    Therefore, $g(x)\geq g(y)$, which gives that
    $$
    \ln\big(1-\hat\beta y\big) + \ln\big(1-\hat\beta (\rho-y)\big) \geq  \ln\big(1-\hat\beta x\big) + \ln\big(1-\hat\beta (\rho-x)\big).$$ 
    Thus, by substituting in the values of~$x$, $y$ and $\rho$ and exponentiating, we have
    \begin{equation*}
    \big(1-\hat\beta p_{i,2}\big)\big(1-\hat\beta p_{i,3}\big)
    \geq \big(1-\hat\beta U\big)\big(1-\hat\beta (p_{i,2}+p_{i,3}-U)\big).
    \end{equation*}
    Using the inequality $p_{i,2}+p_{i,3} = 1-p_{i,1}\leq 1-L$  in the right-hand side, we get
    \begin{equation*}
    \big(1-\hat\beta p_{i,2}\big)\big(1-\hat\beta p_{i,3}\big)
    \geq \big(1-\hat\beta U\big)\big(1-\hat\beta (1-U-L)\big).
    \end{equation*}
    Plugging this into \eqref{eq:wds134} for each $i\in[d]$ and then into \eqref{eq:rvrv342}, we obtain that 
    \[R \geq \frac{2\big(1-\hat\beta U\big)^{d/2}\big(1-\hat\beta(1-U-L)\big)^{d/2}}{\big(1-\hat\beta L\big)^d}.\]
    Therefore, using \eqref{eq:btb5y56}, we obtain the upper bound $ p \leq f_u(d,\beta, U, L)$. 
\end{proof}

\subsection{Properties of the functions $f_u$ and $f_\ell$}
In this section, we establish useful monotonicity properties of the functions $f_u$ and $f_\ell$ that will be relevant later. 
 
\begin{lemma}\label{lem:fbetamono}
For any fixed $d \geq 2$
and  any fixed $x$ and $y$ satisfying $0\leq y \leq x \leq 1$ and
$2y + x \leq 1 \leq 2x + y$,  
    \begin{itemize}
        \item   $f_u(d,\beta,x,y)$ is a decreasing function of $\beta$ on the interval  $(0, 1)$ and
        \item $f_\ell(d, \beta,x,y)$ is an increasing function of $\beta$ on the interval $(0,1)$.
    \end{itemize}
\end{lemma}
\begin{proof}
 Let $\hat{\beta} := 1- \beta$,
 and 
 $W = 1 - 3 y + \hat\beta  \big((3 y - 1) + (1 - x - 2 y) + x (2 x + y - 1) +y (x - y)\big)$. 
The derivative of $f_u$ with respect to~$\beta$ is given by
 $$  \pderiv{f_u}{\beta} = -f_u^2\frac{d(1-\hat\beta x)^{d/2}(1-\hat\beta (1-x-y))^{d/2} {W}\strut }
 {\strut (1-\hat\beta y)^{d+1}(1-\hat\beta x)(1-\hat\beta (1-x-y))}.$$
(Obviously, this can be checked directly, but the reader may prefer to use the Mathematica code in
Section~\ref{app:fbetamono} to check this and the derivative of $f_\ell$ with respect to~$\beta$, which appears below.)
Using the conditions on~$x$ and~$y$   in the statement of the lemma,
we find that
$$ W \geq 1 - 3 y + \hat\beta (3y-1)=\beta(1-3y)\geq 0.$$
We conclude that $ \pderiv{f_u}{\beta}\leq 0$  so
 $f_u(d,\beta,x,y)$ is a decreasing function of $\beta$ on the interval $[0, 1]$.
 Similarly, 
 \begin{equation*}
    \pderiv{f_\ell}{\beta}= f_\ell^2\frac{d  \left( (2 x+y-1) (1-\hat\beta  (1-x-y))^{d-1}+(x-y) (1-\hat\beta  y)^{d-1} \right)}{(1-\hat\beta x)^{d+1}}\geq 0,
    \end{equation*}
    so   $f_\ell(d, \beta,x,y)$ is an increasing function of $\beta$ on the interval $(0,1)$.\end{proof}

\begin{lemma}\label{lem:flxymono}
    For any fixed $d\geq 2$ and $0<\beta\leq 1$,  
    \begin{enumerate}
        \item \label{it:weq51} $f_\ell(d,\beta,x,y)$ is a  decreasing function of $x$ when $x,y\in [0,1]$, and
        \item \label{it:weq52} $f_\ell(d, \beta, x,y)$ is an increasing function of $y$ when $x+2y\leq 1$ and $x,y\in [0,1]$.
    \end{enumerate}
\end{lemma}
\begin{proof}
    Let
    \[R :=     \frac{\big(1-(1-\beta) y\big)^d+\big(1-(1-\beta)(1-x-y)\big)^d}{\big(1-(1-\beta)x\big)^d}, \mbox{ so that } f_\ell(d,\beta,x,y)= \frac{1}{1+R}.\]
    
 We first prove Item~\ref{it:weq52}. Let $a=0$, $b=1$, $\rho=1-x$ 
 and consider the interval $I=[a,b]$. Since $f(y) = (1-(1-\beta) y)^d$ is convex on $I$ and $\rho\in [2a,a+b]$, Item~\ref{it:dec1} of Lemma~\ref{lem:ferfe} yields that $g(y) = f(y) + f(\rho-y)$ is decreasing on
$J = [a,\rho/2]=[0,\rho/2]$. It follows that, for fixed~$x$, $R$ is  a decreasing function of $y$ on $J$ 
and therefore $f_\ell$  is increasing in~$y$. It remains to observe that, for $x\in [0,1]$, the condition $y\in J$ is equivalent to the condition $x+2y\leq 1$ and $y\in [0,1]$ in the statement. 
    
    For Item~\ref{it:weq51}, note that $1-(1-\beta)(1-x-y)$ is an increasing nonnegative function of $x$ and $1-(1-\beta)x$ is a  decreasing nonnegative function of $x$, so $R$ is an increasing function of $x$. Thus, $f_\ell$ is a  decreasing function of $x$.
\end{proof}

\begin{lemma}\label{lem:fuxymono}
    For any fixed $d\geq 2$ and $0<\beta\leq 1$,  
    \begin{enumerate}
        \item \label{it:we14} $f_u(d,\beta,x,y)$ is an increasing function of $x$ when $1 \leq 2x+y$ and $x,y\in[0,1]$, and
        \item \label{it:we15} $f_u(d, \beta, x,y)$ is a  decreasing function of $y$ when $x,y\in [0,1]$.
    \end{enumerate}
\end{lemma}
\begin{proof}
    The proof is analogous to that of Lemma~\ref{lem:flxymono}. Let
    \[R :=     \frac{2\big(1-(1-\beta) x\big)^{d/2}\big(1-(1-\beta)(1-x-y)\big)^{d/2}}{\big(1-(1-\beta)y\big)^d}, \mbox{ so that } f_u(d,\beta,x,y)= \frac{1}{1+R}.\]
    We first prove Item~\ref{it:we14}.  Let $a=-y$, $b=1$, $\rho=1-y$ and consider the interval $I=[a,b]$. Since the function 
    $f(x)=-\ln\big(1-(1-\beta)x\big)$ is convex on the interval $I$ and $\rho\in [a+b,2b]$, by Item~\ref{it:inc2} of Lemma~\ref{lem:ferfe}, the function $g(x)=f(x)+f(\rho-x)$ is increasing on the interval $J=[\rho/2,b]$. It follows that the function 
    \[\exp(-g(x))=\big(1-(1-\beta) x\big)\big(1-(1-\beta)(1-x-y)\big)\] is a decreasing function of $x$ on $J$, and therefore  $R$ has the same property as well. Thus, $f_u$ is an increasing function of $x$ on the interval $J$. It remains to observe that, for $y\in [0,1]$, the condition $x\in J$ is equivalent to the condition $1 \leq 2x+y$ and $x\in [0,1]$ in the statement. 
    
    For Item~\ref{it:we15}, note that $1-(1-\beta)(1-x-y)$ is an  increasing nonnegative function of $y$ and $1-(1-\beta)y$ is a  decreasing nonnegative function of $y$, so $R$ is an increasing function of $y$. Thus, $f_u$ is a  decreasing function of $y$.
\end{proof}

\begin{lemma}\label{lem:fimgrange}
    For any fixed $d \geq 2$, $0 < \beta < 1$ and $0 \leq y\leq x \leq 1$ such that $2y+x \leq 1 \leq 2x + y$, we have\[2f_\ell(d,\beta,x,y) + f_u(d,\beta,x,y) \leq1\leq  2f_u(d,\beta,x,y)+f_\ell(d,\beta,x,y).\]
\end{lemma}

\begin{proof}
    Since $2y+x \leq 1 \leq 2x+y$, we obtain that $y\leq 1-x-y\leq x$. Further, by the AM-GM inequality,
    \[\big(1-(1-\beta) x\big)^d+\big(1-(1-\beta)(1-x-y)\big)^d\geq 2\big(1-(1-\beta)
     x\big)^{d/2}\big(1-(1-\beta)(1-x-y)\big)^{d/2}.\]
    So
    the denominator in the definition of $f_\ell(d,\beta,x,y)$ is at least as big as the denominator in the definition of
    $f_u(d,\beta,x,y)$. We conclude that
    \[2f_\ell+f_u \leq\frac{2\big(1-(1-\beta)x\big)^d+\big(1-(1-\beta)y\big)^d}{\big(1-(1-\beta) y\big)^d + 2\big(1-(1-\beta) x\big)^{d/2}\big(1-(1-\beta)(1-x-y)\big)^{d/2}} \leq 1,\]
    and
    \[2f_u + f_\ell \geq \frac{2\big(1-(1-\beta) y\big)^d+\big(1-(1-\beta) x\big)^d}{\big(1-(1-\beta) x\big)^d + \big(1-(1-\beta) y\big)^d + \big(1-(1-\beta)(1-x-y)\big)^d}\geq 1.\qedhere\]
\end{proof}

\begin{lemma}\label{lem:fulxymono}
    For any fixed $d \geq 2$, $0 < \beta < 1$ and $0 \leq y_1 \leq y_2 \leq x_2 \leq x_1 \leq 1$ such that $2y_1+x_1 \leq 1 \leq 2x_1 + y_1$ and $2y_2+x_2 \leq 1 \leq 2x_2 + y_2$, we have
    \[f_u(d,\beta,x_2,y_2)\leq f_u(d,\beta,x_1,y_1) \mbox{ \ \ and \ \ } f_\ell(d,\beta,x_2,y_2)\geq f_\ell(d,\beta,x_1, y_1).\]
\end{lemma}
\begin{proof}
 Using the assumptions in the statement of the lemma,  we obtain  
 $    
    1\leq 2x_1+y_2$ and   $2y_1+x_2\leq 1$.
    Therefore, by Lemmas~\ref{lem:flxymono}~and~\ref{lem:fuxymono}, we obtain  
    \begin{align*}
    &f_u(d,\beta, x_2,y_2) \leq f_u(d,\beta,x_1, y_2) \leq f_u(d,\beta,x_1,y_1), \text{ and}\\
    &f_\ell(d,\beta,x_2,y_2) \geq f_\ell(d,\beta, x_2,y_1) \geq f_\ell(d, \beta,x_1,y_1).\qedhere
    \end{align*}
\end{proof}

\subsection{Bounding the marginal probability at the root by two sequences}
For any $\beta > 0$ and $d \geq 2$, we define two sequences:
\[
\begin{cases}
u_0(d,\beta) = 1,\cr
\ell_0(d, \beta) = 0,
\end{cases}\]
and for every non-negative integer~$n$,
\begin{equation}\label{def:unelln}
\begin{cases}
u_{n+1}(d,\beta) = f_u(d,\beta,u_n(d,\beta), \ell_n(d,\beta)), \text{and}\cr
\ell_{n+1}(d,\beta) = f_\ell(d,\beta, u_n(d,\beta),\ell_n(d, \beta)).
\end{cases}
\end{equation}

Our interest in the sequences $u_n(d,\beta)$ and $\ell_n(d, \beta)$ is that they  give upper and lower bounds on the probability 
$\Pr_{\thetree}[\sigma(\theroot) = c]$, respectively (subject to any boundary configuration at the leaves).
\begin{lemma}\label{lem:probbound}
    Suppose that 
$q=3$, 
    $d\geq 2$, and $\beta\in(0,1)$. 
        For any $n\geq 0$, for the $d$-ary tree $\thetree$ with depth $n$ and root $\theroot$, for any configuration $\tau\colon\theleaves \to [q]$ on the leaves and any 
    colour $c\in[q]$, it holds that
    \[\ell_n(d,\beta)\leq\Pr_{\thetree}[\sigma(\theroot)=c\mid \sigma(\theleaves)=\tau]\leq u_n(d,\beta).\]
\end{lemma}
\begin{proof}
We prove the lemma by induction on $n$. For the base case $n=0$, note that $\thetree$ has a single vertex.
    Thus, for every $c\in[q]$ and every $\tau$ assigning a colour to this vertex,
    \[\ell_0(d,\beta)=0\leq\Pr_{\thetree}[\sigma(\theroot)=c\mid \sigma(\theleaves)=\tau]\leq 1= u_0(d,\beta).\] 
    
    For the inductive step, suppose $n>0$.
    For convenience, denote by $T$ the tree $\thetree$, by $v$ the root $\theroot$ and by $\Lambda$ the leaves $\theleaves$.
    Let $v_1,\ldots,v_d$ be the children of $v$ in~$T$ and let $\Lambda_i = \Lambda_{T[v_i]}$ denote the set of leaves of the subtree $T[v_i]$.
    Consider any configuration $\tau \colon \Lambda \to [q]$ and any colour~$c$.  By the induction hypothesis, for every $i \in[d]$ and $j \in [q]$, we have
    \[\ell_{n-1}(d,\beta)\leq\Pr_{T[v_i]}[\sigma(v_i)=j\mid \sigma(\Lambda_i) = \tau(\Lambda_i)] \leq u_{n-1}(d,\beta).\] 
		By  Lemma~\ref{lem:boundbyonesteprecursion} and~\eqref{def:unelln}, we conclude that $\ell_n(d,\beta)\leq   \Pr_{T}[\sigma(v)=c\mid \sigma(\Lambda)=\tau]   \leq u_n(d,\beta)$.
\end{proof}

The following lemma will be used to show that the sequences $u_n(d,\beta)$ and $\ell_n(d, \beta)$ converge. 
\begin{lemma}\label{lem:seqxymono}
    For any fixed $d\geq 2$, $0<\beta < 1$ and $n \in \bn$, we have
    \begin{enumerate} 
        \item \label{it:ununp1} $u_n(d,\beta) \geq u_{n+1}(d,\beta)$,
        \item \label{it:lnlnp1} $\ell_n(d, \beta) \leq \ell_{n+1}(d,\beta)$, and
        \item \label{it:unln1} $2\ell_n(d,\beta) + u_n(d,\beta) \leq 1 \leq 2u_n(d,\beta) + \ell_n(d,\beta)$.
    \end{enumerate}
\end{lemma}
\begin{proof}
 We prove this lemma by induction on~$n$. 
Since $d$ and $\beta$ are fixed, we simplify the notation by
writing~$u_n$ for $u_n(d,\beta)$, writing~$\ell_n$ for $\ell_n(d,\beta)$,
writing~$f_u(x,y)$ for $f_u(d,\beta,x,y)$ and writing $f_\ell(x,y)$ for $f_\ell(d,\beta,x,y)$.

For the base case $n=0$,  we have $u_0=1$ and $\ell_0=0$, so  Item~\ref{it:unln1} holds  since
$2\ell_0  + u_0= 1 < 2 u_0 +  \ell_0$. Items~\ref{it:ununp1} and~\ref{it:lnlnp1} follow from 
\[u_{1}  = f_u( 1,0)=\frac{1}{1+ 2\beta^{d/2}} < 1 = u_0, \mbox{\  and  \ } \ell_{ 1}  = f_\ell( 1,0) = \frac{\beta^d}{\beta^d + 2} > 0 = \ell_0.\]

 For the inductive step, suppose $n > 0$. 
 Item~\ref{it:unln1} follows (using the induction hypothesis)
 from    Lemma~\ref{lem:fimgrange} 
   with $x=u_{n-1}$ and $y=\ell_{n-1}$. 
      We now obtain Items~\ref{it:ununp1} and~\ref{it:lnlnp1}.
    By the induction hypothesis,
    \[0\leq \ell_{n-1}  \leq\ell_n \leq u_n  \leq u_{n-1} \leq 1, \mbox{\  and  \ } 2\ell_{n-1}  + u_{n-1}  \leq 1 \leq 2u_{n-1}  + \ell_{n-1}.\]
    Using Lemma~\ref{lem:fulxymono} (with these facts and with item~3), we obtain
    $$
    u_{n+1}  = f_u( u_n , \ell_n )
    \leq f_u( u_{n-1} , \ell_{n-1} )
    = u_n,$$
		proving Item~\ref{it:ununp1}. 
Similarly, we also obtain that $\ell_{n+1}  = f_\ell( u_n, \ell_n )
    \geq f_\ell( u_{n-1}, \ell_{n-1})
    = \ell_n$, proving Item~\ref{it:lnlnp1}.
\end{proof}

By Lemma~\ref{lem:seqxymono}, we have that the sequences $\{u_n(d,\beta\})$ and $\{\ell_n(d,\beta)\}$ are bounded and monotonic, so they both converge. Let
\begin{equation}\label{eq:tb53tb5b3fwefe}
u_\infty(d,\beta):=\lim_{n\to\infty} u_n(d,\beta), \mbox{\  and \ } \ell_\infty(d,\beta):=\lim_{n\to\infty} \ell_n(d, \beta).
\end{equation} 
We have the following characterisation of the limits $u_\infty(d,\beta), \ell_\infty(d,\beta)$.
\begin{lemma}\label{lem:fixedpoints}
    For any $d\geq 2$ and $0 < \beta \leq 1$, $(x,y)=(u_\infty(d, \beta), \ell_\infty(d,\beta))$ is a solution to the 
    system of equations
    \[
    \begin{cases}
    f_u(d,\beta,x,y) = x\cr
    f_\ell(d,\beta,x,y)=y
    \end{cases}
    \]
    satisfying $0<y\leq 1-x-y \leq x < 1$.
\end{lemma}
\begin{proof}
Since $d$ and $\beta$ are fixed, we simplify the notation by
writing $u_\infty$ for $u(d,\beta)$ and $\ell_\infty$ for $\ell(d,\beta)$.
We also drop $d$ and $\beta$ as parameters of $u_n$, $\ell_n$, $f_u$ and $f_\ell$ 
(as in the proof of Lemma~\ref{lem:seqxymono}). 
By Lemma~\ref{lem:seqxymono}, we have
$\ell_\infty  \geq \ell_1  =  {\beta^d}/(\beta^d+2) > 0$ and
 $u_\infty  \leq u_1  = 1/(1+2\beta^{d/2}) < 1$.
 Also,  for every non-negative integer~$n$, we have
    $
    \ell_n    \leq 1 
    - u_n  - \ell_n 
    \leq u_n$,
    which  implies, by applying limits,
    that 
    $ \ell_\infty \leq 1-u_\infty  - \ell_\infty  \leq u_\infty$.

    Recall that, for $n\geq 0$,
    $u_{n+1} =f_u( u_n ,\ell_n )$
    and 
    $\ell_{n+1} = f_\ell( u_n ,\ell_n )$.
    Using these definitions and
    the continuity of the functions $f_u( x,y)$ and $f_\ell( x,y)$
    with respect to~$x$ and~$y$ (in the third equality below), we have
    $$  
    u_\infty  = \lim_{n\rightarrow \infty} u_n  
    = \lim_{n\rightarrow \infty} f_u( u_{n-1} ,\ell_{n-1})
    = f_u( \lim_{n\rightarrow \infty} u_{n-1},\lim_{n\rightarrow\infty} \ell_{n-1})
    = f_u(u_\infty,\ell_\infty).
 $$
    Similarly, 
    $\ell_\infty  = f_\ell(u_\infty,\ell_\infty)$.\end{proof}

\subsection{Bounding the maximum ratio}
In this section, we place the final pieces for the proof of Lemma~\ref{lem:onestepuniq}. The first lemma accounts for the $d=2$ case of Lemma~\ref{lem:onestepuniq}. 
\begin{lemma}\label{lem:d=2bound}
    Suppose $q = 3$ and  $\beta\in(0,1)$.    
    Then there  is a positive integer $n_0$ such that for every $n \geq n_0$, every $c \in [q]$, 
    and every configuration $\tau\colon  \theleavestwo \to [q]$,
\[\frac{459}{2000}\leqslant \Pr_{\thetreetwo}[\sigma(\theroottwo)=c\mid \sigma(\theleavestwo) = \tau] \leqslant \frac{1107}{2500}.\]
\end{lemma}

\begin{proof}
    For any $0 < \beta \leq 1$, by Lemma~\ref{lem:fixedpoints}, $(x, y) = (u_\infty(2,\beta), \ell_\infty(d,\beta))$ is a solution to the  system of equations
    \begin{equation}\label{eq:d=2fixedpoints}
    \begin{cases}
    f_u(2,\beta,x,y) = x\cr
    f_\ell(2,\beta,x,y)=y
    \end{cases}
    \end{equation}
     satisfying $0<y\leq 1-x-y \leq x < 1$. In Section~\ref{app:d=2bound}, we use the Resolve function of Mathematica to show
rigorously that there is no solution to~\eqref{eq:d=2fixedpoints}
satisfying      
    \[0 < y \leq \frac{1}{3}  \text{ and } \frac{1106}{2500} \leq x < 1\]
and there is no solution satisfying
    \[0 < y \leq \frac{460}{2000}  \text{ and } \frac{1}{3} \leq x < 1.\]

If   $0<y\leq 1-x-y \leq x < 1$ then $0<y\leq1/3$ and $1/3 \leq x < 1$.
So any solution to~\eqref{eq:d=2fixedpoints}
which satisfies $0<y\leq 1-x-y \leq x < 1$
must also satisfy 
$ y> {460}/{2000}$
and $x <  {1106}/{2500}$.
We conclude that 
$ \ell_\infty(d,\beta)> {460}/{2000}$
and $u_\infty(d,\beta) <  {1106}/{2500}$.
     
Since $\ell_\infty(2,\beta)$ and $u_\infty(2,\beta)$ are the limits of the sequences
$\ell_n(2,\beta)$ and $u_n(2,\beta)$, respectively,
there is a positive integer~$n_0$
such that, for all $n\geq n_0$,
$ \ell_n(2,\beta)\geq {459}/{2000}$
and $u_n(2,\beta) \leq  {1107}/{2500}$.
Thus by Lemma~\ref{lem:probbound}, 
for every $n\geq n_0$ and every $\tau\colon{\theleavestwo} \to [3]$, 
    \[\frac{459}{2000}\leq\ell_n(2,\beta)\leq \Pr_{\thetreetwo}[\sigma(\theroottwo)=c\mid \sigma(\theleavestwo) = \tau] \leq
    u_n(2,\beta) \leq
    \frac{1107}{2500}.\qedhere\]
    \end{proof}

\begin{definition}
For any $d \geq 3$,  
define the critical parameter $\beta_*(d)$ by $\beta_*(d) = 1 - \displaystyle\frac{3}{d+1}$.
\end{definition}

Note that $\beta_*(d)>0$.
The following lemma 
shows that $u_n(d,\beta)$ and $\ell_n(d,\beta)$ are
bounded by the values corresponding to the critical parameter.

\begin{lemma}\label{lem:seqbetamono}
Fix any $d\geq 3$.
For any $\beta$  in the range $\beta_*(d)\leq \beta<1$ 
and any non-negative integer~$n$,  we have
$u_n(d,\beta_*(d))\geq u_n(d, \beta)$ and  $\ell_n(d,\beta_*(d)) \leq \ell_n(d,\beta)$.
\end{lemma}
\begin{proof}
We prove the lemma by induction on~$n$. 
Since $d$ is fixed, we simplify the notation by writing $\beta_*$ for $\beta_*(d)$.
We also drop 
the argument~$d$ from $u_n(d,\beta)$, $\ell_n(d,\beta)$, $f_u(d,\beta,x,y)$ and $f_\ell(d,\beta,x,y)$.
 
For the base case $n=0$, note that for every~$\beta$, it holds that $u_0(\beta) = u_0(\beta_*)=1$ and
$\ell_0(\beta) = \ell_0(\beta_*)=0$.

For the inductive step, suppose $n>0$. 
By Lemma~\ref{lem:seqxymono}, we have
\begin{align*}
&2\ell_{n-1}(\beta)+u_{n-1}(\beta)\leq 1 \leq 2u_{n-1}(\beta)+\ell_{n-1}(\beta), \text{ and}\\
&2\ell_{n-1}(\beta_*)+u_{n-1}(\beta_*)\leq 1 \leq 2u_{n-1}(\beta_*)+\ell_{n-1}(\beta_*).
\end{align*}
By Lemma~\ref{lem:seqxymono} and 
  the induction hypothesis, we have   
    \[0\leq \ell_{n-1}(\beta_*)\leq \ell_{n-1}(\beta) \leq u_{n-1}(\beta) \leq u_{n-1}(\beta_*)\leq 1.\]
Now, using the definitions and  Lemma~\ref{lem:fbetamono}, we get 
$$
        u_n(\beta_*) =f_u(\beta_*,u_{n-1}(\beta_*),\ell_{n-1}(\beta_*))
        \geq f_u(\beta,u_{n-1}(\beta_*),\ell_{n-1}(\beta_*)).$$
Then using Lemma~\ref{lem:fulxymono}, we continue with
$$f_u(\beta,u_{n-1}(\beta_*),\ell_{n-1}(\beta_*))
        \geq
        f_u(\beta,u_{n-1}(\beta),\ell_{n-1}(\beta))
        =u_n(\beta).
$$
 Similarly, again using Lemma~\ref{lem:fbetamono} and then Lemma~\ref{lem:fulxymono}, we have that
$$f_\ell(\beta_*,u_{n-1}(\beta_*),\ell_{n-1}(\beta_*))
    \leq f_\ell(\beta,u_{n-1}(\beta_*),\ell_{n-1}(\beta_*))
    \leq
    f_\ell(\beta,u_{n-1}(\beta),\ell_{n-1}(\beta)),$$ 
which gives that $\ell_n(\beta_*)\leq \ell_n(\beta)$.
\end{proof}

Our next goal is to prove Lemma~\ref{lem:ffixedpoints} below,
which will help us to obtain an upper bound on the ratio $u_\infty(d,\beta_*(d))/\ell_\infty(d,\beta_*(d))$ when $d$ is sufficiently large. In order to do this,
we first define some  useful
re-parameterisations of $f_u$ and $f_\ell$, and establish some properties of these.

\begin{definition}\label{def:gs}
Let 
    $g_u(d,\mu,y)=f_u(d, \beta_*(d), \mu\cdot y, y)-\mu\cdot y$ and
    $g_\ell(d, \mu, y)= f_\ell(d, \beta_*(d), \mu\cdot y, y) - y$.
\end{definition}
Note that the argument $\mu$ in $g_u,g_\ell$ corresponds to the ratio $x/y$  of the arguments $x,y$ of $f_u,f_\ell$.
 
   \begin{lemma}\label{lem:guymono}
For every  
$d\geq 5$
and $\mu \geq 1$,   $g_u(d,\mu,y)$ is a decreasing function of $y$ 
in the range
$1/(2\mu+1) \leq y \leq 1/(\mu+2)$.
\end{lemma}

\begin{proof}
Let  
$A={(1-{3 y}/{(d+1)})}^d$ and
\[ B = \left(1-\frac{3(1-y(\mu+1))}{d+1}\right)^{d/2} \left(1-\frac{3 \mu y}{d+1}\right)^{d/2}.\]
Since 
$y\leq 1/(\mu+2) \leq 1/3$ and $d\geq 5$,  we 
have $3y<d+1$, so $A>0$.
Also,   $3\mu y< d+1$
and $3(1-y(\mu+1))<d+1$, so $B>0$.
Let $W = A B / {(A+2B)}^2>0$.
The derivative of $g_u$ with respect to $y$ 
(see Section~\ref{app:guymono} for Mathematica assistance)
is given by the following.
$$
\pderiv{g_u}{y} =-\mu+\left(\frac{9 d W}{3y(\mu+1)+d-2}\right)\left(\frac{\mu(2\mu y+y-1)}{1+d-3\mu y}-\frac{2\mu y+y+d-1}{1+d-3y}\right).$$
The upper bound on $y$
yields (crudely) $2\mu y+y-1< 1$ and $1+d-3\mu y > d-2$.
Similarly, the lower bound on $y$
yields   $2\mu y+y+d-1\geq d$ and  (since $y$ is non-negative)
$1+d-3y < 1+d$.
Plugging these in, we obtain
\begin{equation}\label{eq:evref344346}
\pderiv{g_u}{y}  <-\mu +\left(\frac{9 d W}{3 y(\mu +1)+d-2}\right)\left(\frac{\mu }{d-2}-\frac{d}{1+d}\right).
\end{equation}
Note that 
$3y(\mu +1)> 0$  and $d> 2$ so 
the expression
$3y(\mu +1)+d-2$ in the denominator is positive.
We now consider two cases.

If $\mu /(d-2)-d/(1+d)\leq0$
then recall that $W>0$, so  \eqref{eq:evref344346} gives that $\pderiv{g_u}{y}<-\mu  < 0$.
 
Otherwise,  $\mu /(d-2)-d/(1+d)>0$.
In this case, note that, for any $A$ and $B$,  
$(A-2B)^2 \geq 0$, so 
$8 A B \leq A^2 + 4 AB + 4B^2=(A+2B)^2$.
From the definition of~$W$, this ensures that $W\leq 1/8$.
So, \eqref{eq:evref344346} gives that
\[\pderiv{g_u}{y}< -\mu +\frac{\frac{9 d}{8} \left(\frac{\mu }{d-2}-\frac{d}{1+d}\right)}{3 y (\mu +1) +d-2} <
-\mu  + \frac{9d}{8(d-2)}\left(\frac{\mu }{d-2}-\frac{d}{1+d}\right)<\mu \left(\frac{9d}{8(d-2)^2}-1\right)<0,\]
where the final inequality uses $d\geq 5$.  
\end{proof}

The following lemma is analogous to Lemma~\ref{lem:guymono}, but for the function~$g_\ell$.
 
\begin{lemma}\label{lem:glymono}
For every  
$d\geq 3$ and $\mu \geq 1$,
$g_\ell(d,\mu ,y)$ is a decreasing function of $y$ 
in the range   $1/(2\mu +1) \leq y \leq 1/(\mu +2)$.
\end{lemma}

\begin{proof}
Let 
\[W = \frac{(\mu  (2 d-1)+d+1)\big(\frac{3 y(\mu +1) +d-2}{d+1}\big)^d}{(1 + d - 3 \mu  y) (d - 2 + 3 y (1 + \mu ) )} +\frac{(\mu -1) (d+1)\big(1-\frac{3 y}{d+1}\big)^d}{(1 + d - 3 y) (1 + d - 3 \mu  y)}.\]
Since $\mu \geq 1$ and $d\geq 3$
and $y\leq 1/(\mu +2)$ all of the factors in $W$ are  positive, so $W> 0$.
The derivative of $g_\ell$ with respect to $y$ 
(see Section~\ref{app:glymono} for Mathematica assistance)
is given by the following.
\[
\pderiv{g_\ell}{y} = - \frac{3 d \big(1-\frac{3 \mu  y}{d+1}\big)^d W}{\Big(\big(\frac{3 (\mu +1) y+d-2}{d+1}\big)^d+\big(1-\frac{3 \mu  y}{d+1}\big)^d+\big(1-\frac{3 y}{d+1}\big)^d\Big)^2}-1.
\]
We've already seen that $W>0$ and the denominator is greater than~$0$ since it is a square.
Since $3\mu y < 3<d+1$, the remaining term is also positive, so 
$\pderiv{g_\ell}{y}<0$, as required.
\end{proof}

Next, we will identify a value $y_\mu $ so that, when $\mu$ 
and $d$ are sufficiently large, 
$g_u(d,\mu ,y_\mu )<0$ and $g_\ell(d,\mu ,y_\mu )>0$.

\begin{definition} \label{def:ymu}
Define the quantity $y_\mu $ as follows. 
\[y_\mu =
\begin{cases}
\frac{7}{10\mu +12}+\frac{3}{500} & \mbox{if $\mu  < 32$,}\cr
\frac{7}{10\mu +12} & \mbox{if $\mu  \geqslant 32$}.
\end{cases}
\]
Let $x_\mu = \mu y_\mu$.
Now define the functions $h_u$ and $h_\ell$ as
$h_u(d, \mu ) = g_u(d,\mu , y_\mu )$
and   $h_\ell(d, \mu ) = g_\ell(d,\mu , y_\mu )$.
\end{definition}
Then we have the following lemmas.

\begin{lemma}
\label{lem:ineqs}
If $\mu\geq 157/80$ then
$0 < y_\mu < 1-x_\mu-y_\mu < \tfrac13 < x_\mu < 1-y_\mu$.
\end{lemma}
 
\begin{proof} 
The inequalities follow directly from Definition~\ref{def:ymu}.
Mathematica code is given in Section~\ref{sec:ineqs}.
\end{proof}

\begin{lemma}\label{lem:dhu<0}
Suppose $d \geq 23$.
If $157/80\leq \mu< 32$ or $32<\mu$ then $\pderiv{h_u}{\mu} < 0$.
\end{lemma}
\begin{proof}
Since $d$ is fixed in the proof of this lemma, we will drop it as an argument of
$\beta_*$, $h_u$, $g_u$.
We will use $\hat\beta_*$ to denote $1-\beta_*=3/(d+1)$.
We will drop $d$ and $\beta_*$ as an argument of $f_u$.
So, plugging in Definitions~\ref{def:gs} and~\ref{def:ymu},
we get
$h_u(\mu) = g_u(\mu,y_\mu) = f_u(x_\mu, y_\mu)- x_\mu$.
We have
\begin{equation}\label{eq:parexpand}
\pderiv{h_u(\mu)}{\mu}
=\pderiv{f_u(x_\mu,y_\mu)}{x_\mu}\cdot\pderiv{x_\mu}{\mu}+\pderiv{f_u(x_\mu,y_\mu)}{y_\mu}\cdot\pderiv{y_\mu}{\mu}-\pderiv{x_\mu}{\mu}.\end{equation}
Let 
\[R(x,y)=\frac{(1-\hat\beta_* x)^{d/2}(1-\hat\beta_* (1-x-y))^{d/2}}{(1-\hat\beta_* y)^d}.\]
The derivatives of $f_u(x,y)$ with respect to $x$ and $y$ are as follows
(see Section~\ref{app:lem:dhu<0} for Mathematica assistance).
\begin{equation}\label{mypartial}
\begin{aligned} 
\pderiv{f_u(x,y)}{x}&= \left(\frac{f_u(x, y)^2 R(x,y) d \hat\beta_*}{1-\hat\beta_* (1-x-y)}\right) \left( \frac{\hat\beta_*(2x+y-1)  }
{ 1-\hat\beta_* x}\right), \text{ and}\cr 
\pderiv{f_u(x,y)}{y}&=-\left(\frac{f_u(x, y)^2 R(x,y) d \hat\beta_*}{1-\hat\beta_* (1-x-y)}\right) 
\left( \frac{ 3 + \hat\beta_*  (2 x + y-2)  }{1-\hat\beta_* y  }\right).
 \end{aligned}
\end{equation}
If $0\leq x\leq1$, $0\leq y\leq 1$, $0\leq 1-x-y\leq 1$
and $2x+y>1$  then
 all of the factors are   positive, so 
 by Lemma~\ref{lem:ineqs},
 $\pderiv{f_u(x_\mu,y_\mu)}{x_\mu} > 0$
and $\pderiv{f_u(x_\mu,y_\mu)}{y_\mu} < 0$.
Let 
\begin{equation}\label{defz}
z =-\pderiv{f_u(x_\mu,y_\mu)}{x_\mu}\bigg/\pderiv{f_u(x_\mu,y_\mu)}{y_\mu}
\end{equation}
and note that 
$z$~is positive. Using~\eqref{eq:parexpand} and~\eqref{defz}, we can express 
$\pderiv{h_u(\mu)}{\mu}$
as 
$$ \pderiv{h_u(\mu)}{\mu}=\left(-z\cdot\pderiv{f_u(x_\mu,y_\mu)}{y_\mu}-1\right)\pderiv{x_\mu}{\mu}+\pderiv{f_u(x_\mu,y_\mu)}{y_\mu}\cdot\pderiv{y_\mu}{\mu}.$$
From the definition of $y_\mu$ (Definition~\ref{def:ymu}),
$\pderiv{y_\mu}{\mu}=-\frac{35}{2 (5 \mu+6)^2}<0$ for all $\mu\neq 32$.
If $\mu < 32$ then 
$\pderiv{x_\mu}{\mu} = \frac{21}{(5 \mu+6)^2}+\frac{3}{500}>0$.
If $\mu>32$ then
$\pderiv{x_\mu}{\mu} = \frac{21}{(5 \mu+6)^2}>0$. Thus, to show $\pderiv{h_u(\mu)}{\mu}<0$,
it suffices to show 
\begin{equation}
\label{dec:goal}
-\pderiv{f_u(x_\mu,y_\mu)}{y_\mu}<
\frac{\pderiv{x_\mu}{\mu}\strut}{\strut z\cdot\pderiv{x_\mu}{\mu}-\pderiv{y_\mu}{\mu}}.\end{equation}

We will simplify~\eqref{dec:goal} by finding an upper bound for~$z$. Using \eqref{mypartial} and \eqref{defz}, we have
\[z=\bigg( \frac{1-\hat\beta_*  y_\mu}{1-\hat\beta_*  x_\mu}\bigg)\bigg(\frac{\hat\beta_* (2x_\mu+y_\mu-1)}{3+\hat\beta_* (2x_\mu+y_\mu-2)}\bigg)= \bigg(1+\frac{\hat\beta_*(x_\mu-y_\mu)}{1-\hat\beta_*  x_\mu}\bigg)\bigg(\frac{\hat\beta_* (2x_\mu+y_\mu-1)}{3-\hat\beta_* (2-2x_\mu-y_\mu)}\bigg).\] 
Since, by Lemma~\ref{lem:ineqs}, 
$x_\mu > y_\mu$,
$2 x_u + y_y > 1$,
$x_\mu > 0$ and
(since $x_\mu +y_\mu < 1$)
$2 > 2 x_u + y_u$,
 $z$ is an increasing function of $\hat\beta_*$. 
Since $d\geq 23$, we have $\hat\beta_* = 3/(d+1)\leq 1/8$, so $z$ is upper-bounded by its value with $\hat\beta_*$ replaced by~$1/8$. This gives that
\[z\leq \frac{(8-y_\mu) (2 x_\mu+y_\mu-1)}{(8-x_\mu) (2 x_\mu+y_\mu+22)}.
\]
Moreover, using Mathematica, we show in Appendix~\ref{app:lem:dhu<0} that  
\begin{equation}\label{eq:err346gevrerf53}
\frac{(8-y_\mu) (2 x_\mu+y_\mu-1)}{(8-x_\mu) (2 x_\mu+y_\mu+22)}<\frac{1}{24} \mbox{ for all } \mu> 1.
\end{equation}
 It follows that $z<1/24$. Thus, we can re-write our goal from~\eqref{dec:goal} ---
 to prove the lemma, it suffices to show
 \begin{equation}
\label{dec:newgoal}
-\pderiv{f_u(x_\mu,y_\mu)}{y_\mu}<
\frac{\pderiv{x_\mu}{\mu}\strut}{\strut \tfrac{1}{24}\cdot\pderiv{x_\mu}{\mu}-\pderiv{y_\mu}{\mu}}.\end{equation} 
The definitions of~$f_u$ and~$R$
imply that $f_u(x,y) = 1/(1+2 R(x,y))$. Therefore, using the fact that $\frac{a}{(1 + 2a)^2}\leq 1/8$ for all $a>0$, we have
\[
f_u(x,y)^2 R(x,y)  = 
\frac{R(x,y)}{\big(1 + 2R(x,y)\big)^2}
\leqslant \frac{1}{8}.
\]
So, plugging this into the second equality in~\eqref{mypartial}, 
recalling that $ \pderiv{f_u(x_\mu,y_\mu)}{y_\mu}<0$
and $\beta_*=3/(d+1)$,
we get
$$-\pderiv{f_u(x_\mu,y_\mu)}{y_\mu} 
\leq \frac{d\hat\beta_* (3+\hat\beta_* (2x_\mu+y_\mu-2))}{8(1-\hat\beta_* (1-x_\mu-y_\mu))(1-\hat\beta_* y_\mu)}
\leq \frac{ 3 (3+\hat\beta_* (2x_\mu+y_\mu-2))}{8(1-\hat\beta_* (1-x_\mu-y_\mu))(1-\hat\beta_* y_\mu)}.
$$

Let $Y$ be the right-hand-side of the previous expression. Using $\hat\beta_*\leq 1/8$ and the 
inequalities from Lemma~\ref{lem:ineqs}, we find 
that $Y$ is increasing in~$\hat\beta_*$ (see that Mathematica code in Appendix~\ref{app:lem:dhu<0}).
Thus, we can replace $Y$ with its value with   $\hat\beta_*$ replaced by~$1/8$,
which is 
$ {3(2 x_\mu+y_\mu+22)}/{((8-y_\mu) (x_\mu+y_\mu+7))}$.
Plugging this into~\eqref{dec:newgoal}, it suffices 
to show

\begin{equation}\label{eq:zbound}
\frac{3(2 x_\mu+y_\mu+22)}{(8-y_\mu) (x_\mu+y_\mu+7)}
 <
\frac{\pderiv{x_\mu}{\mu}\strut}{\strut \tfrac{1}{24}\cdot\pderiv{x_\mu}{\mu}-\pderiv{y_\mu}{\mu}}.
\end{equation}
We prove~\eqref{eq:zbound} in two cases.  
 
\noindent{\bf Case 1: $\mu> 32$:\quad}  Using the values $\pderiv{x_\mu}{\mu},\pderiv{y_\mu}{\mu}$ that we calculated earlier,  the right-hand side of~\eqref{eq:zbound} 
is~$8/7$.
The Mathematica code in Appendix~\ref{app:lem:dhu<0}
uses Resolve to show rigorously that there is no $\mu>32$
 satisfying~\eqref{eq:zbound}. \vskip 0.2cm

\noindent{\bf Case 2: $ \mu<32$:\quad} Using the values that we calculated earlier, the right-hand side of~\eqref{eq:zbound} is
$$ 
    \frac{24 \left(25 \mu^2+60 \mu+3536\right)}{25 \mu^2+60 \mu+73536}.
 $$ 
The Mathematica code in Appendix~\ref{app:lem:dhu<0}
uses Resolve to show rigorously that there is no $\mu>1$   
 satisfying~\eqref{eq:zbound}. 
\end{proof}

We will use the following function in several of the remaining lemmas.

\begin{definition}
\label{def:psi}
Let 
$\psi(d, z) = {d}/(d-3 z+1)+\ln \left(d+1-3z\right)$.
\end{definition}

\begin{lemma}\label{lem:muone}
Suppose $d\geq 23$. Then $h_u(d,157/80) < 0$
and $h_u(d,32) < 0$.\end{lemma}

\begin{proof} 
Let $\zeta(d,x,y):=2\psi(d,y)-\psi(d,x)-\psi(d,1-x-y)$, where   $\psi$ is 
the function defined  in Definition~\ref{def:psi}. 
The derivative of $h_u(d,\mu)$ with respect to~$d$ is given 
as follows (see Appendix~\ref{app:muone} for the Mathematica code).
\begin{align}
\pderiv{h_u(d,\mu)}{d} &=
\pderiv{f_u(d,\beta_*(d),x_\mu, y_\mu)}{d}\nonumber\\
\label{eq:Dhud}
&=\frac{\big(1-\frac{3 x_\mu }{d+1}\big)^{\frac{d}{2}} \big(1-\frac{3 y_\mu }{d+1}\big)^d \big(1-\frac{3(1-x_\mu-y_\mu)}{d+1}\big)^{\frac{d}{2}} \zeta(d,x_\mu,y_\mu)}{\Big(\big(1-\frac{3 y_\mu }{d+1}\big)^d+2 \big(1-\frac{3 x_\mu }{d+1}\big)^{\frac{d}{2}} \big(1-\frac{3(1-x_\mu-y_\mu)}{d+1}\big)^{\frac{d}{2}}
 \Big)^2}.
\end{align}

First, fix $\mu=157/80$.
We will prove three facts.
\begin{itemize}
\item{\bf Fact 1:} For all $d\geq 23$, $\pderiv{\zeta(d,x_\mu,y_\mu)}{d} < 0$.
\item{\bf Fact 2: $\lim_{d\rightarrow \infty} \zeta(d,x_\mu,y_\mu)=0$.}
\item{\bf Fact 3: $\lim_{d\rightarrow \infty} h_u(d,\mu)<0$.}
\end{itemize}

Facts~1 and~2 guarantee that, for all $d\geq 23$, $\zeta(d,x_\mu,y_\mu)>0$. 
Lemma~\ref{lem:ineqs} guarantees that all other factors in~\eqref{eq:Dhud} are also positive.
 Thus,
$\pderiv{h_u(d,\mu)}{d}$ is positive for all $d\geq 23$.
Together with Fact~3, this proves 
the first part of the lemma, that 
$h_u(d,157/80) < 0$.
The three facts are proved in the Mathematica code in Section~\ref{app:muone}.

Finally, fix $\mu=32$.
Lemma~\ref{lem:ineqs} again guarantees that all   factors in~\eqref{eq:Dhud} 
other than $\zeta(d,x_\mu,y_\mu)$
are  positive. 
Thus, it suffices to prove the three facts for $\mu=32$, and this is done in
 the Mathematica code in Section~\ref{app:muone}.
\end{proof}

Lemmas~\ref{lem:dhu<0} and~\ref{lem:muone} have the following corollary.

\begin{corollary}\label{lem:hu<0}
For every $d \geq 23$ and $\mu \geq 157/80$, $h_u(d,\mu) <0$.
\end{corollary}
\begin{proof}
By Lemma~\ref{lem:dhu<0},
$h_u(d,\mu)$ is decreasing for $\mu\in[157/80,32)$.
Thus, for $\mu$ in this range, $h_u(d,\mu) \leq h_u(d,157/80)$
and by Lemma~\ref{lem:muone}, $h_u(d,157/80)<0$.

By Lemma~\ref{lem:dhu<0},
$h_u(d,\mu)$ is decreasing for $\mu> 32$.
Thus, for $\mu>32$, $h_(d,\mu) \leq h(d,32)$
and by Lemma~\ref{lem:muone}, $h_u(d,32)<0$.
\end{proof}

\begin{lemma}\label{lem:hl>0}
    For every $d \geq 23$ and $\mu \geq 157/80$, $h_\ell(d,\mu) >0$.
\end{lemma}
\begin{proof}
 We will show that
 \begin{equation}\label{eq:hellmu}
\pderiv{h_\ell(d,\mu)}{d}>0 \mbox{ for all } d\geq 23 \mbox{ and } \mu\geq 157/80.
\end{equation}
The Mathematica code in Appendix~\ref{app:hl>0}
verifies that $h_\ell(23,\mu)>0$ for all $\mu\geq157/80$.
Together with~\eqref{eq:hellmu}, this proves the lemma.

Therefore, in the rest of the proof, we prove~\eqref{eq:hellmu}.
Using Definitions~\ref{def:ymu} and~\ref{def:gs},
we have
$h_\ell(d,\mu) =   f_\ell(d,\beta_*(d) , x_\mu,y_\mu )-y_{\mu}$.
We use the following definitions in order to describe 
$\pderiv{h_\ell(d,\mu)}{d}$.
Recall from Definition~\ref{def:psi} that $\psi(d,z)={d}/(d-3 z+1)+\ln \left(d+1-3z\right)$.
Let $A = (1-\frac{3 x_\mu}{d+1})^d$
and $B = (1-\frac{3 y_\mu}{d+1})^d$
and $C = (1-\frac{3(1-x_\mu-y_\mu)}{d+1})^d$.
Then the derivative of $h_\ell(d,\mu)$
with respect to~$d$ is given as follows 
(see Appendix~\ref{app:hl>0} for Mathematica assistance). 
\begin{equation} 
\label{eq:Dhld}
\pderiv{h_\ell(d,\mu)}{d}=
\pderiv{f_\ell(d,\beta_*(d),x_\mu,y_\mu)}{d}=\frac{
A C(\psi(d,x_\mu)-\psi(d,1-x_\mu-y_\mu))+A B(\psi(d,x_\mu)-\psi(d, y_\mu))
}{(A+B+C)^2}.
\end{equation}

Lemma~\ref{lem:ineqs} guarantees that 
$A$, $B$ and $C$ are positive, so 
to prove~\eqref{eq:hellmu}, and hence the lemma,
it suffices to show 
$\psi(d,x_\mu)>\psi(d,1-x_\mu-y_\mu)$
and
$\psi(d,x_\mu)>\psi(d, y_\mu)$.

Note that  
\[\pderiv{\,\psi\!\left(d, \frac{1}{3}+t\right)}{t} = \frac{9 t}{(d-3 t)^2}, \quad \mbox{ and } \quad \pderiv{\,\psi\!\left(d, \frac{1}{3}-t\right)}{t} = \frac{9 t}{(d+3 t)^2}.\]
Thus, for fixed~$d$, the function $\psi(d,z)$ is decreasing for $z\in [0,1/3]$.
Since, by Lemma~\ref{lem:ineqs}, $0<y_\mu < 1-x_\mu-y_\mu   < 1/3$,
we have
\begin{equation}
\label{eq:first8}
\psi(d,y_\mu)\geq \psi(d,1-x_\mu-y_\mu).
\end{equation}

The function $\psi(d,z)$ is increasing for $z\in[1/3,1]$.
Since, Lemma~\ref{lem:ineqs} guarantees $1/3 < 2/3-y_\mu < x_\mu < 1$,
we have
\begin{equation}
\label{eq:second8}
\psi(d,x_\mu)\geq \psi(d,\tfrac{2}{3} - y_\mu).
\end{equation}

Since the function $\psi(d, \tfrac{1}{3}+t)-\psi(d, \tfrac{1}{3}-t)$ is increasing
for $t\in[0,1/3]$, 
and it is~$0$ at $t=0$, 
we have
$\psi(d, \tfrac{1}{3}+t)\geq\psi(d, \tfrac{1}{3}-t)$
for $t\in[0,1/3]$.
Lemma~\ref{lem:ineqs} guarantees $0<y_\mu<1/3$,
so taking $t=1/3-y_\mu$, we get
\begin{equation}
\label{eq:third8}
 \psi(d,\tfrac{2}{3} - y_\mu)\geq \psi(d,y_\mu).
\end{equation}
Combining~\eqref{eq:second8}, \eqref{eq:third8} and~\eqref{eq:first8}
we obtain 
 $\psi(d,x_\mu)>\psi(d,1-x_\mu-y_\mu)$
and
$\psi(d,x_\mu)>\psi(d, y_\mu)$, which prove~\eqref{eq:hellmu}, and hence the lemma.\end{proof}

\begin{lemma}\label{lem:ffixedpoints}
If  $d \geq 23$ then there is no solution to the system of equations
\begin{equation}\label{eq:ffixedpoints}
\begin{cases}
f_u(d,\beta_*(d),x,y) = x\cr
f_\ell(d,\beta_*(d),x,y)=y
\end{cases}
\end{equation}
which satisfies $x \geq 157y/80 \geq 0$ and $2x + y \geq 1 \geq 2y+x$.
\end{lemma}
\begin{proof}
Consider any fixed~$d\geq 23$ and, for the sake of contradiction, assume that 
such an $(x,y)$ exists. Let $\mu = x/y$, so that (by Definition~\ref{def:gs}), $(\mu, y)$ is a solution to the equation
\[g_u(d,\mu,y) = g_\ell(d,\mu,y) = 0.\]
The conditions $x \geq 157y/80 \geq 0$ and $2x + y \geq 1 \geq 2y+x$ translate into
$\mu\geq {157}/{80}$ and
$  {1}/{(2\mu+1)}\leq y \leq  {1}/{(\mu+2)}$.
Since $g_u(d,\mu,y) = 0$, by Lemma~\ref{lem:guymono} and
Corollary~\ref{lem:hu<0}, we have $y < y_\mu$. Since $g_\ell(d,\mu,y) = 0$, by Lemmas~\ref{lem:glymono}~and~\ref{lem:hl>0}, we have $y > y_\mu$. This yields a contradiction.
\end{proof}

\begin{corollary}\label{lem:d>=23bound}
For every integer $d \geq 23$, there exists a positive integer $n_0$ such that for all $n \geq n_0$,
\[\frac{u_n(d,\beta_*(d))}{\ell_n(d,\beta_*(d))}\leq \frac{53}{27}.\]
\end{corollary}
\begin{proof}
Fix $d\geq 23$. For simplicity, we will write $\beta_*$ instead of $\beta_*(d)$. 

Recall that the sequences $\{u_n(d,\beta_*)\}$ and $\{\ell_n(d,\beta_*)\}$ converge to the limits $u_\infty(d, \beta_*)$ and $\ell_\infty(d,\beta_*)$, respectively (cf. \eqref{eq:tb53tb5b3fwefe}). Moreover, by Lemma~\ref{lem:fixedpoints},  the pair $(x,y)=(u_\infty(d, \beta_*), \ell_\infty(d,\beta_*))$ is a solution to the  system of equations~\eqref{eq:ffixedpoints}
satisfying $0<y\leq 1-x-y \leq x < 1$. By Lemma~\ref{lem:ffixedpoints},   there is no solution $(x,y)$ to \eqref{eq:ffixedpoints} such that $x \geq 157y/80 \geq 0$ and $2x + y \geq 1 \geq 2y+x$. So, it must be the case that
    ${u_\infty(d,\beta_*)}/{\ell_\infty(d,\beta_*)}<{157}/{80}$.
Since $53/27>157/80$, there exists 
a positive integer~$n_0$ such that for all $n \geq n_0$, ${u_n(d,\beta_*)}/{\ell_n(d,\beta_*)}\leq{53}/{27}$.
\end{proof}

Corollary~\ref{lem:d>=23bound} accounts for integers $d\geq 23$. To account for integers $3\leq d\leq 22$, we define the following two sequences.
    \[
    \begin{cases}
    u'_0(d) = 1\cr
    \ell'_0(d) = 0
    \end{cases} \]
    and for  every non-negative integer~$n$,
    \[\begin{cases}
    u'_{n+1}(d) = \displaystyle\frac{\lceil 10000\,f_u(d,\beta_*(d),u'_n(d), \ell'_n(d))\rceil}{10000}\\[8pt]
    \ell'_{n+1}(d) = \displaystyle\frac{\lfloor 10000\,f_\ell(d,\beta_*(d), u'_n(d),\ell'_n(d))\rfloor}{10000}
    \end{cases}.\]

We have the following lemma, which is proved by brute force. 
\begin{lemma}\label{lem:d<=22sequence}
For every integer $d\in\{3,\ldots,22\}$  and 
every integer $n\in \{0,\ldots,60\}$, we have 
$u'_n(d) \geq u'_{n+1}(d)$,   $\ell'_n(d) \leq \ell'_{n+1}(d)$,   $2u'_n(d)+ \ell'_n(d) \geq 1\geq 2\ell'_n(d)+u_n(d)$ 
and $u'_{60}(d)/\ell'_{60}(d)\leq \frac{53}{27}$.
\end{lemma}
\begin{proof}
In Appendix~\ref{app:d<=22sequence}, we use Mathematica to compute all the values $u'_n(d)$ and $\ell'_n(d)$ for $n\in\{0,\ldots,60\}$ and $d\in\{3,\ldots,22\}$. We  then check that all of the desired inequalities hold.
\end{proof}

We  next show that the sequences $\{u'_n(d)\}$ and $\{\ell'_n(d)\}$ bound the sequences  $\{u_n(d,\beta_*(d))\}$ and $\{\ell_n(d,\beta_*(d))\}$ for $n\leq 60$.  
\begin{lemma}\label{lem:d<=22seqmono}
For every integer $d\in\{3,\ldots,22\}$  and 
every integer $n\in \{0,\ldots,60\}$, we have 
$u'_n(d) \geq u_n(d,\beta_*(d))$ and $\ell'_n(d) \leq \ell_n(d,\beta_*(d))$. 
\end{lemma}
\begin{proof}
Fix $d$ to be an integer between 3 and 22. Since $d$ is fixed,   we simplify the notation by
writing~$u_n$ for $u_n(d,\beta_*(d))$, $u_n'$ for $u_n'(d)$, $\ell_n$ for $\ell_n(d,\beta_*(d))$, $\ell_n'$ for $\ell'_n(d)$,  
$f_u(x,y)$ for $f_u(d,\beta_*(d),x,y)$ and   $f_\ell(x,y)$ for $f_\ell(d,\beta_*(d),x,y)$. 

We prove the lemma by induction on~$n$. For the base case $n=0$, we have $u_n = u_n'=1$ and $\ell_n = \ell_n'=0$. 
For the inductive step, suppose $n>0$. By Lemmas~\ref{lem:seqxymono} and~\ref{lem:d<=22sequence}, we have 
\[2\ell_{n-1}+u_{n-1}\leq 1 \leq 2u_{n-1}+\ell_{n-1}, \quad \text{and} \quad 2\ell'_{n-1}+u'_{n-1}\leq 1 \leq 2u'_{n-1}+\ell'_{n-1}.\]
By the induction hypothesis, we have
    \[ \ell'_{n-1}\leq \ell_{n-1}\leq u_{n-1}\leq u'_{n-1}.\]
Using Lemma~\ref{lem:fulxymono}, we therefore obtain that
    \begin{align*}
    &u'_n\geq f_u(u'_{n-1},\ell'_{n-1})\geq f_u(u_{n-1},\ell_{n-1})=u_n, \mbox{ and }\\
		&\ell'_n\leq f_\ell(u'_{n-1},\ell'_{n-1})\leq f_\ell(u_{n-1},\ell_{n-1})=\ell_n.
    \end{align*}
This completes the proof.
\end{proof}

\begin{corollary}\label{lem:d<=22bound}
    For every integer 
   $d\in\{3,\ldots,22\}$ and every integer $n\geq 60$,  
   $\displaystyle \frac{u_n(d,\beta_*(d))}{\ell_n(d,\beta_*(d))}\leq \displaystyle \frac{53}{27}$.
\end{corollary}
\begin{proof}
Fix an arbitrary integer $d$ between 3 and 22. We have the following chain of inequalities (see below for explanation): 
\[\frac{u_n(d,\beta_*(d))}{\ell_n(d,\beta_*(d))}\leq \frac{u_{60}(d,\beta_*(d))}{\ell_{60}(d,\beta_*(d))}\leq \frac{u_{60}'(d)}{\ell_{60}'(d)}\leq \frac{53}{27}.\]
The first inequality holds by Lemma~\ref{lem:seqxymono},  since the sequence $\{u_n(d,\beta_*(d))\}$ is increasing and the sequence $\{\ell_n(d,\beta_*(d))\}$ is decreasing. The second inequality holds by Lemma~\ref{lem:d<=22seqmono}. Finally, the third inequality holds by Lemma~\ref{lem:d<=22sequence}.
\end{proof}

We can now prove Lemma~\ref{lem:onestepuniq}, which we restate here for convenience.
\begin{lemonestepuniq}
\statelemonestepuniq
\end{lemonestepuniq} 
\begin{proof} 

The statement for $d=2$ follows directly from Lemma~\ref{lem:d=2bound}.

Suppose $d\geq 3$.
Let $U_n=u_n(d,\beta_*(d))$, $L_n=u_n(d,\beta_*(d))$. By  Corollaries~\ref{lem:d>=23bound} 
and~\ref{lem:d<=22bound}, there exists an integer $n_0$ such that for all $n \geq n_0$,
        \[\frac{U_n}{L_n}=\frac{u_n(d,\beta_*(d))}{\ell_n(d,\beta_*(d))}\leq \frac{53}{27}.\]
    Furthermore, by Lemmas~\ref{lem:probbound}~and~\ref{lem:seqbetamono}, for any $n\geq 0$, any configuration $\tau\colon\theleaves \to [3]$ and any colour
    $c\in[3]$, we have
    \[L_n=\ell_n(d,\beta_*(d))\leq\Pr_{\thetree}[\sigma(\theroot)=c\mid \sigma(\theleaves)=\tau]\leq u_n(d,\beta_*(d))=U_n.\]
This completes the proof.
\end{proof}

\section{Analysing the two-step recursion}\label{sec:twostep}

In this section, we fix $q\geq 3$, $d\geq 2$ and $\beta\in [0,1)$. 
All of our notation depends implicitly on these three parameters, but
when possible we  avoid using them as indices to aid readability.

Our ultimate goal is to understand the case where $q=3$,
but some of the lemmas are true more generally, so we start with $q\geq 3$.
When we later fix $q=3$, we say so explicitly.

\subsection{Characterising the maximiser of $h_{c_1,c_2,\beta}$ --- Proof of Lemmas~\ref{lem:existence} and~\ref{lem:mo12no12tone}}\label{sec:existence}

In this section, we prove Lemmas~\ref{lem:existence} and~\ref{lem:mo12no12tone} from Section~\ref{sec:simplecondition}. Recall that 
\begin{equation*}\tag{\ref{eq:gh12def}}
\begin{aligned}
g_{c_1,c_2,\beta}(\child{\pb}{1},\ldots,\child{\pb}{d})
&:=\prod^d_{k=1}\bigg(1-\frac{(1-\beta) \big(\child{p}{k}_{c_1}-\child{p}{k}_{c_2}\big)}{\beta \child{p}{k}_{c_2}+\sum_{c\neq c_2}\child{p}{k}_{c}}\bigg).\\ 
h_{c_1,c_2,\beta}(\child{\pb}{1}, \ldots, \child{\pb}{d})
&:=1+\frac{(1-\beta)\big(1-g_{c_1,c_2,\beta}(\child{\pb}{1}, \ldots, \child{\pb}{d})\big)}{\beta +\sum_{c\neq c_2}g_{c,c_2,\beta}(\child{\pb}{1}, \ldots, \child{\pb}{d})}.
\end{aligned}
\end{equation*}
To prove Lemma~\ref{lem:existence}, it will be helpful in this section to consider the set of maximisers of $h_{c_1,c_2,\beta}$.
\begin{definition}
Suppose $q\geq 3$, $ d\geq 2$ and $\beta\in [0,1)$. 
For colours $c_1,c_2\in[q]$ and $\alpha>1$, let
\begin{equation}\label{eq:max12misers}
\maxset{\alpha,c_1,c_2,\beta}={\arg\max}_{(\child{\pb}{1},\hdots,\child{\pb}{d})\in\triangle_{\alpha}^d}\, h_{c_1,c_2,\beta}\big(\child{\pb}{1},\hdots,\child{\pb}{d}\big).
\end{equation}
\end{definition}

The following lemmas give   properties of 
  the maximisers in $\maxset{\alpha,c_1,c_2,\beta}$.

\begin{lemma}\label{lem:gc1c2betaless1}
Fix $\alpha> 1$ and 
$\beta\in [0,1)$ and
colours $c_1,c_2\in[q]$. Then for any vector
 $(\child{\pb}{1},\hdots,\child{\pb}{d})\in\maxset{\alpha,c_1,c_2,\beta}$,   we have $g_{c_1,c_2,\beta}(\child{\pb}{1},\hdots,\child{\pb}{d})\leq 1$.
\end{lemma}
\begin{proof}
Assume for the sake of contradiction that $g_{c_1,c_2,\beta}(\child{\pb}{1},\hdots,\child{\pb}{d})> 1$. Then, we have that  $h_{c_1,c_2,\beta}(\child{\pb}{1},\hdots,\child{\pb}{d})<1$, which contradicts the fact that $(\child{\pb}{1},\hdots,\child{\pb}{d})\in\maxset{\alpha,c_1,c_2,\beta}$ since $h_{c_1,c_2,\beta}$ can take the value 1 by setting all of its arguments to be equal to the uniform vector $(1/q,\hdots,1/q)\in \triangle_\alpha$.
\end{proof}

\begin{lemma}\label{lem:pc2min}
Fix $\alpha> 1$ and $\beta\in[0,1)$ and any two distinct colours $c_1\in[q]$ and $c_2\in[q]$.
Then for any vector
 $(\child{\pb}{1},\hdots,\child{\pb}{d})\in\maxset{\alpha,c_1,c_2,\beta}$ and any $k\in[d]$,   we have
 $\child{{p}}{k}_{c_2} = \min_{c\in[q]} \{\child{{p}}{k}_{c}\}$.
\end{lemma}
\begin{proof}
Assume for the sake of contradiction that there  is $k\in[d]$ and $c\in[q]$ such that 
$\child{{p}}{k}_{c_2}> \child{{p}}{k}_{c}$. 
Define $( \child{\tilde\pb}{1},\hdots,\child{\tilde\pb}{d})\in \triangle_{\alpha}^d$ as follows.
\begin{itemize}
\item If $j\neq k$, then $\child{\tilde\pb}{j} = \child{\pb}{j}$.
\item If $c'\notin\{c,c_2\}$, then $\child{\tilde p}{k}_{c'} = \child{p}{k}_{c'}$.
\item $\child{\tilde{p}}{k}_{c_2}=\child{{p}}{k}_{c}$.
\item $\child{\tilde{p}}{k}_{c}= \child{{p}}{k}_{c_2}$.
\end{itemize}
The definition of $g_{c',c_2}$ ensures that, for all 
$c'\neq c_2$, we have
\[g_{c',c_2,\beta}  (\child{\tilde{\pb}}{1},\ldots, \child{\tilde{\pb}}{d})< g_{c',c_2,\beta}  (\child{{\pb}}{1},\ldots, \child{{\pb}}{d}),\]  
since the $k$-th factor in the  definition \eqref{eq:gh12def} of $g_{c_1,c_2,\beta}$ became larger (by switching $\child{\pb}{k}$ to  $\child{\tilde{\pb}}{k}$). The definition of $h_{c_1,c_2,\beta}$, together with the fact that $c_1$ and $c_2$ are distinct and Lemma~\ref{lem:gc1c2betaless1},
implies   
\[h_{c_1,c_2,\beta}  (\child{\tilde{\pb}}{1},\ldots, \child{\tilde{\pb}}{d})> h_{c_1,c_2,\beta}  (\child{{\pb}}{1},\ldots, \child{{\pb}}{d}),\]
which contradicts the 
fact that $\big(\child{\pb}{1},\hdots,\child{\pb}{d}\big)\in\maxset{\alpha,c_1,c_2,\beta}$.  
\end{proof}

\begin{lemma}\label{lem:alpha>1}
Fix $\alpha> 1$ and $\beta\in [0,1)$ and any two distinct colours $c_1\in[q]$ and $c_2\in[q]$.
Then for any vector
 $\big(\child{\pb}{1},\hdots,\child{\pb}{d}\big)\in\maxset{\alpha,c_1,c_2,\beta}$ and any $k\in[d]$,   we have
 $\child{{p}}{k}_{c_2} < \max_{c\in[q]} \{\child{{p}}{k}_{c}\}$.
\end{lemma}

\begin{proof}
For the sake of contradiction, suppose that there is  $k\in[d]$ such that $\child{{p}}{k}_{c_2}\geq\max_{c\in[q]}\{\child{{p}}{k}_{c}\}$. By Lemma~\ref{lem:pc2min},   $\child{p}{k}_{c_2}$ is the minimum entry of the vector 
$\child{\pb}{k}$, so all entries of $\child{\pb}{k}$ must be equal (and hence, equal to $1/q$).
Define the vector 
$( \child{\tilde\pb}{1},\hdots,\child{\tilde\pb}{d})\in \triangle_{\alpha}^d$ as follows.
\begin{itemize}
\item If $j\neq k$, then $\child{\tilde\pb}{j} = \child{\pb}{j}$.
\item $\child{\tilde{p}}{k}_{c_1}= \alpha/(\alpha+q-1)$.
\item If $c \neq c_1$, then $\child{\tilde p}{k}_{c} =  1/(\alpha+q-1)$.
\end{itemize}
The definition of $g_{c,c_2,\beta}$ together with Lemma~\ref{lem:gc1c2betaless1} ensure that
\[g_{c_1,c_2,\beta}  (\child{\tilde{\pb}}{1},\ldots, \child{\tilde{\pb}}{d})< g_{c_1,c_2,\beta}  (\child{{\pb}}{1},\ldots, \child{{\pb}}{d})=1\]
and that for every $c \neq c_1$,
$g_{c,c_2,\beta}  (\child{\tilde{\pb}}{1},\ldots, \child{\tilde{\pb}}{d})= g_{c,c_2,\beta}  (\child{{\pb}}{1},\ldots, \child{{\pb}}{d})$.
The definition of $h_{c_1,c_2,\beta}$ therefore implies   that
\[h_{c_1,c_2,\beta}  (\child{\tilde{\pb}}{1},\ldots, \child{\tilde{\pb}}{d})> h_{c_1,c_2,\beta}  (\child{{\pb}}{1},\ldots, \child{{\pb}}{d}),\]
which contradicts the 
fact that $\big(\child{\pb}{1},\hdots,\child{\pb}{d}\big)\in\maxset{\alpha,c_1,c_2,\beta}$.
\end{proof}

To proceed, we will need the following technical fact.
\begin{lemma}\label{lem:linearoverlinear}
Let $A_0$, $B_0$, $A_1$ and $B_1$ be real numbers. Let~$a$ and~$b$ be real numbers satisfying $a\leq b$
such that, for all $x\in[a,b]$, $A_1 + B_1 x \neq 0$. Let
$\mathcal{L}(x) =  {(A_0+B_0 x)}/{(A_1+B_1 x)}$. Then
\[\max_{x\in[a,b]} \mathcal{L}(x) = \max\{\mathcal{L}(a), \mathcal{L}(b)\}.\]
\end{lemma}
\begin{proof}
    Since $A_1+B_1x\neq 0$ and \[\deriv{\cal L}{x} = \frac{A_1B_0-A_0B_1}{(A_1+B_1x)^2},\]
    $\mathcal{L}(x)$ is a monotone function on $[a, b]$. Thus, $\max_{x\in[a,b]} \mathcal{L}(x) = \max\{\mathcal{L}(a), \mathcal{L}(b)\}$.
\end{proof}
\begin{lemma}\label{lem:ratio}
 Fix $\alpha > 1$ and $\beta\in[0,1)$ and two distinct colours~$c_1$ and~$c_2$ in $[q]$.
 Suppose that $(\child{\pb}{1},\hdots,\child{\pb}{d})\in\maxset{\alpha,c_1,c_2,\beta}$.
 Then for every $k\in[d]$, there exists $\tilde\pb\in\triangle_\alpha$ 
 such that
 $$(\child{\pb}{1},\hdots,\child{\pb}{k-1},\tilde\pb,\child{\pb}{k+1},\hdots,\child{\pb}{d})\in\maxset{\alpha,c_1,c_2,\beta}$$ 
 and, for all $c\in[q]$, 
$\tilde{p}_{c}/\tilde{p}_{c_2}  \in\{1,\alpha\}$.
\end{lemma}

\begin{proof}
Fix   a tuple 
$(\child{\pb}{1},\hdots,\child{\pb}{d})\in\maxset{\alpha,c_1,c_2,\beta}$.
Fix $k\in[d]$.
Given any  
$\hat\pb\in \triangle_\alpha$,
we will be interested in the quantity
$h_{c_1,c_2,\beta}(\child{\pb}{1},\hdots,\child{\pb}{k-1},\hat{\pb},\child{\pb}{k+1},\hdots,\child{\pb}{d})$.
Given any $c'\neq c_2$, define  
\[P_{c'}  := 
\prod_{j\neq k}\left(1-\frac{(1-\beta) \big(\child{p}{j}_{c'}-\child{p}{j}_{c_2}\big)}{\beta \child{p}{j}_{c_2}+\sum_{c\neq c_2}\child{p}{j}_{c}}\right).\]

It will be helpful to re-parameterise the elements of $\hat\pb$.
Recall  that the definition of $\triangle_\alpha$ implies that  
$\hat p_c>0$ for every $c\in[q]$.
For every $c\in[q]$, let  $\mu_c(\hat\pb) = \hat p_{c}/\hat p_{c_2}$.
Let ${\boldsymbol{\mu}}(\hat\pb)$ be the tuple
${\boldsymbol{\mu}}(\hat\pb) = (\mu_1(\hat\pb),\ldots,\mu_q(\hat\pb))$.
Going the other direction from a tuple $\boldsymbol{\mu}$ with entries in $[1,\alpha]$,
for every $c\in[q]$, let $p_c(\boldsymbol{\mu}) = 
 \mu_c /\sum_{c'\in[q]} \mu_{c'}$ and
 let $\pb(\boldsymbol{\mu})$ be the tuple $(p_1(\boldsymbol{\mu}),\ldots,p_q(\boldsymbol{\mu}))$.

It is going to be important to note that the re-parameterisation is without loss of
information, so, to this end, let 
$\Omega_\alpha = \{\boldsymbol{\mu} \in [1,\alpha]^q \mid \mu_{c_2} = 1\}$.
The definition of $\triangle_\alpha$ and Lemma~\ref{lem:pc2min}
ensure that, for every 
$(\child{\tilde\pb}{1},\hdots,\child{\tilde\pb}{d})\in\maxset{\alpha,c_1,c_2,\beta}$ and any $j\in[d]$,  
$\boldsymbol{\mu}(\child{\tilde\pb}{j}) \in \Omega_\alpha$. 
Also, given any $\boldsymbol{\mu} \in \Omega_\alpha$, the vector
$\pb(\boldsymbol{\mu})$  
is in $\triangle_\alpha$.

Given a tuple $\boldsymbol{\mu}\in \Omega_\alpha$,
we will simplify notation by letting
\def\denkludge{m} 
$\denkludge(\boldsymbol{\mu}) := \beta + \sum_{c\neq c_2} \mu_c$.
 Then
we can write $g_{c',c_2}$ as
\begin{align*}g_{c',c_2}(\child{\pb}{1},\hdots,\child{\pb}{k-1},\hat{\pb},\child{\pb}{k+1},\hdots,\child{\pb}{d})
&=P_{c'} \left(1- \frac{(1-\beta)  ( {\hat p}_{c'}- {\hat p} _{c_2} )}
{\beta  {\hat p} _{c_2}+\sum_{c\neq c_2} {\hat p} _{c}}\right)\\
&= P_{c'}\left(1-\frac{(1-\beta)(\mu_{c'}(\hat\pb)-1)}{ \denkludge(\boldsymbol{\mu}(\hat\pb))}\right).
\end{align*}

Given the (fixed) values of 
$\child{\pb}{1},\hdots,\child{\pb}{k-1}$
and $\child{\pb}{k+1},\hdots,\child{\pb}{d} $, let
$$
h(\boldsymbol{\mu}) := 
  \frac{ 
 \denkludge(\boldsymbol{\mu}) - \denkludge(\boldsymbol{\mu}) P_{c_1} + P_{c_1} (1-\beta)(\mu_{c_1}-1)  
 }
 {\beta \denkludge(\boldsymbol{\mu})  + \denkludge(\boldsymbol{\mu}) \sum_{c'\neq c_2} P_{c'} 
 - \sum_{c'\neq c_2} 
P_{c'}   (1-\beta) (\mu_{c'}-1) 
}.$$
  
  Then we can write $h_{c_1,c_2,\beta}$ as
\begin{align*}
h_{c_1,c_2,\beta}(\child{\pb}{1},\hdots,\child{\pb}{k-1},\hat{\pb},\child{\pb}{k+1},\hdots,\child{\pb}{d})
&=   
1+\frac{(1-\beta)\left(1-
P_{c_1}\left(1-\frac{(1-\beta)(\mu_{c_1}(\hat\pb)-1)}{ \denkludge(\boldsymbol{\mu}(\hat\pb))}\right)
\right)}{\beta +\sum_{c'\neq c_2}
P_{c'}\left(1-\frac{(1-\beta)(\mu_{c'}(\hat\pb)-1)}{ \denkludge(\boldsymbol{\mu}(\hat\pb))}\right)
}
\\ 
&= 
 1+ (1-\beta) h(\boldsymbol{\mu}(\hat\pb)).
 \end{align*}
 
Since  
$(\child{\pb}{1},\hdots,\child{\pb}{d})\in\maxset{\alpha,c_1,c_2,\beta}$,
taking $\hat{\pb} = \child{\pb}{k}  $ maximises  
$$h_{c_1,c_2,\beta}(\child{\pb}{1},\hdots,\child{\pb}{k-1},\hat{\pb},\child{\pb}{k+1},\hdots,\child{\pb}{d})$$
over $\triangle_\alpha^d$.
Thus, for any maximiser ${\boldsymbol{\mu}}$ of 
$h(\boldsymbol{\mu})$ over $\Omega_\alpha$,
we have
$$h_{c_1,c_2,\beta}(\child{\pb}{1},\hdots,\child{\pb}{k-1}, 
\boldsymbol{p}({\boldsymbol{\mu}}),\child{\pb}{k+1},\hdots,\child{\pb}{d})
\in \maxset{\alpha,c_1,c_2,\beta}.$$
So to prove the lemma
(taking $\tilde\pb = \pb({\boldsymbol{\mu}})$)
it suffices to find a maximiser ${\boldsymbol{\mu}}$ of 
$h(\boldsymbol{\mu})$ over $\Omega_\alpha$
such that 
for all $c\in[q]$, 
$\mu_c  \in\{1,\alpha\}$.
This is what we will do in the rest of the proof.
The definition of $\Omega_\alpha$ guarantees that
$\mu_{c_2} = 1$. 

Fix any $c\neq c_2$.
For any fixed values $\mu_1,\ldots,\mu_{c-1}$
and $\mu_{c+1},\ldots,\mu_q$, all
in $[1,\alpha]$, satisfying $\mu_{c_2}=1$,
consider $h(\boldsymbol{\mu})$ as a function of $\mu_c$.
Note that both the numerator and denominator of $h(\boldsymbol{\mu})$ are linear in $\mu_c$.
We will argue that the denominator is not zero when $\mu_c\in[1,\alpha]$.
Using this, 
Lemma~\ref{lem:linearoverlinear}
guarantees that, given $\mu_1,\ldots,\mu_{c-1},\mu_{c+1},\ldots,\mu_q$,
$h(\boldsymbol{\mu})$ is maximised by setting $\mu_c \in \{1,\alpha\}$.
Going through the colours~$c$ one-by-one, we find the desired maximiser~$\boldsymbol\mu$.

To complete the proof, we just need to show that the denominator 
of $h(\boldsymbol{\mu})$
is not zero when  $\boldsymbol{\mu} \in \Omega_\alpha$.
This follows easily,  
since, for all $c'\neq c_2$, $P_{c'}>0$
and $m(\boldsymbol{\mu}) \geq (1-\beta)(\mu_{c'}-1)$.
\end{proof}

Lemmas~\ref{lem:alpha>1}~and~\ref{lem:ratio} yield  Lemma~\ref{lem:existence} as an immediate corollary.
\begin{lemexistence}
\statelemexistence
\end{lemexistence}
\begin{proof}
Just use Lemmas~\ref{lem:alpha>1}~and~\ref{lem:ratio}.
\end{proof}

We also now prove Lemma~\ref{lem:mo12no12tone}.
\begin{lemmonotone}
\statelemmonotone
\end{lemmonotone}
\begin{proof}
Fix $\alpha>1$ and arbitrary colours $c_1,c_2\in [q]$. Note that the set of $(\alpha,c_2)$-extremal tuples does not depend on the parameter $\beta$, and for each $\beta\in (0,1)$ there exists by Lemma~\ref{lem:existence} an $(\alpha,c_2)$-extremal tuple which achieves the maximum in ${\max}_{(\child{\pb}{1},\hdots,\child{\pb}{d})\in\triangle_{\alpha}^d}\, h_{c_1,c_2,\beta}\big(\child{\pb}{1},\hdots,\child{\pb}{d}\big)$.

Therefore, the inequality will follow by showing  that 
\begin{equation}\label{eq:bb45yby6nu42}
h_{c_1,c_2,\beta''}(\pb^{(1)},\hdots,\pb^{(d)})\leq h_{c_1,c_2,\beta'}(\pb^{(1)},\hdots,\pb^{(d)}),
\end{equation}
where $\big(\pb^{(1)},\hdots,\pb^{(d)}\big)$ is an arbitrary $(\alpha,c_2)$-extremal tuple. Using the extremality of the tuple $\big(\pb^{(1)},\hdots,\pb^{(d)}\big)$,  we have that  $p^{(k)}_{c_2}\leq p^{(k)}_{c}$ for all colours $c\in [q]$ and every $k\in [d]$. Hence, using the definition \eqref{eq:gh12def} of $g_{c_1,c_2,\beta}$ and that $\beta'\leq \beta''$, we have
\[1\geq g_{c,c_2,\beta''}(\pb^{(1)},\hdots,\pb^{(d)})\geq g_{c,c_2,\beta'}(\pb^{(1)},\hdots,\pb^{(d)}) \mbox{ for all } c\in [q].\]
In turn, using the definition \eqref{eq:gh12def} of $h_{c_1,c_2,\beta}$, we obtain from  this that \eqref{eq:bb45yby6nu42} holds, as wanted.
\end{proof}

\subsection{Bounding the two-step recursion when $q=3$ --- Proof  of Lemma~\ref{lem:twosteptoprove}}\label{sec:twosteptoprove}

In this section, we assume that $q = 3$ and give the proof of Lemma~\ref{lem:twosteptoprove} that verifies Condition~\ref{cond:we} for all $\alpha\in (1,53/27]$. 

Recall that for a pair $(q,d)$, Condition~\ref{cond:we}, for a fixed value of $\alpha>1$ and colours $c_1,c_2\in [q]$, amounts to checking
\begin{equation}\label{eq:2ed2g2v2vtvt}
h_{c_1,c_2,\beta_*}\big(\child{\pb}{1},\hdots,\child{\pb}{d}\big)<\alpha^{1/d} \mbox{ for all } \child{\pb}{1},\hdots,\child{\pb}{d}\in \mathrm{Ex}_{c_2}(\alpha),
\end{equation}
where the set $\mathrm{Ex}_{c_2}(\alpha)$ is given by (cf. \eqref{ex:ca}) 
\begin{equation*}
\mathrm{Ex}_{c}(\alpha)=\big\{(p_1,\hdots,p_q)
\in \{1,\alpha\}^q
\mid p_{c}=1 \land \exists c'\in[q]  \text{ such that }  p_{c'}=\alpha \big\}. 
\end{equation*}
Note that, for $q=3$, $\mathrm{Ex}_{c}(\alpha)$ has exactly 3 vectors. As we shall see shortly, we can capture the value of $h_{c_1,c_2,\beta}$ when $\child{\pb}{1},\hdots,\child{\pb}{d}\in  \mathrm{Ex}_{c}(\alpha)$ using a function 
${\varphi}(d,d_0,d_1,\alpha,\beta)$ which depends on $\alpha$, $\beta$ and $d$, but also on two non-negative integers $d_0$ and $d_1$ with $d_0+d_1 \leq d$; roughly, $d_0$ is the number of   $\child{\pb}{1},\hdots,\child{\pb}{d}$ which are equal to the first vector in $\mathrm{Ex}_{c}(\alpha)$, 
$d_1$ to the second vector and the remaining $d-d_0-d_1$ to the third vector. Let us first define the function ${\varphi}(d,d_0,d_1,\alpha,\beta)$.

\begin{definition}\label{def:xy}
Let $\beta\in [0,1]$. Fix any $\alpha >1$ and any integer $d\geq 2$.
Let $x(\alpha,\beta)=1- \frac{(1-\beta)(\alpha-1)}{\beta+2\alpha}$
and $y(\alpha,\beta)=1- \frac{(1-\beta)(\alpha-1)}{\beta+\alpha+1}$.
For any nonnegative integers $d_0$ and $d_1$ such that
$d_0 + d_1 \leq d$, let
  $$
{\varphi}(d,d_0,d_1,\alpha,\beta)=1+\frac{(1-\beta)(1-x(\alpha,\beta)^{d_0}y(\alpha,\beta)^{d_1})}{\beta+x(\alpha,\beta)^{d_0}(y(\alpha,\beta)^{d_1}+y(\alpha,\beta)^{d-d_0-d_1})}.
 $$ 
\end{definition}

We then have the following lemma.
\begin{lemma}\label{lem:vald0d1d2}
    Suppose $q=3$, $d \geq 2$ and $\beta\in[0,1]$.    
    Fix $\alpha > 1$ and distinct colours $c_1,c_2\in [q]$,  and let $\child{\pb}{1}, \child{\pb}{2},\ldots,\child{\pb}{d} \in \mathrm{Ex}_{c_2}(\alpha)$.
    Then  there are nonnegative integers 
    $d_0$ and $d_1$ with $d_0+d_1\leq d$ such that
\[h_{c_1,c_2,\beta}\big(\child{\pb}{1},\hdots,\child{\pb}{d}\big)= \varphi(d, d_0, d_1, \alpha, \beta).\]
\end{lemma}

\begin{proof} Let $c_3$ be the remaining colour in $[q]$ (other than $c_1$ and $c_2$), so that $[q]=\{c_1,c_2,c_3\}$. Since $\child{\pb}{1}, \child{\pb}{2},\ldots,\child{\pb}{d} \in \mathrm{Ex}_{c_2}(\alpha)$,  for every $k\in[d]$ we have that $\child{p}{k}_{c_2}=1$ and one of three situations
occurs: 
\[\mbox{ (i)   $\child{p}{k}_{c_1} = \child{p}{k}_{c_3} = \alpha$, 
(ii) 
 $\child{p}{k}_{c_1} =\alpha$,  $\child{p}{k}_{c_3} =1$,
 (iii) $\child{p}{k}_{c_1} =1$,  $\child{p}{k}_{c_3} =\alpha$.}\]
If situation (i)
occurs, then, using the notation from Definition~\ref{def:xy}, we have
\begin{equation}\label{eq:usex}
1-\frac{(1-\beta) \big(\child{p}{k}_{c_1}-\child{p}{k}_{c_2}\big)}{\beta \child{p}{k}_{c_2}+
 \child{p}{k}_{c_1} + \child{p}{k}_{c_3}
}
= x(\alpha,\beta).\end{equation} 
If situation (ii) occurs, then this quantity is $y(\alpha,\beta)$. If situation (iii) occurs, then it is~$1$.
Now let
$d_0$ be the number of $k\in[d]$ for which situation (i) occurs and
let $d_1$ be the number of $k\in[d]$ for which situation (ii) occurs.
Then plugging \eqref{eq:usex} and the other similar observations above into the definition \eqref{eq:gh12def} of $g_{c_1,c_2,\beta}$ and $g_{c_3,c_2,\beta}$, we have
\begin{align*}
g_{c_1,c_2,\beta} \big(\child{\hat{\pb}}{1},\ldots, \child{\hat{\pb}}{d}\big)&= 
 x(\alpha,\beta)^{d_0}y(\alpha,\beta)^{d_1}, \mbox{ and }\\
g_{c_3,c_2,\beta} \big(\child{\hat{\pb}}{1},\ldots, \child{\hat{\pb}}{d}\big)&= 
 x(\alpha,\beta)^{d_0}y(\alpha,\beta)^{d-d_0-d_1}.
\end{align*}
 Plugging this into the definition of $h_{c_1,c_2,\beta}$, we have
\begin{equation*}
h_{c_1,c_2,\beta}\big(\child{\hat\pb}{1},\hdots,\child{\hat\pb}{d}\big)=1+\frac{(1-\beta)\Big(1-g_{c_1,c_2,\beta}(\child{\hat\pb}{1}, \ldots, \child{\hat\pb}{d})\Big)}{\beta +\sum_{c\neq c_2}g_{c,c_2,\beta}(\child{\hat\pb}{1}, \ldots, \child{\hat\pb}{d})}=\varphi(d, d_0, d_1, \alpha, \beta).\qedhere
\end{equation*}
\end{proof}

The following definition applies Definition~\ref{def:xy} to the critical value of $\beta$
for the special case where $d_0 + d_1=d$; we will see that this special case is all that we need to consider to verify Condition~\ref{cond:we} for $\alpha\in (1,2)$.
\begin{definition}\label{def:gocritical}
Let $q=3$.
Fix any $\alpha > 1$ and any integer $d\geq 2$.
Let $\beta_*(d) = 1-q/(d+1)$. 
Let $X = x(\alpha,\beta_*(d))$ and $Y = y(\alpha,\beta_*(d))$.
Let  
    \[\varphi_*(d,d_0,\alpha)=\varphi(d, d_0, d-d_0, \alpha, \beta_*(d))=1+\frac{3}{d+1}\cdot\frac{1-X^{d_0}Y^{d-d_0}}{\beta_*(d)+X^{d_0}(Y^{d-d_0}+1)}.\]
\end{definition}
The values~$X$ and~$Y$ from definition~\ref{def:gocritical} are clearly functions of~$d$ and~$\alpha$,
but we suppress this in the notation to avoid clutter. (These values will arise in proofs, but,
unlike $\varphi_*(d,d_0,\alpha)$, they will not arise in statements of lemmas.)

The following lemma applies for $\alpha\in (1,2)$. 
\begin{lemma}\label{lem:hbound}
Suppose $q=3$, $d\geq 2$ and $ 1-q/(d+1)\leq \beta\leq 1$. 
Fix $\alpha\in (1,2)$ and distinct colours $c_1,c_2\in [q]$. Suppose that $\child{\pb}{1}, \child{\pb}{2},\ldots,\child{\pb}{d} \in \mathrm{Ex}_{c_2}(\alpha)$. Then, there is a non-negative integer $d_0 \leq d$ such that
    \begin{equation*}
    h_{c_1,c_2,\beta}\big(\child{\pb}{1},\hdots,\child{\pb}{d}\big)\leq \varphi_*(d, d_0,\alpha).
    \end{equation*}
\end{lemma}
\begin{proof}
Since $d$ is fixed, we simplify the notation   by writing $\beta_*$ for $\beta_*(d)$.

Note that both $x(\alpha, \beta)$ and $y(\alpha, \beta)$ are increasing functions of $\beta$, so $\varphi$ is a decreasing function of $\beta$. Thus for all 
$\beta \in [\beta^*,1]$, we have   
$
    {\varphi}(d,d_0,d_1,\alpha,\beta_*) \geq {\varphi}(d, d_0, d_1,\alpha,\beta)
$.
Combining this with Lemma~\ref{lem:vald0d1d2},
we find that     
there are nonnegative integers~$d_0$ and $d_1$ with   $d_0 + d_1 \leq d$ 
such that     
    \begin{equation}\label{eq:hphibound}
    h_{c_1,c_2,\beta}\big(\child{\pb}{1},\hdots,\child{\pb}{d}\big) = {\varphi}(d, d_0, d_1, \alpha, \beta)\leq {\varphi}(d, d_0, d_1, \alpha, \beta_*).
    \end{equation}
    Let $d_2 = d- d_0 - d_1 \geq 0$.
    If $d_1 < d_2$, then \[ {\varphi}(d,d_0,d_2,\alpha,\beta_*)=1+\frac{(1-\beta_*)(1-x(\alpha,\beta_*)^{d_0}y(\alpha,\beta_*)^{d_2})}{\beta_*+x(\alpha,\beta_*)^{d_0}(y(\alpha,\beta_*)^{d_1}+y(\alpha,\beta_*)^{d_2})}\geq{\varphi}(d,d_0,d_1,\alpha,\beta_*).\] 
    So we can further assume that $d_1 \geq d_2$.
     Since $1<\alpha < 2$, and $\beta_*\leq 1$, 
    we have
$$
    x(\alpha, \beta_*)^2-y(\alpha, \beta_*)   =\frac{(1-\beta_*)(\alpha-1)}{(\beta_*+2\alpha)(\beta_*+\alpha+1)}\left(\frac{\beta_*+\alpha+1}{\beta_*+2\alpha}(1-\beta_*)(\alpha-1)-(\beta_*+2)\right) \leq  0,$$
    which implies
    \[x(\alpha,\beta_*)^{d_0}y(\alpha,\beta_*)^{d_1} \geq x(\alpha,\beta_*)^{d_0+2}y(\alpha,\beta_*)^{d_1-1},\]
    and 
        \[x(\alpha,\beta_*)^{d_0}y(\alpha,\beta_*)^{d_2} \geq x(\alpha,\beta_*)^{d_0+2}y(\alpha,\beta_*)^{d_2-1}.\]
 These (together with the definition of $\varphi$) imply that if $d_1\geq d_2\geq 1$
 then 
 $${\varphi}(d,d_0,d_1,\alpha,\beta_*) \leq{\varphi}(d,d_0+2,d_1-1,\alpha,\beta_*).$$
 Repeating this until $d_2=0$, we obtain
        \[{\varphi}(d,d_0,d_1,\alpha,\beta_*) \leq{\varphi}(d,d_0+2d_2,d_1-d_2,\alpha,\beta_*)=\varphi_*(d,d_0+2d_2,\alpha).\]
    Combining this with~\eqref{eq:hphibound},  we obtain
    \begin{equation*}
    h_{c_1,c_2,\beta}\big(\child{\pb}{1},\hdots,\child{\pb}{d}\big)\leq \varphi_*(d,d_0+2d_2,\alpha).\qedhere
    \end{equation*}
\end{proof}

\begin{lemma}\label{lem:d<=22techlem}
 For every fixed $d\in\{2,\ldots,22\}$, $d_0 \in \{0,\ldots,d\}$ and $\alpha \in (1,53/27]$,
    we have  $\varphi_*(d,d_0,\alpha) < \alpha^{1/d}$. 
\end{lemma}

\begin{proof}
This is rigorously verified using the Resolve function of Mathematica in  Section~\ref{app:d<=22techlem}.
\end{proof}

\begin{lemma}\label{lem:d>=23techlem} 
 For every fixed integer $d\geq 23$, $d_0 \in \{0,\ldots,d\}$ and $\alpha \in (1,53/27]$,
 we have  $\varphi_*(d,d_0,\alpha) < \alpha^{1/d}$.    
 \end{lemma}

\begin{proof}    
Fix $d$ and $d_0$ such that 
$d\geq 23$ and  $d_0 \in \{0,\ldots,d\}$.
Since $d$ is fixed, we simplify the notation by writing $\beta_*$ for $\beta_*(d)$.
Note that $\beta_*\in (0,1)$.
Given~$d$, let $X$ and $Y$ be the functions of~$\alpha$  defined in Definition~\ref{def:gocritical}, and observe that these are positive for all $\alpha\geq 1$.
Let
    \[\tilde{\varphi}(d, d_0, \alpha)=1+\frac{1-X^{d_0}Y^{d-d_0}\strut}{\strut d X^{\frac{2d_0}{2+\beta_*}}Y^{\frac{d-d_0}{2+\beta_*}}}-\alpha^{1/d}=1+\frac{1}{d}\Big(X^{-\frac{2d_0}{2+\beta_*}}Y^{-\frac{d-d_0}{2+\beta_*}}-X^{\frac{d_0 \beta_*}{2+\beta_*}}
    Y^{\frac{(d-d_0)(1+\beta_*)}{2+\beta_*}}\Big)-\alpha^{1/d}.\]
    
By the weighted arithmetic-mean geometric-mean inequality\footnote{
The inequality says that for non-negative $x_1,x_2,x_3,w_1,w_2,w_3$ with $w=w_1+w_2+w_3>0$,
$w_1 x_1 + w_2 x_2 + w_3 x_3 \geq w x_1^{w_1/w} x_2^{w_2/w} x_3^{w_3/w}$. Take $x_1=1$, $x_2 = X^{d_0}Y^{d_1}$,
$x_3 = X^{d_0}$, $w_1 = \beta_*$, $w_2=1$ and $w_3=1$.} we have
\[\beta_*+X^{d_0}(Y^{d_1}+1)\geq (2+\beta_*)X^{\frac{2d_0}{2+\beta_*}}Y^{\frac{d_1}{2+\beta_*}},\]
which yields $\tilde{\varphi}(d,d_0,\alpha) \geqslant \varphi_*(d, d_0, \alpha) - \alpha^{1/d}$.
Thus, our goal is to prove $\tilde{\varphi}(d,d_0,\alpha)< 0$ when $\alpha\in (1,53/27]$. 
When $\alpha=1$,  $X$ and $Y$ are $1$  
so  $\tilde{\varphi}(d, d_0, 1)=0$.
To prove the lemma, it suffices to show that $\pderiv{\tilde{\varphi}}{\alpha}< 0$
for $\alpha\in (1,53/27]$. The rest of the proof is devoted to this technical fact.    
It is broken up into four steps.\vskip 0.2cm
 
\noindent{\bf Step 1.}
Let 
$ \xi_1= 2d_0 \big(\frac{2 + \alpha \beta_*}{2 + \beta_*}\big)+(d-d_0)\alpha$,
$M = -X^{d_0}Y^{d-d_0}\left(d_0\beta_*\frac{2 + \alpha \beta_*}{2 + \beta_*}+ (d-d_0)(\beta_*+1)\alpha\right)$,
$S = X^{\frac{2d_0}{2+\beta_*}} Y^{\frac{d-d_0}{2+\beta_*}} (\alpha+\beta_*+1) (\alpha \beta_*+2) / (1-\beta_*)$,
and $\xi_2 = M + \alpha^{1/d} S$. 
  The goal of Step~1 is to show that   $\pderiv{\tilde{\varphi}}{\alpha}< 0$
 follows from $\xi_1 \leq \xi_2$.
    
\noindent{\bf Technical details of Step~1.}
To calculate the derivative of the function $\tilde{\varphi}$, let 
\[g_X :=  (2 \alpha+\beta_*) (\alpha \beta_*+\alpha+1),\qquad g_Y :=(\alpha+\beta_*+1) (\alpha \beta_*+2).\]
Using Mathematica, we verify in Section~\ref{app:d>=23techlem}  the formula
\begin{equation}\label{eq:vetvt4g66}
\begin{aligned}
 \pderiv{\tilde{\varphi}}{\alpha}&=  
 \bigg(\frac{(1-\beta)X^{-\frac{2d_0}{2+\beta_*}}Y^{-\frac{d-d_0}{2+\beta_*}}}{d\alpha g_Y}  \bigg)
 \left(2 d_0\alpha  \frac{g_Y}{g_X}  +  (d-d_0)\alpha   +\right.\\
&\hskip 5.5cm \left. X^{ d_0} Y^{ d-d_0} \Big(d_0 \beta_*\alpha  \frac{g_Y}{g_X} +  
 (d-d_0)(1+\beta_*)\alpha  \Big) 
 - \alpha^{1/d} S\right).
\end{aligned}
\end{equation}
For all $\alpha>1$, we have that 
\[\alpha(2+\beta_*)(\alpha + \beta_* + 1)- (2\alpha + \beta_*)(\alpha \beta_* + \alpha +1)=-\beta_*(\alpha-1)^2 < 0,\]
so we obtain the bound
        \begin{equation}\label{eqn:fxfy}
        \frac{g_Y }{g_X }= \frac{(\alpha+\beta_*+1) (\alpha \beta_*+2)}{(2 \alpha+\beta_*) (\alpha \beta_*+\alpha+1)}< \frac{(\alpha+\beta_*+1) (\alpha \beta_*+2)}{\alpha(2+\beta_*)(\alpha + \beta_* + 1)}= \frac{2 + \alpha \beta_*}{\alpha (2 + \beta_*)}.
        \end{equation}

Now note that the first parenthesised expression in \eqref{eq:vetvt4g66} is  positive since $X,Y>0$ for all $\alpha>1$.  Thus, 
  to show  $\pderiv{\tilde{\varphi}}{\alpha}< 0$, 
 it suffices to show that the second parenthesised expression in \eqref{eq:vetvt4g66}
 is less than~$0$.
 To do this, we can apply the strict upper bound on $g_X/g_Y$ from \eqref{eqn:fxfy},  and show that the resulting expression,
 which is $\xi_1 - \xi_2$, is at most~$0$. Thus, we have completed Step~1.
\vskip 0.2cm

\noindent{\bf Step 2.}
Let $W = 1+\frac{\alpha-1}{d}-\frac{(d-1) (\alpha-1)^2}{2 d^2} $
and
$\xi_3 =  M + W S$.
The goal of Step~2 is to show 
$W \leq \alpha^{1/d}$, which implies
$\xi_3 \leq \xi_2$. 
Given Step~1, this means that 
   $\pderiv{\tilde{\varphi}}{\alpha}< 0$
will follow from showing that $\xi_1 \leq \xi_3$. 
 
 \noindent{\bf Technical details of Step~2.} 
 Recall that $d\geq 23$.
 Let $m(\alpha,d) = \alpha^{1/d}- W$.
 Note that
   \[m(1,d) = \pderiv{m}{\alpha}\Big\vert_{\alpha=1} = \pderiv{^2 m}{\alpha^2}\Big\vert_{\alpha=1}=0,\quad\mbox{and}\quad\pderiv{^3m}{\alpha^3} = \frac{\left(2-\frac{1}{d}\right) \left(1-\frac{1}{d}\right) }{d \alpha^{3-\frac{1}{d}} }>0\mbox{\ for all $\alpha > 1$.}\]
 Thus, $m(\alpha,d) \geq 0$ for $\alpha > 1$. 
 \vskip 0.2cm       
    
 \noindent{\bf Step 3.}   
 Let $\kappa=d_0/d$. After re-parameterising $\xi_1$ and $\xi_3$
 (so that they depend on $d$, $\kappa$ and $\alpha$),
 we show that $\xi_1 \leq \xi_3$ follows from 
 $\pderiv{^2 \xi_3}{\alpha^2} > 0$.
 Given the other steps, this means that 
    $\pderiv{\tilde{\varphi}}{\alpha}< 0$
 follows from  $\pderiv{^2 \xi_3}{\alpha^2} > 0$. 
 For future reference, the re-parameterisation is as follows.
 Let 
     \begin{align*}
    s_0&=-\left(\frac{(2 d-1) d (1-\kappa) \alpha}{d+1}+\frac{(d-2) \kappa ((d-2) \alpha+2(d+1))}{3(d+1)}\right), \text{ and}\cr
    t_0&=\frac{((d+1)\alpha+2d-1)  ((d-2) \alpha+2(d+1))}{3(d+1)}\left(1+\frac{\alpha-1}{d}-\frac{(d-1) (\alpha-1)^2}{2 d^2}\right).
    \end{align*} 
 Then 
\begin{equation}\label{eq:rt4tv4ybyb}
\begin{aligned}
\xi_1&= \frac{\alpha-1}{3} (-d \kappa+3 d-4 \kappa)+d(\kappa+1), \text{ and}\\
\xi_3&= s_0 X^{d\kappa }Y^{d(1-\kappa)} + t_0 X^{\frac{2\kappa(d+1)}{3}} Y^{\frac{(1-\kappa)(d+1)}{3}}.
\end{aligned}
\end{equation}
 
  \noindent{\bf Technical details of Step~3.}  
 
 The mathematica code in Section~\ref{app:d>=23techlem} 
 verifies that, at $\alpha=1$,
 $ \xi_3 = d(\kappa+1)$ and
 $  \pderiv{\xi_3}{\alpha} = \tfrac{1}{3}(-d\kappa+3d-4\kappa) $. 
     By  Taylor's theorem and the Lagrange form of the remainder, there exists a number  $\tilde{\alpha} \in (1, \alpha]$ such that
 \begin{equation}\label{eq:xi3}
    \xi_3 =d(\kappa+1) + \frac{1}{3}(-d\kappa+3d-4\kappa)(\alpha-1)+ \frac{1}{2}\left(\pderiv{^2 \xi_3}{\alpha^2}\Big\vert_{\alpha=\tilde{\alpha}}\right)(\alpha-1)^2.
    \end{equation}
    Thus, if $\pderiv{^2 \xi_3}{\alpha^2} > 0$ for all 
 $\alpha \in (1,53/27]$,   we can  conclude that 
    \begin{align}\label{eq:xi1xi3}
    \xi_3 \geqslant d(\kappa+1) + \frac{1}{3}(-d\kappa+3d-4\kappa)(\alpha-1) = \xi_1.
    \end{align} 
    
   \noindent{\bf Step 4.}     
 Using the parameterisation \eqref{eq:rt4tv4ybyb} of Step~3, we show that
  $\pderiv{^2 \xi_3}{\alpha^2} > 0$, 
  for all $d\geq 23$, $\kappa\in [0,1]$ and 
  $\alpha \in (1,53/27]$, thus  
    completing the proof.
   
     \noindent{\bf Technical details of Step~4.}  
We start with the  observation that, for any $k_1$ and $k_2$ (not depending on $\alpha$) and any function $r$ of $\alpha$, it holds that
$$\pderiv{(r X^{k_1} Y^{k_2})}{\alpha} = f(r,k_1,k_2) X^{k_1} Y^{k_2}, \mbox{ where } f(r,k_1,k_2) = \pderiv{r}{\alpha} + \frac{k_1 r}{X} \pderiv{X}{\alpha} + \frac{k_2 r}{Y}  \pderiv{Y}{\alpha}.$$ 
Applying this observation twice to each of the two summands in the expression \eqref{eq:rt4tv4ybyb} for $\xi_3$, 
we see that there   are rational functions~$s_2$ and~$t_2$  
 of $\alpha$, $d$ and $\kappa$
 such that 
 \[\pderiv{^2 \xi_3}{\alpha^2} = s_2 X^{d\kappa }Y^{d(1-\kappa)} + t_2 X^{\frac{2\kappa(d+1)}{3}} Y^{\frac{(1-\kappa)(d+1)}{3}}.\]

 In Section~\ref{app:d>=23techlem} we use Mathematica to calculate $t_2$ explicitly
 and to verify that for every 
 $d\geq 23$ and $\alpha \in (1,53/27]$, we have
 $ \pderiv{^2t_2}{\kappa^2}\geq0$,
 $ \pderiv{t_2}{\kappa}\Big\vert_{\kappa=0}\geq 0$
 and  $ t_2\vert_{\kappa=0}\geq 0 $.
 We conclude that $t_2\geq 0$ for all $\kappa \in [0,1]$
 (for the given ranges of $d$ and $\alpha$). 
 Since  $X$ and $Y$ are less than or equal to~$1$
 and $2 \kappa(d+1)/3 \leq d \kappa$  
 and  $(1-\kappa)(d+1)/3 \leq d(1-\kappa)$, the fact that $t_2\geq 0$ implies

    \begin{equation}\label{eq:xi3derivative}
    \pderiv{^2\xi_3}{\alpha^2} = s_2 X^{d\kappa} Y^{d(1-\kappa)} + t_2 X^{\frac{2\kappa(d+1)}{3}} Y^{\frac{(1-\kappa)(d+1)}{3}} \geqslant (s_2 + t_2) X^{d\kappa} Y^{d(1-\kappa)}.
    \end{equation}
    
      The final Mathematica code segment verifies that $s_2+t_2>0$
  for all $d\geq 23$, $\alpha \in (1,53/27]$ and $\kappa \in [0,1]$.     
  This, together with \eqref{eq:xi3derivative}, 
  completes the proof since $X$ and $Y$ are positive.
 \end{proof}

We finish by giving the proof of Lemma~\ref{lem:twosteptoprove}.
\begin{lemtwosteptoprove}
\statelemtwosteptoprove
\end{lemtwosteptoprove}
\begin{proof}
To prove Condition~\ref{cond:we}, it suffices to check \eqref{eq:2ed2g2v2vtvt} for $\alpha\in (1,53/27)$. By Lemma~\ref{lem:hbound}, we only need to check that $\varphi_*(d, d_0,\alpha)<\alpha^{1/d}$ for $\alpha\in (1,53/27)$ and integers $0\leq d_0\leq d$. This has been verified in Lemma~\ref{lem:d<=22techlem} for all $d\leq 22$ and in Lemma~\ref{lem:d>=23techlem} for all $d\geq 23$.
\end{proof}

\section{Appendix: Mathematica Code}\label{sec:code}

\subsection{Lemma~\ref{lem:examplecond}}\label{app:examplecond}
The following code checks that, for all $d$-tuples $(\pb^{(1)},\hdots, \pb^{(d)})$ with  
$\pb^{(1)},\hdots, \pb^{(d)}\in\mathrm{Ex}_q(\alpha)$, it holds that 
$h_{c_1,c_2,\beta_*}(\pb^{(1)},\hdots, \pb^{(d)})<\alpha^{1/d}$.
The output is {\tt True}, and the same is true when the first line changes to $q=4,d=4$. 
\begin{verbatim}
q = 3; d = 3; b = 1 - q/(d + 1);
EX = Tuples[{1, alpha}, q-1]; 
G[vector_, colour_] := 1 - (1-b) (vector[[colour]]- 1)/
                    (b + Sum[vector[[cc]], {cc,1,q-1}]);
dTUPLES=Tuples[EX,d]; L=Length[dTUPLES];
UNIQ = True;
For[l = 1, l<= L && UNIQ, l++,
   currenttuple=dTUPLES[[l]];
   For[k=1, k <= d, k++,
         vectorofchild[k]=currenttuple[[k]];
   ];
   For[colour = 1, colour <= q-1, colour++,
        g[colour] = Product[  G[vectorofchild[k], colour], {k, 1, d}];
   ];
   h = 1 + (1 - b) (1 - g[1])/(b + Sum[g[c],{c,1,q-1}])/.{alpha->u^d};
   CHECK = Resolve[Exists[u, h >= u && u > 1]];
   If[CHECK == True, UNIQ = False];
];
Print[UNIQ]
\end{verbatim}

\subsection{Lemma~\ref{lem:fbetamono}} \label{app:fbetamono}

Both of the queries in the following code give the output {\tt True}.

\begin{verbatim}
fu = (1 - (1 - b) y)^d /
((1 - (1 - b) y)^d + 2 (1 - (1 - b) x)^(d/2) (1 - (1 - b) (1 - x - y))^(d/2));
bb = 1-b;
W = 1 - 3y + bb ((3y - 1) + (1 - x - 2y) + x(2x + y - 1) + y(x - y));
lhs = D[fu, b];
rhs = -fu^2 (d (1 - bb x)^(d/2) (1- bb(1-x-y))^(d/2) W /
((1- bb y)^(d+1) (1- bb x) (1- bb (1-x-y))));
FullSimplify[lhs == rhs]

fl = (1 - (1 - b) x)^d /
((1 - (1 - b) x)^d  + (1 - (1 - b) y)^d + (1 - (1 - b) (1 - x - y))^d);
lhs = D[fl, b];
rhs = fl^2 d ((2x+y-1) (1-bb(1-x-y))^(d-1) + (x-y) (1-bb y)^(d - 1)) /
((1-bb x)^(d+1));
FullSimplify[lhs == rhs]
\end{verbatim}

\subsection{Lemma~\ref{eq:d=2fixedpoints}}\label{app:d=2bound}

Both of the queries in the following code give the output {\tt False}.

\begin{verbatim}
fu = (1 - (1 - b) y)^d /
((1 - (1 - b) y)^d + 2 (1 - (1 - b) x)^(d/2) (1 - (1 - b) (1 - x - y))^(d/2));
fl = (1 - (1 - b) x)^d /
((1 - (1 - b) x)^d  + (1 - (1 - b) y)^d + (1 - (1 - b) (1 - x - y))^d);
Resolve[Exists[{x, y, b}, 0 < y <= 1/3 && 1106/2500 <= x < 1 && 
0 < b <= 1 && (fu /. {d -> 2}) == x && (fl /. {d -> 2}) == y]]
Resolve[Exists[{x, y, b}, 0 < y <= 460/2000 && 1/3 <= x < 1 && 
0 < b <= 1 && (fu /. {d -> 2}) == x && (fl /. {d -> 2}) == y]]
\end{verbatim}

\subsection{Lemma~\ref{lem:guymono}}\label{app:guymono}

The following code gives the output {\tt True}.

\begin{verbatim}
fu = (1 - (1 - b) y)^d /
((1 - (1 - b) y)^d + 2 (1 - (1 - b) x)^(d/2) (1 - (1 - b) (1 - x - y))^(d/2));
gu = fu - x /. {b -> 1 - 3/(d + 1), x -> m y};
A =  (1 - 3 y/(d + 1))^d;
B = (1- 3(1-y(m+1))/(d+1))^(d/2) (1- 3m y/(d+1))^(d/2);
W = A B/(A+2B)^2;   
lhs = D[gu, y];
rhs = -m+(9 d W /(3y(m+1)+d-2)) (m(2m y+y-1)/(1+d-3m y) - (2m y+y+d-1)/(1+d-3y));
FullSimplify[lhs == rhs]
\end{verbatim}

\subsection{Lemma~\ref{lem:glymono}}\label{app:glymono}

The following code gives the output {\tt True}.

\begin{verbatim}
fl = (1 - (1 - b) x)^d /
((1 - (1 - b) x)^d  + (1 - (1 - b) y)^d + (1 - (1 - b) (1 - x - y))^d);
gl = fl - y /. {b -> 1 - 3/(d + 1), x -> m y};
W = (m(2d-1)+d+1) ((3y(m+1)+d-2) / (d+1))^d / 
((1 + d - 3m y) (d-2+ 3y(1 + m))) +
(m-1)(d+1)(1-(3 y)/(d+1))^d / ((1 + d - 3 y) (1 + d - 3m y));
lhs = D[gl, y];
rhs = -3 d (1 - (3m y)/(d+1))^d W / ( ((3y(m+1)+d-2)/(d+1))^d 
    + (1 - (3m y)/(d+1))^d + (1-(3 y)/(d+1))^d )^2 - 1;
FullSimplify[lhs == rhs]  
\end{verbatim}

\subsection{Lemma~\ref{lem:ineqs}}\label{sec:ineqs}

The following code outputs {\tt False} and {\tt False}.

\begin{verbatim}
y2 = 7/(10 m + 12);
y1 = y2 + 3/500;
x2 = m y2;
x1 = m y1;
Resolve[Exists[m,  157/80 <= m < 32 && (0 >= y1 || y1 >= 1 - x1 - y1 || 
1 - x1 - y1 >= 1/3 ||  1/3 >= x1 || x1 >= 1 - y1 ) ]]
Resolve[Exists[m, 32 <= m && (0 >= y2 || y2 >= 1 - x2 - y2 || 
1 - x2 - y2 >= 1/3 ||  1/3 >= x2 || x2 >= 1 - y2 ) ]]
\end{verbatim}

\subsection{Lemma~\ref{lem:dhu<0}}\label{app:lem:dhu<0}

Both queries in the following code for the differentiation in \eqref{mypartial} output {\tt True}. 
 
\begin{verbatim}
fu = (1 - (1 - b) y)^d /
((1 - (1 - b) y)^d + 2 (1 - (1 - b) x)^(d/2) (1 - (1 - b) (1 - x - y))^(d/2));
bb = 1 - b;
R = (1 - bb x)^(d/2) (1 - bb (1 - x - y))^(d/2) / (1 - bb y)^d;
xlhs = D[fu, x];
xrhs = fu^2 R d bb^2 (2 x + y - 1) / ((1 - bb(1 - x - y))(1 - bb x));
ylhs = D[fu, y];
yrhs = -fu^2 R d bb (3 + bb(2 x + y - 2)) /((1 - bb(1 - x - y))(1 - bb y));
FullSimplify[xlhs == xrhs]
FullSimplify[ylhs == yrhs]
\end{verbatim}

\noindent The proof of \eqref{eq:err346gevrerf53}. 
We consider two cases --- when $\mu<32$ and when $\mu\geq 32$. The output is {\tt False} in both cases.
\begin{verbatim}
lhs=(8-y)(2x+y-1)/( (8-x)(2x+y+22) );
x1 = 7 m/(10 m + 12) + 3 m/500;
y1 = 7/(10 m + 12) + 3/500;
Resolve[Exists[m, 1<m <32 && (lhs/.{x->x1, y->y1}) >= 1/24]]
x2 = 7 m/(10 m + 12);
y2 = 7/(10 m + 12);
Resolve[Exists[m, m>=32 && (lhs/.{x->x2, y->y2}) >= 1/24]]
\end{verbatim}

\noindent Here is the code to show that  $Y$ is increasing  in~$\hat\beta_*$. The output is {\tt False.}
\begin{verbatim}
Num = 3 + b (2 x + y - 2);
Den = (1 - b (1 - x - y)) (1 - b y);
Der = D[Num/Den  , b];
Resolve[Exists[{d, x, y, b}, 
  d >= 0 && 0 < b <= 1/8 && 0 < y < 1 - x - y < 1/3 < x < 1 - y && 
   Der < 0]]
\end{verbatim}

Here is the code for Case 1. The output is {\tt False}.

\begin{verbatim} 
lhs = 3 (2 x + y + 22)/((8 - y) (x + y + 7)) /. 
        {x -> 7 m/(10 m + 12), y -> 7/(10 m + 12)};
rhs = 8/7;
Resolve[Exists[m, lhs >= rhs && m > 32]]
\end{verbatim}

Here is the code for Case 2. The output is {\tt False}.
\begin{verbatim}
lhs = 3 (2 x + y + 22)/((8 - y) (x + y + 7)) /. 
   {x -> 7 m/(10 m + 12) + 3 m/500, y -> 7/(10 m + 12) + 3/500};
rhs = 24 (25 m^2 + 60 m + 3536)/(25 m^2 + 60 m + 73536);
Resolve[Exists[m, lhs >= rhs && 1 < m]]
\end{verbatim}

\subsection{Lemma~\ref{lem:muone}}\label{app:muone}

The output of the following code, for the differentiation in \eqref{eq:Dhud}, is  {\tt True.}
\begin{verbatim}
fu = (1 - (1 - b) y)^d /
((1 - (1 - b) y)^d + 2 (1 - (1 - b) x)^(d/2) (1 - (1 - b) (1 - x - y))^(d/2));
psi[d_,z_] := d/(d-3z+1) + Log[d +1-3z];
zeta[d_,x_,y_]:=2psi[d,y]-psi[d,x]-psi[d,1-x-y];
lhs = D[ (fu/.{b -> 1 - 3/(d + 1)}), d];
rhs = (1- 3x/(d+1))^(d/2) (1- 3y/(d+1))^d (1- 3(1-x-y)/(d+1))^(d/2)zeta[d,x,y]/
			  ((1- 3y/(d+1))^d + 2(1- 3x/(d+1))^(d/2) (1- 3 (1-x-y)/(d+1))^(d/2))^2;
FullSimplify[lhs == rhs,d>=0]
\end{verbatim}

\noindent The following code establishes Facts~1, 2, and 3 for $\mu=157/80$.
The output is {\tt False}, $0$, then {\tt True.}
\begin{verbatim}
fu = (1 - (1 - b) y)^d /
((1 - (1 - b) y)^d + 2 (1 - (1 - b) x)^(d/2) (1 - (1 - b) (1 - x - y))^(d/2));
hu=fu-x /.{b -> 1 - 3/(d+1)};
psi[d_,z_] := d/(d-3z+1) + Log[d +1-3z];
zeta[d_,x_,y_]:=2psi[d,y]-psi[d,x]-psi[d,1-x-y];
Fd = D[zeta[d,x,y], d];
xm = 7 m/(10 m + 12) + 3 m/500 /. {m -> 157/80};
ym = 7/(10 m + 12) + 3/500 /. {m -> 157/80};
Resolve[Exists[d, d >= 23 && (Fd /. {x -> xm, y -> ym}) >= 0]]
Limit[zeta[d,xm,ym], d -> \[Infinity]]
Limit[hu /. {x -> xm, y -> ym}, d -> \[Infinity]]<0 
\end{verbatim}

\noindent The following code establishes Facts~1, 2, and 3 for $\mu=32$.
The output is {\tt False}, $0$, then {\tt True.}
\begin{verbatim}
fu = (1 - (1 - b) y)^d /
((1 - (1 - b) y)^d + 2 (1 - (1 - b) x)^(d/2) (1 - (1 - b) (1 - x - y))^(d/2));
hu=fu-x /.{b -> 1 - 3/(d+1)};
psi[d_,z_] := d/(d-3z+1) + Log[d +1-3z];
zeta[d_,x_,y_]:=2psi[d,y]-psi[d,x]-psi[d,1-x-y];
Fd = D[zeta[d,x,y], d];
xm = 7 m/(10 m + 12) /. {m -> 32};
ym= 7/(10 m + 12) /. {m -> 32}; 
Resolve[Exists[d, d >= 23 && (Fd /. {x -> xm, y -> ym}) >= 0]]
Limit[zeta[d,xm,ym], d -> \[Infinity]] 
Limit[hu /. {x -> xm, y -> ym}, d -> \[Infinity]]<0
\end{verbatim}

\subsection{Lemma~\ref{lem:hl>0}}\label{app:hl>0}

We first prove that $h_\ell(23,\mu)>0$ for $\mu\geq 157/80$. The output to both
queries is {\tt False.}
 
\begin{verbatim}
fl = (1 - (1 - b) x)^d /
((1 - (1 - b) x)^d  + (1 - (1 - b) y)^d + (1 - (1 - b) (1 - x - y))^d);
hl = fl - y /. {b -> 1 - 3/(d+1), x -> m y};
x1 = 7 m/(10 m + 12) + 3 m/500;
y1 = 7/(10 m + 12) + 3/500;
h1 = hl /. {d -> 23, x -> x1, y -> y1};
Resolve[Exists[m, 157/80<=m <32 && h1 <= 0]]
x2 = 7 m/(10 m + 12);
y2 = 7/(10 m + 12);
h2 = hl /. {d -> 23, x -> x2, y -> y2};
Resolve[Exists[m, m >= 32 && h2 <= 0]]
\end{verbatim}

\noindent The output of the following code, for the differentiation in~\eqref{eq:Dhld}, is  {\tt True.}
\begin{verbatim}
fl = (1 - (1 - b) x)^d/
((1 - (1 - b) x)^d + (1 - (1 - b) y)^d + (1 - (1 - b) (1 - x - y))^d); 
psi[d_, z_] := d/(d - 3 z + 1) + Log[d + 1 - 3 z];
A = (1 - 3 x/(d + 1))^d;
B = (1 - 3 y/(d + 1))^d;
CC = (1 - 3 (1 - x - y)/(d + 1))^d;
lhs = D[(fl /. {b -> 1 - 3/(d + 1)}), d];
rhs = (A CC (psi[d, x] - psi[d, 1 - x - y]) + A B (psi[d, x] - psi[d, y]))/
   (A + B + CC)^2;
FullSimplify[lhs == rhs, d >= 0]
\end{verbatim}

\subsection{Lemma~\ref{lem:d<=22sequence}}\label{app:d<=22sequence}
The code checks that all of the desired inequalities are satisfied. The output is {\tt True.}
\begin{verbatim}
fu = (1 - (1 - b) y)^d /
((1 - (1 - b) y)^d + 2 (1 - (1 - b) x)^(d/2) (1 - (1 - b) (1 - x - y))^(d/2));
fl = (1 - (1 - b) x)^d /
((1 - (1 - b) x)^d  + (1 - (1 - b) y)^d + (1 - (1 - b) (1 - x - y))^d);
Flag = True;
u[0] = 1;
l[0] = 0;
For[dd = 3, dd <= 22, dd++, (
   ffu = Ceiling[10000 fu /. {d -> dd, b -> 1 - 3/(dd + 1)}]/10000;
   ffl = Floor[10000 fl /. {d -> dd, b -> 1 - 3/(dd + 1)}]/10000;
   For[n = 0, n <= 60, n++, (
     u[n + 1] = ffu /. {x -> u[n], y -> l[n]};
     l[n + 1] = ffl /. {x -> u[n], y -> l[n]};
     Flag = Flag && u[n] >= u[n + 1] && l[n] <= l[n + 1] &&  
       2 u[n] + l[n] >= 1 >= 2 l[n] + u[n])];
   Flag = Flag && u[60]/l[60] <= 53/27;)];
Flag 
\end{verbatim}

\subsection{Lemma~\ref{lem:d<=22techlem}}\label{app:d<=22techlem}

The code checks all relevant values of~$d$ and~$d_0$. 
The output is {\tt True.}
The substitution of $u$ for $\alpha^{1/d}$ is there to make the code run faster.
Despite this, it takes more than 5 minutes to run on our machine.

\begin{verbatim}
b = 1 - 3/(d + 1);
x = 1 - (1 - b) (a - 1)/(b + 2 a);
y = 1 - (1 - b) (a - 1)/(b + a + 1);
phi = 1 + (3/(d + 1)) (1 - x^d0 y^(d - d0))/(b + 
        x^d0 (y^(d - d0) + 1)) /. {a -> u^d};
Flag = True;
For[dd = 2, dd <= 22, dd++, u0 = (53/27)^dd;
  For[dd0 = 0, dd0 <= dd, dd0++,
   Flag =  Flag && ! Resolve[Exists[u, 
      (phi /. {d -> dd, d0 -> dd0}) >= u && 1 < u <= u0]];];];
Flag
\end{verbatim}

\subsection{Lemma~\ref{lem:d>=23techlem}}\label{app:d>=23techlem}
The following code outputs {\tt True}, therefore verifying the differentiation in  \eqref{eq:vetvt4g66}.
\begin{verbatim}
b = 1 - 3/(d + 1);
x = 1 - (1 - b) (a - 1)/(b + 2 a);
y = 1 - (1 - b) (a - 1)/(b + a + 1);
phi = 1+(1/d)( x^(-2d0/(2+b)) y^(-(d-d0)/(2+b)) 
                   - x^(d0 b/(2+b)) y^((d-d0)(1+b)/(2+b)) )-a^(1/d);
gx=(2a+b)(a b+a+1);
gy=(a+b+1)(a b+2);
S=x^(2d0/(2+b)) y^((d-d0)/(2+b)) (a+b+1) (a b+2)/(1-b);
rhs=( (1-b)x^(-2d0/(2+b)) y^(-(d-d0)/(2+b)) / (d a gy) ) *
    (2 d0 a (gy/gx) + (d-d0)a + x^d0 y^(d-d0) (d0 b a (gy/gx)+ (d-d0)(1+b)a)
        -a^(1/d) S);
Resolve[Simplify[D[phi, a] - rhs] == 0]
\end{verbatim}

\noindent The following code outputs {\tt True} {\tt True}, verifying the calculation for Step 3.
\begin{verbatim}
b = 1 - 3/(d + 1);
x = 1 - (1 - b) (a - 1)/(b + 2 a);
y = 1 - (1 - b) (a - 1)/(b + a + 1); 
s0 = -(((2 d - 1) d (1 - k) a)/(d + 
        1) + ((d - 2) k ((d - 2) a + 2 (d + 1)))  /(3 (d + 1)));
W = 1 + (a - 1)/d - (d - 1) (a - 1)^2/(2 d^2);
t0 = W ((d + 1) a + 2 d - 1) ((d - 2) a + 2 (d + 1))/(3 (d + 1));
xi3 = s0 x^(d k) y^(d (1 - k)) + 
   t0 x^(2 k (d + 1)/3) y^((1 - k) (d + 1)/3);
FullSimplify[(xi3 /. {a -> 1}) == d + d k]
FullSimplify[(D[xi3, a] /. {a -> 1}) == (1/3) (- d k + 3 d - 4 k) ]
\end{verbatim}

\noindent The final two code segments are for Step~4.
The following code  calculates $t_2$ and verifies  
that $ \pderiv{^2t_2}{\kappa^2}\geq0$,
 $ \pderiv{t_2}{\kappa}\Big\vert_{\kappa=0}\geq 0$
 and  $ t_2\vert_{\kappa=0}\geq 0 $. The output is {\tt False}, {\tt False}, and {\tt False.}
 
\begin{verbatim}
b = 1 - 3/(d + 1);
x = 1 - (1 - b) (a - 1)/(b + 2 a);
y = 1 - (1 - b) (a - 1)/(b + a + 1);
s0 = -(((2 d - 1) d (1 - k) a)/(d + 
        1) + ((d - 2) k ((d - 2) a + 2 (d + 1)))/(3 (d + 1)));
W = 1 + (a - 1)/d - (d - 1) (a - 1)^2/(2 d^2);
t0 = W ((d + 1) a + 2 d - 1) ((d - 2) a + 2 (d + 1))/(3 (d + 1));
 
t2 = Simplify[ D[t0 x^(2 k (d + 1)/3) y^((1 - k) (d + 1)/3), {a, 
      2}]/(x^(2 k (d + 1)/3) y^((1 - k) (d + 1)/3))];
      
tk1 = D[t2, k];
tk2 = D[tk1, k];
Resolve[Exists[{d, a}, tk2 < 0  && d >= 23 && 1 <= a <= 53/27]]
Resolve[Exists[{d, a}, (tk1 /. {k -> 0}) < 0 && d >= 23 && 1 <= a <= 53/27]]
Resolve[Exists[{d, a}, (t2 /. {k -> 0}) < 0  && d >= 23 && 1 <= a <= 53/27]] \end{verbatim}

\noindent The following code calculates $s_2$ (as well as $t_2$)
and verifies  that $s_2+t_2$ is positive. It takes  about 10 minutes to run. 
The output is {\tt False.} The reason for the transformation of $\alpha$ (as a function of $r$) is to speed up the calculation.

\begin{verbatim}
b = 1 - 3/(d + 1);
x = 1 - (1 - b) (a - 1)/(b + 2 a);
y = 1 - (1 - b) (a - 1)/(b + a + 1);
s0 = -(((2 d - 1) d (1 - k) a)/(d + 
        1) + ((d - 2) k ((d - 2) a + 2 (d + 1)))/(3 (d + 1)));
W = 1 + (a - 1)/d - (d - 1) (a - 1)^2/(2 d^2);
t0 = W ((d + 1) a + 2 d - 1) ((d - 2) a + 2 (d + 1))/(3 (d + 1));
 
t2 = Simplify[ D[t0 x^(2 k (d + 1)/3) y^((1 - k) (d + 1)/3), {a, 
      2}]/(x^(2 k (d + 1)/3) y^((1 - k) (d + 1)/3))];
        
 s2 = Simplify[
   D[s0 x^(d k) y^(d (1 - k)), {a, 2}]/(x^(d k) y^(d (1 - k)))];     
p = Simplify[s2 + t2 /. {a -> 1 + 3 r}];
Resolve[Exists[{d, r, k}, 
  p <= 0 && d >= 23 && 0 <= k <= 1 && 0 < r <= 26/81]]    

\end{verbatim}

\bibliographystyle{plain} \bibliography{\jobname}

\end{document}